\documentclass[runningheads]{llncs}

\usepackage{graphicx}
\usepackage{xspace}
\usepackage{nicefrac}
\usepackage{tikz}
	\usetikzlibrary{arrows}
\usepackage{listings}
\usepackage{booktabs}
\usepackage{url}
\usepackage[colorlinks, breaklinks]{hyperref}
\usepackage{amsmath,amssymb}
\usepackage{stmaryrd}
\usepackage[left=2.25cm,right=2.25cm,top=3cm,bottom=3cm]{geometry}

\usepackage{MPSTtoSequential}

\allowdisplaybreaks

\begin{document}


\title{Taming Concurrency for Verification\\ Using Multiparty Session Types (Technical Report)}
\titlerunning{Taming Concurrency for Verification using \MPST (Tech.\ Rep.)}


\author{Kirstin Peters \and Christoph Wagner \and\\ Uwe Nestmann}
\institute{TU Berlin/TU Darmstadt, Germany}

\authorrunning{K. Peters, C. Wagner, and U. Nestmann}

\maketitle


\begin{abstract}
	The additional complexity caused by concurrently communicating processes in distributed systems render the verification of such systems into a very hard problem.
	Multiparty session types were developed to govern communication and concurrency in distributed systems.
	As such, they provide an efficient verification method \wrt properties about communication and concurrency, like communication safety or progress.
	However, they do not support the analysis of properties that require the consideration of concrete runs or concrete values of variables.

	We sequentialise well-typed systems of processes guided by the structure of their global type to obtain interaction-free abstractions thereof.
	Without interaction, concurrency in the system is reduced to sequential and completely independent parallel compositions.
	In such abstractions, the verification of properties such as \eg data-based termination that are not covered by multiparty session types, but rely on concrete runs or values of variables, becomes significantly more efficient.
	
	This technical report provides proofs and additional material for the paper \cite{petersWagnerNestmann19}.
\end{abstract}


\section{Introduction}

Modern society is increasingly dependent on large-scale software systems that are distributed, collaborative, and communication-centred.
One of the techniques developed to handle the additional complexity caused by distributed actors are \emph{multiparty session types} (\MPST) \cite{hondaYoshidaCarbone08}. \MPST allow to specify the desired behaviour of communication protocols as by-design correct types that are used to verify the communication structure of software products. The properties guaranteed by well-typed processes cover communication safety (all processes conform to globally agreed communication protocols) and liveness properties such as deadlock-freedom. Their main advantage is that their verification method is extremely efficient---in comparison to \eg standard model checking.

\MPST were developed to govern communication and concurrency in distributed systems.
However, as it is typical for type systems, standard \MPST variants (without dependable types) do not support the analysis of properties that require the consideration of concrete runs or concrete values of variables.

The hardest part about the verification of distributed systems is the state space explosion that results from concurrent communication attempts, \ie the exponential blow-up that results from computing all possible combinations of potential communication partners.
The problem of concurrency mainly lies in the communication structure, which is already completely captured by \MPST.
We show that the knowledge of a program/system to be well-typed, allows us to sequentialise it following the structure of its global type and thereby to remove all communication.
Accordingly, we show how we can benefit from the effort we spend on an \MPST analysis of a system also for the verification of its properties that go beyond its communication structure.

We use the global type of a well-typed system to guide its sequentialisation.
We refer to the result as \emph{sequential global process} (\SGP), although it might still contain parallel compositions, albeit only on completely independent parts.
Since the structure of communication was already verified by the well-typedness proof, we can reduce communication to value updates.
More precisely, we map well-typed systems that interact concurrently, to \SGP-systems without any interaction mechanisms or name binders.
Such \SGP-systems consist of a vector of variables with values and a \SGP-process that simulates the data flow of the original system.
Therefore, we translate the reception of data in communication into updates of the vector in the \SGP-system.
By removing the communication we remove also the problem of state space explosion.
Our translation is valid if the considered process is well-typed \wrt a (set of) global type(s).
Thereby, we sequentialise communications that may happen concurrently in the original system but are sequential in global types.
Note that such communications are always causally independent of each other, thus ordering them does not significantly influence the behaviour of the system, \eg it does not influence what values are computed.
Apart from such sequentialisations the original system and its abstraction into a \SGP-system behave similarly.

\smallskip
\textbf{Contributions.}
We provide an algorithm to remove communication from well-typed systems and thereby sequentialise them, while preserving the evolution of data of the original system.
Deriving this algorithm was technically challenging but the result is a simple rewriting function and easy to automate.

Then we prove that, provided that the original system was well-typed, the algorithm produces a \SGP-system that is closely related to the original system:
the original system and its abstraction are related by a variant of operational correspondence \cite{gorla10} and are coupled similar \cite{parrowCoupled92}.
With that, the derived \SGP-system is a good abstraction of the original system that can be used instead of the original to verify properties on concrete data.
Since the mapping into \SGP-systems is usually linear and because \SGP-systems do not contain any form of interaction or binders, properties can be checked more efficiently.

Finally, we provide a mapping---that is again a simple rewriting algorithm---from \SGP-processes into \textsf{Promela}, the input language of the model checker \textsf{Spin} \cite{spin,spin91}.
With that, the properties that are not already guaranteed by the \MPST analysis but require the consideration of concrete runs or concrete data can be checked.
Since the main challenge here is the sequentialisation of concurrent systems into interaction-free abstractions, the translation of \SGP-systems into \textsf{Promela} is simple and can be used as a role model to obtain similar mappings for other model checkers.

\smallskip
\textbf{Overview.}
In Section~\ref{sec:AppendixMPST} we extend Section~2 of \cite{petersWagnerNestmann19} and introduce multiparty session types (including the things that are missing in \cite{petersWagnerNestmann19} such as local types, projection, and typing rules).
Section~\ref{sec:basicProperties} proves the basic properties of the introduced type system.
Section~3 of \cite{petersWagnerNestmann19} introduces \SGP-systems and a mapping that translates well-typed systems into \SGP-systems.
In Section~\ref{sec:proofs} we prove the relations between the original systems and their abstractions into \SGP-systems as they are described in Section~4 of \cite{petersWagnerNestmann19}.
Then, Section~5 of \cite{petersWagnerNestmann19} illustrates how the sequentialisation can be used to verify properties of the original system.
Section~\ref{sec:examples} introduces some small examples to illustrate this method.


\section{Multiparty Session Types}
\label{sec:AppendixMPST}

In the following we extend Section~2 of \cite{petersWagnerNestmann19}. In particular, we introduce some additional concepts of multiparty session types such as local types, derive the notion of well-typedness, and show some standard properties.

Multiparty session types describe global behaviours as \emph{sessions}, \ie units of conversations. The participants of such sessions are called \emph{roles}. \emph{Global types} specify protocols from a global point of view, whereas \emph{local types} describe the behaviour of individual roles within a protocol. \emph{Projection} ensures that a global type and its local types are consistent. These types are used to reason about processes formulated in a \emph{session calculus}. Most of the existing session calculi are extensions of the well-known $ \pi $-calculus \cite{milnerParrowWalker92} with specific operators adapted to correlate with local types.

Similar to \cite{synchMPST,yoshida10}, we assume that roles, \ie the identifiers for participants, are natural numbers.
Assume a countably infinite set of names. Names are used to denote channels and variables that may stand for a channel or some value. In the session calculus we distinguish \emph{shared channels} that are used outside of sessions (to initialise sessions) and \emph{session channels} that are used within sessions.

\subsection{Global Types, Local Types, and Projection}
\label{sec:types}

Global types describe protocols from a global point of view on systems by interactions between roles. They are used to formalise specifications that describe the desired properties of a system.
We inherit the definition of global types from \cite{synchMPST}, but unify the transmission of values and branching into a single construct as done in \cite{DemangeonHonda12}.

\begin{definition}[Global Types]
	\label{def:globalTypes}
	The \emph{global types} are given by
	\begin{align*}
		\GT &\deffTerms
		\GCom{\Role_1}{\Role_2}{\Set{ \GLab{\Label_i}{\tilde{\Sort}_i}{\GT_i} }_{i \in \indexSet}}
		\sepTerms \GPar{\GT_1}{\GT_2}
		\sepTerms \GRec{\TVar}{\GT} \sepTerms \TVar
		\sepTerms \GEnd
	\end{align*}
	where $ \Role_1, \Role_2 $ are roles, $ \Label_i $ are labels, $ \tilde{\Sort}_i $ are sequences of sorts, $ \indexSet $ are non-empty finite index sets, and $ \TVar $ are type variables.
\end{definition}

The global type $ \GCom{\Role_1}{\Role_2}{\Set{ \GLab{\Label_i}{\tilde{\Sort}}{\GT_i} }_{i \in \indexSet}} $ specifies a communication from role $ \Role_1 $ to $ \Role_2 $, where $ \Role_1 $ picks a label $ \Label_i $, \ie one of the indexed set of options, transmits values of the sorts $ \tilde{\Sort}_i $ and then the type proceeds with $ \GT_i $.
The parallel composition $ \GPar{\GT_1}{\GT_2} $ allows to combine two independent global types $ \GT_1 $ and $ \GT_2 $, where independence means that these two global types do not share roles.
The operators $ \GRec{\TVar}{\GT} $ and $ \TVar $ introduce recursion,
whereas successful termination of a global type is specified by $ \GEnd $.

Let $ \Roles{\cdot} $ return the roles used in a global type (or a process as introduced later).

Global types describe systems from a global point of view.
To link them with the local points of view of processes they are projected onto their roles to obtain local types.
Again, we use the local types of \cite{synchMPST}, where we combine communication and branching into single constructs for the sender and the receiver as done in \cite{DemangeonHonda12}.

\begin{definition}[Local Types]
	\label{def:localTypes}
	The \emph{local types} are given by
	\begin{align*}
		\LT &\deffTerms
		\LSend{\Role}{\Set{ \LLab{\Label_i}{\tilde{\Sort}_i}{\LT_i} }_{i \in \indexSet}}
		\sepTerms \LGet{\Role}{\Set{ \LLab{\Label_i}{\tilde{\Sort}_i}{\LT_i} }_{i \in \indexSet}}
		\sepTerms \LRec{\TVar}{\LT} \sepTerms \TVar
		\sepTerms \LEnd
	\end{align*}
	where \Role are roles, $ \Label_i $ are labels, $ \tilde{\Sort}_i $ are sequences of sorts, \indexSet are non-empty finite index sets, and $ \TVar $ are type variables.
\end{definition}

The two local end points of communication are the types $ \LSend{\Role}{\Set{ \LLab{\Label_i}{\tilde{\Sort}_i}{\LT_i} }_{i \in \indexSet}} $ for the sender, where the role~\Role indicates the receiver, and $ \LGet{\Role}{\Set{ \LLab{\Label_i}{\tilde{\Sort}_i}{\LT_i} }_{i \in \indexSet}} $ for the receiver, where the role~\Role indicates the sender.
Recursion with the constructs $ \LRec{\TVar}{\LT} $ and $ \TVar $ and successful termination represented by $ \LEnd $ are similar to global types.

The partial mapping from global types onto their roles is called \emph{projection}.
It is undefined for parallel global types that share a role and communications that branch such that roles that are neither the sender nor the receiver have to behave differently.
The first case reflects that parallel composition on global types defines independence, \ie parallel global types specify the behaviour of partitions of distributed systems that do not interact.
The latter case ensures that if a process---the sender of this communication---decides to branch then only processes that are informed about this decision can adapt their behaviour accordingly.
If for a global type \GT projection is defined for all its roles then we call this type \emph{projectable}.

\begin{definition}[Projection]
	\label{def:projection}
	\emph{Projection} of a global type \GT onto a role \Role[p], written as $ \Proj{\GT}{\Role[p]} $ is defined as:
	\[ \begin{array}{c}
		\Proj{\left( \GCom{\Role_1}{\Role_2}{\Set{ \GLab{\Label_i}{\tilde{\Sort}_i}{\GT_i} }_{i \in \indexSet}} \right)}{\Role[p]} =
			\begin{cases}
				\LSend{\Role_2}{\Set{ \LLab{\Label_i}{\tilde{\Sort}_i}{\left( \Proj{\GT_i}{\Role[p]} \right)} }_{i \in \indexSet}} & \text{if } \Role[p] = \Role_1 \neq \Role_2\\
				\LGet{\Role_1}{\Set{ \LLab{\Label_i}{\tilde{\Sort}_i}{\left( \Proj{\GT_i}{\Role[p]} \right)} }_{i \in \indexSet}} & \text{if } \Role[p] = \Role_2 \neq \Role_1\\
				\bigsqcup_{\Role[p], i \in \indexSet} \left( \Proj{\GT_i}{\Role[p]} \right) & \text{if } \Role[p] \notin \Set{ \Role_1, \Role_2 }
			\end{cases}\\
		\Proj{\left( \GPar{\GT_1}{\GT_2} \right)}{\Role[p]} =
			\begin{cases}
				\Proj{\GT_i}{\Role[p]} & \text{if } \Role[p] \in \GT_i \text{ and } \Role[p] \notin \GT_j, i \neq j \in \Set{ 1, 2 }\\
				\LEnd & \text{if } \Role[p] \notin \GT_1 \text{ and } \Role[p] \notin \GT_2
			\end{cases}\\
		\Proj{\left( \GRec{\TVar}{\GT} \right)}{\Role[p]} = \LRec{\TVar}{\left( \Proj{\GT}{\Role[p]} \right)}
		\hspace{2em}
		\Proj{\TVar}{\Role[p]} = \TVar
		\hspace{2em}
		\Proj{\GEnd}{\Role[p]} = \LEnd
	\end{array} \]
	and undefined for all missing cases.
\end{definition}

In the last case of the rule for communication---when projecting onto a role that does not participate in this communication---we map to:
\begin{align*}
	\bigsqcup_{\Role[p], i \in \Set{1, \ldots, n}} \left( \Proj{\GT_i}{\Role[p]} \right) = \left( \Proj{\GT_1}{\Role[p]} \right) \sqcup_{\Role[p]} \ldots \sqcup_{\Role[p]} \left( \Proj{\GT_n}{\Role[p]} \right)
\end{align*}
The operation $ \sqcup_{\Role[p]} $ is (similar to \cite{yoshida10}) inductively defined as:
\begin{align*}
	\LT \sqcup_{\Role[p]} \LT &= \LT\\
	\left( \LGet{\Role}{\indexSet_1} \right) \sqcup_{\Role[p]} \left( \LGet{\Role}{\indexSet_2} \right) &= \LGet{\Role}{\left( \indexSet_1 \sqcup_{\Role[p]} \indexSet_2 \right)}\\
	\indexSet \sqcup_{\Role[p]} \emptyset &= \indexSet\\
	\indexSet \sqcup_{\Role[p]} \left( \Set{ \LLab{\Label}{\tilde{\Sort}}{\LT} } \cup \indexSet[J] \right) &=\\
		& \hspace{-4em} \begin{cases}
			\Set{ \LLab{\Label}{\tilde{\Sort}}{\left( \LT' \sqcup_{\Role[p]} \LT \right)} } \cup \left( \left( \indexSet \setminus \Set{ \LLab{\Label}{\tilde{\Sort}}{\LT'} } \right) \sqcup_{\Role[p]} \indexSet[J] \right) & \text{if } \LLab{\Label}{\tilde{\Sort}}{\LT'} \in \indexSet\\
			\Set{ \LLab{\Label}{\tilde{\Sort}}{\LT} } \cup \left( \indexSet \sqcup_{\Role[p]} \indexSet[J] \right) & \text{if } \Label \notin \indexSet
		\end{cases}
\end{align*}
where $ \Label \notin \indexSet $ is short hand for $ \nexists \tilde{\Sort}', \LT'.\, \LLab{\Label}{\tilde{\Sort}'}{\LT'} \in \indexSet $, and is undefined in all other cases.

The projection of the global type \GAsys in Example~\ref{exa:globalTypeAuctioneer} onto the Roles~\Role[A], \Role[B1], and \Role[B2] is given below by the local types \LAsysA, \LAsysB, and \LAsysC, respectively.
\begin{align*}
	\LAsysA &= \Proj{\GAsys}{\Role[A]} = \LGet{B1}{\LLab{\Label[bid]}{\Sort[Int]}{\LSend{B2}{\LLab{\Label}{\Sort[Int]}{\LRec{\TVar}{}}}}}\\
		& \hspace{2em} \big( \LGet{B2}{} \big\{
		\begin{array}[t]{l}
			\LLab{\Label[bid]}{\Sort[Int]}{\LSend{B1}{\LLab{\Label}{\Sort[Int]}{\LGet{B1}{}}}} \{
				\begin{array}[t]{l}
					\LLab{\Label[bid]}{\Sort[Int]}{\LSend{B2}{\LLab{\Label}{\Sort[Int]}{\TVar}}},\\
					\LLab{\Label[no]}{}{\LSend{B2}{\LLab{\Label[s]}{\Sort[Int]}{\LEnd}}} \},
				\end{array}\\
			\LLab{\Label[no]}{}{\LSend{B1}{\LLab{\Label[s]}{\Sort[Int]}{\LEnd}}} \big\} \big)
		\end{array}\\
	\LAsysB &= \Proj{\GAsys}{\Role[B1]} = \LSend{A}{\LLab{\Label[bid]}{\Sort[Int]}{\LRec{\TVar}{}}}\big( \LGet{A}{} \big\{
		\begin{array}[t]{l}
			\LLab{\Label}{\Sort[Int]}{\LSend{A}{}} \{
				\begin{array}[t]{l}
					\LLab{\Label[bid]}{\Sort[Int]}{\TVar},\\
					\LLab{\Label[no]}{}{\LEnd} \},
				\end{array}\\
			\LLab{\Label[s]}{\Sort[Int]}{\LEnd} \big\} \big)
		\end{array}\\
	\LAsysC &= \Proj{\GAsys}{\Role[B2]} = \LGet{A}{\LLab{\Label}{\Sort[Int]}{\LRec{\TVar}{}}}\big( \LSend{A}{} \big\{
		\begin{array}[t]{l}
			\LLab{\Label[bid]}{\Sort[Int]}{\LGet{A}{}} \{
				\begin{array}[t]{l}
					\LLab{\Label}{\Sort[Int]}{\TVar},\\
					\LLab{\Label[s]}{\Sort[Int]}{\LEnd} \},
				\end{array}\\
			\LLab{\Label[no]}{}{\LEnd} \big\} \big)
		\end{array}
\end{align*}

\subsection{Session Calculus}
\label{sec:sessionCalculus}

Global types (and the local types that are derived from them) can be considered as specifications that describe the desired properties of the considered distributed system. To analyse such systems they are implemented in a session calculus. Again we use a version of the session calculus in \cite{synchMPST}, where we unify communication and branching into a single construct as done in \cite{DemangeonHonda12}. Moreover, instead of using different session channels, we annotate the session channel that is unique for each session with the roles (as it was done in \cite{DemangeonHonda12,bocchiChenDemangeonHondaYoshida13}).

\begin{definition}[Processes]
	\label{def:sessionCalculus}
	The \emph{processes} of the session calculus are given by
	\begin{align*}
		\PT &\deffTerms
		\PReq{\Role[2]..\Role[n]}{\PT}
		\sepTerms \PAcc{\Role}{\PT}
		\sepTerms \PSend{\Role_1}{\Role_2}{\Label}{\tilde{\expr}}{\PT}
		\sepTerms \PGet{\Role_2}{\Role_1}{\Set{ \PLab{\Label_i}{\tilde{\Args}_i}{\PT_i} }_{i \in \indexSet}}\\
		&\sepTerms \PCond{\cond}{\PT_1}{\PT_2}
		\sepTerms \PPar{\PT_1}{\PT_2}
		\sepTerms \PEnd
		\sepTerms \PRes{\Chan}{\PT}
		\sepTerms \PRep{\PVar}{\PT} \sepTerms \PVar
	\end{align*}
	where $ \Chan[a] $ are shared channels, $ \Role[2], \ldots, \Role[n], \Role, \Role_1, \Role_2 $ are roles, $ \Chan $ are session channels, $ \Label, \Label_i $ are labels, $ \tilde{\expr} $ are sequences of expressions to calculate values, $ \tilde{\Args}_i $ are sequences of variables, $ \indexSet $ are non-empty finite index sets, \cond are boolean conditions, and \PVar are process variables.
\end{definition}

A process initialises a session with $ \PReq{\Role[2]..\Role[n]}{\PT} $ inviting via the shared channel \Chan[a] other processes to play the roles $ \Role[2], \ldots, \Role[n] $ in a session \Chan, \ie with the session channels \Chan. Then the inviting process itself becomes \Role[1] in this session and proceeds after transmitting the invitations as \PT.
Processes can accept such an invitation to play role \Role in a session \Chan that they receive on a shared channel \Chan[a] with $ \PAcc{\Role}{\PT} $ and then proceed as \PT.
Within session \Chan role $ \Role_1 $ can transmit to role $ \Role_2 $ with $ \PSend{\Role_1}{\Role_2}{\Label}{\tilde{\expr}}{\PT} $ a label \Label and a sequence of values $ \tilde{\expr} $ and then proceed as \PT or
$ \Role_1 $ can receive from $ \Role_2 $ with $ \PGet{\Role_1}{\Role_2}{\Set{ \PLab{\Label_i}{\tilde{\Args}_i}{\PT_i} }_{i \in \indexSet}} $ one of the labels $ \Label_i $ from an indexed set of options together with a sequence of values to substitute $ \tilde{\Args}_i $ in the continuation $ \PT_i $.
The conditional $ \PCond{\cond}{\PT_1}{\PT_2} $ allows a process to proceed as $ \PT_1 $ if \cond holds or else as $ \PT_2 $.
Since we want to use the model checker \textsf{Spin} later, we restrict expressions \expr and conditions \cond to functions that are known by \textsf{Promela}, the input language of \textsf{Spin}.
\textsf{Promela} captures a wide range of functions such as basic logical and arithmetical operators.
With $ \PPar{\PT_1}{\PT_2} $ the processes $ \PT_1 $ and $ \PT_2 $ are composed in parallel.
Successful termination is denoted by $ \PEnd $.
With $ \PRes{\Chan}{\PT} $ we restrict the scope of the session channel \Chan to \PT.
With $ \PRep{\PVar}{\PT} $ we define recursion using process variables $ \PVar $.

We usually omit the curly brackets in branching with only one alternative, \ie abbreviate $ \GCom{\Role_1}{\Role_2}{\Set{ \GLab{\Label}{\tilde{\Sort}}{\GT} }} $ by $ \GCom{\Role_1}{\Role_2}{\GLab{\Label}{\tilde{\Sort}}{\GT}} $ and $ \PGet{\Role_1}{\Role_2}{\Set{ \PLab{\Label}{\tilde{\Args}}{\PT} }} $ by $ \PGet{\Role_1}{\Role_2}{\PLab{\Label}{\tilde{\Args}}{\PT}} $.
We also often omit trailing $ \PEnd $.
Throughout the paper we use '$.$' to denote sequential composition, where the part before the '$.$' is called \emph{prefix} and the sub-term(s) in the scope of the '$.$' are \emph{guarded} by this prefix.
Moreover, conditionals guard both of their sub-terms.
Similarly, we use round brackets to denote \emph{binders}, where the variables (for names, types, or processes) within the brackets are \emph{bound} in the following sub-term.
A name is free if it is not bounded.
Let $ \Names{M} $ denote the set of names and $ \FreeNames{M} $ denote the set of free names in $ M $, where $ M $ is a type or a process.
We assume that all process variables in processes and type variables in types are bounded and guarded, where process variables have to be guarded by communication prefixes.

A substitution $ \Set{ \Subst{\Args[y]_1}{\Args_1}, \ldots \Subst{\Args[y]_n}{\Args_n} } = \Set{ \Subst{\tilde{\Args[y]}}{\tilde{\Args}} } $ is a finite mapping from names to names, where the names in $ \tilde{\Args} $ are pairwise distinct. The application of a substitution on a term $ \PT\Set{ \Subst{\tilde{\Args[y]}}{\tilde{\Args}} } $ is defined as the result of simultaneously replacing all free occurrences of $ \Args_i $ by $ \Args[y]_i $, possibly applying alpha-conversion to avoid capture or name clashes. For all names $ n \notin \tilde{\Args} $ the substitution behaves as the identity mapping.
We naturally extend substitution of names to the substitution of process variables by terms.

We use structural congruence ($ \equiv $) to abstract from syntactically different but semantically similar processes, where $ \equiv $ is the least congruence that satisfies alpha-conversion ($ \equiv_{\alpha} $) and the rules:
\vspace{-0.5em}
\[ \begin{array}{c}
	\PPar{\PT}{\PEnd} \equiv \PT
	\hspace{2em} \PPar{\PT_1}{\PT_2} \equiv \PPar{\PT_2}{\PT_1}
	\hspace{2em} \PPar{\PT_1}{\left( \PPar{\PT_2}{\PT_3} \right)} \equiv \PPar{\left( \PPar{\PT_1}{\PT_2} \right)}{\PT_3}\\
	\PRep{\PVar}{\PT} \equiv \PT\Set{ \Subst{\PRep{\PVar}{\PT}}{\PVar} }
	\hspace{2em} \PRes{\Chan}{\PRes{\Chan'}{\PT}} \equiv \PRes{\Chan'}{\PRes{\Chan}{\PT}}
	\hspace{2em} \PRes{\Chan}{\PEnd} \equiv \PEnd\\
	\PRes{\Chan}{\left( \PPar{\PT_1}{\PT_2} \right)} \equiv \PPar{\PT_1}{\PRes{\Chan}{\PT_2}} \quad \text{if } \Chan \notin \FreeNames{\PT_1}
\end{array} \]

\begin{figure}[t]
	\[ \begin{array}{c}
		\left( \textsf{Link} \right) \dfrac{}{\PPar{\PReq{\Role[2]..\Role[n]}{\PT_1}}{\PPar{\PAcc{\Role[2]}{\PT_2}}{\PPar{\ldots}{\PAcc{\Role[n]}{\PT_n}}}} \Step \PRes{\Chan}{\left( \PPar{\PT_1}{\PPar{\PT_2}{\PPar{\ldots}{\PT_n}}}} \right)} \vspace{0.75em}\\
		\left( \textsf{Com} \right) \dfrac{j \in \indexSet}{\PPar{\PSend{\Role_1}{\Role_2}{\Label_j}{\tilde{\expr}}{\PT}}{\PGet{\Role_2}{\Role_1}{\Set{ \PLab{\Label_i}{\tilde{\Args}_i}{\PT_i} }_{i \in \indexSet}}} \Step \PPar{\PT}{\left( \PT_j\Set{ \Subst{\tilde{\expr}}{\tilde{\Args}_j} } \right)}} \vspace{0.75em}\\
		\left( \textsf{If-T} \right) \dfrac{\cond}{\PCond{\cond}{\PT_1}{\PT_2} \Step \PT_1}
		\hspace{1.5em}
		\left( \textsf{If-F} \right) \dfrac{\neg\cond}{\PCond{\cond}{\PT_1}{\PT_2} \Step \PT_2} \vspace{0.75em}\\
		\left( \textsf{Par} \right) \dfrac{\PT_1 \Step \PT_1'}{\PPar{\PT_1}{\PT_2} \Step \PPar{\PT_1'}{\PT_2}}
		\hspace{1.5em}
		\left( \textsf{Res} \right) \dfrac{\PT \Step \PT'}{\PRes{\Chan}{\PT} \Step \PRes{\Chan}{\PT'}} \vspace{0.75em}\\
		\left( \textsf{Struc} \right) \dfrac{\PT_1 \equiv \PT_2 \quad \PT_2 \Step \PT_2' \quad \PT_2' \equiv \PT_1'}{\PT_1 \Step \PT_1'}
	\end{array} \]
	\vspace{-1em}
	\caption{Reduction Semantics of the Session Calculus.}
	\label{fig:reductionSemantics}
\end{figure}

The reduction semantics of the session calculus is given by the rules in Figure~\ref{fig:reductionSemantics}.
The Rule~\textsf{Link} initialises a session \Chan on the roles $ \Role[1], \ldots, \Role[n] $, where \Role[1] requested the session on channel \Chan[a] and each \Role[i] participates in the session as $ \PT_i $.
Communication within a session \Chan is described by Rule~\textsf{Com}, where in the case of matching roles and labels the continuations of sender and receiver are unguarded and the variables $ \tilde{\Args} $ are replaced by the values $ \tilde{\expr} $ in the receiver.
The Rules~\textsf{If-T} and \textsf{If-F} reduce conditionals as expected.
The remaining rules allow for steps in various contexts and are standard.

In contrast to the standard \piCal (as \eg in \cite{milnerParrowWalker92}) the communication prefixes of the session calculus mention an explicit acting role next to the channel, regardless whether the prefix is used to initialise a session or to transmit a message within a session.

\begin{definition}[Actor]
	A process $ \PT $ has an \emph{actor on $ \Actor{c}{r_1} $} if $ \PT $ has an unguarded subterm of the form $ \PReq[\Chan[c]]{\Role[2]..\Role[n]}{\PT} $ with $ \Role_1 = \Role[1] $ or $ \PAcc[\Chan[c]]{\Role_1}{\PT} $ (for session invitations) or an unguarded subterm of the form $ \PSend[\Chan[c]]{\Role_1}{\Role_2}{\Label}{\tilde{\expr}}{\PT} $ or $ \PGet[\Chan[c]]{\Role_1}{\Role_2}{\Set{ \PLab{\Label_i}{\tilde{\Args}_i}{\PT_i} }_{i \in \indexSet}} $ (for communication).
Let $ \Actors{\PT} $ be set of actors in \PT.
\end{definition}

\noindent
If unambiguous, \ie if there is only one session, we omit the session channel and abbreviate actors by their role.

As described in \cite{synchMPST,hondaYoshidaCarbone08,hondaYoshidaCarbone16}, global types are projected onto to their roles into so-called local types (compare to Section~\ref{sec:types}) that are then used to build type environments.
Intuitively, a process is \emph{well-typed \wrt a global type} if it behaves as specified in the type.
Therefore, the process is compared in a static analysis with type environments that are derived from the global type.
Similarly, a system that implements more than a single session is \emph{well-typed \wrt $ \Set{ \left( \GT_i, \Chan_i \right) }_{i \in \indexSet} $} if each session $ \Chan_i $ behaves as specified in the global type $ \GT_i $ and if the interleaving of different sessions does not introduce deadlocks (compare to \cite{BettiniAtall08}).
We formally define well-typed processes for the above variant of \MPST in Section~\ref{sec:wellTypedProcesses} and show basic properties---in particular subject reduction, linearity, and error-freedom---in Section~\ref{sec:basicProperties}.
Moreover, we rely on the observation that in well-typed processes different actors of the same session are composed in parallel, whereas all actions of the same actor are composed sequentially.

\begin{lemma}[Actors are Sequential]
	\label{lem:actorsAreSequential}
	If \PT is well-typed then all actions of the same actor are composed sequentially in \PT.
\end{lemma}

\subsection{Well-Typed Processes}
\label{sec:wellTypedProcesses}

Processes are combined with type environments into \emph{type judgements}.
A judgement is of the form $ \Gamma \vdash \PT \triangleright \Delta $, where $ \Gamma $ is a global type environment connecting shared channels with their global types and values with their sorts, \PT is a process, and $ \Delta $ is a session environment containing the projections of global types onto their roles and some additional control information.
Type environments are sets of assignments, but we usually omit the curly brackets and write $ \Gamma, A $ (or $ \Delta, A $) for the union of $ \Gamma $ (or $ \Delta $) and $ \Set{ A } $.

\begin{definition}[Type Environments]
	The \emph{global type environments} and the \emph{session environments} are given by
	\begin{align*}
		\Gamma & \deffTerms \emptyset \sepTerms \Gamma, \GGlob{a}{G} \sepTerms \Gamma, \GGlobS{a}{G}{s} \sepTerms \Gamma, \GVal{\Args}{\Sort} \sepTerms \Gamma, \GPVar{\PVar}{\Delta}\\
		\Delta & \deffTerms \emptyset \sepTerms \Delta, \LInv{a}{r} \sepTerms \Delta, \LLoc{r}{\LT}
	\end{align*}
	where \Chan[a] are shared channels, \GT are global types, \Chan are session channels, \Args are names, \Sort are sorts, \PVar are process variables, \Role are roles, and \LT are local types.
\end{definition}

Assignments \GGlob{a}{G} connect a global type of a session with the shared channel that is used to initialise this session.
The type system ensures that shared channels are used exactly once.
Therefore, we add the session channel \Chan, \ie rewrite \GGlob{a}{G} into \GGlobS{a}{G}{s}, after the shared channel was used and require that $ \Gamma, \GGlob{a}{G} $ implies $ \nexists \GT', \Chan \logdot \GGlob{a}{G'} \in \Gamma \lor \GGlobS{a}{G'}{s} \in \Gamma $.
Similarly, we assume that $ \Gamma, \GGlobS{a}{G}{s} $ implies $ \nexists \Chan[a]', \GT', \Chan' \logdot \GGlob{a}{G'} \in \Gamma \lor \GGlobS{a}{G'}{s'} \in \Gamma \lor \GGlobS{a'}{G'}{s} \in \Gamma $.
Assignments \GVal{\Args}{\Sort} state that variable \Args is of sort \Sort, which implies that \Args is no session channel.
We assume that the sort of variables is unique, \ie that $ \Gamma, \GVal{\Args}{\Sort} $ implies $ \nexists \Sort'\logdot \GVal{\Args}{\Sort'} \in \Gamma $.
We write $ \Gamma \Vdash \GVal{\expr}{\Sort} $ if for all names $ v $ in \expr there is some $ \GVal{v}{\Sort_v} $ in $ \Gamma $ and if with these sorts for its variables \expr is of sort \Sort for all possible evaluations.
We abbreviate $ \Gamma, \GVal{\expr_1}{\Sort_1}, \ldots, \GVal{\expr_n}{\Sort_n} $ by $ \Gamma, \GVal{\tilde{\expr}}{\tilde{\Sort}} $ and $ \Gamma \Vdash \GVal{\expr_1}{\Sort_1}, \ldots, \Gamma \Vdash \GVal{\expr_n}{\Sort_n} $ by $ \Gamma \Vdash \GVal{\tilde{\expr}}{\tilde{\Sort}} $.
Assignments \GPVar{\PVar}{\Delta} save the current state of a session environment connected to a process variable \PVar, in order to check recursive processes.
We assume that $ \Gamma, \GPVar{\PVar}{\Delta} $ implies $ \nexists \Delta'\logdot \GPVar{\PVar}{\Delta'} \in \Gamma $.

The obligation \LInv{a}{r} tells us that the considered process needs to invite role~\Role via the shared channel \Chan[a].
We assume that $ \Delta, \LInv{a}{r} $ implies $ \LInv{a}{r} \notin \Delta $.
Assignments \LLoc{r}{\LT} connect the local type \LT to the role~\Role in the session~\Chan.
We assume that $ \Delta, \LLoc{r}{\LT} $ implies that $ \nexists \LT' \logdot \LLoc{r}{\LT'} \in \Delta $.
Let $ \Delta_1 \otimes \Delta_2 = \Delta_1 \cup \Delta_2 $ if $ \Delta_1 \cap \Delta_2 = \emptyset $ and undefined else.

\begin{figure}[tp]
	\[ \begin{array}{c}
		\left( \textsf{Req} \right) \dfrac{\Roles{\GT} = \Set{ 1, \ldots, n } \quad \Gamma, \GGlobS{a}{G}{s} \vdash \PT \triangleright \Delta, \LLoc{1}{\Proj{\GT}{\Role[1]}}}{\Gamma, \GGlob{a}{G} \vdash \PReq{\Role[2]..\Role[n]}{\PT} \triangleright \Delta, \LInv{a}{1}}
		\vspace{0.75em}\\
		\left( \textsf{Acc} \right) \dfrac{\Role \in \Roles{\GT} \quad \Gamma, \GGlobS{a}{G}{s} \vdash \PT \triangleright \Delta, \LLoc{r}{\Proj{\GT}{\Role}}}{\Gamma, \GGlob{a}{G} \vdash \PAcc{\Role}{\PT} \triangleright \Delta, \LInv{a}{r}}
		\vspace{0.75em}\\
		\left( \textsf{Send} \right) \dfrac{j \in \indexSet \quad \Gamma \Vdash \GVal{\tilde{\expr}}{\tilde{\Sort}_j} \quad \Gamma \vdash \PT \triangleright \Delta, \LLoc{\Role_1}{\LT_j}}{\Gamma \vdash \PSend{\Role_1}{\Role_2}{\Label_j}{\tilde{\expr}}{\PT} \triangleright \Delta, \LLoc{\Role_1}{\LSend{\Role_2}{\Set{ \LLab{\Label_i}{\tilde{\Sort}_i}{\LT_i} }_{i \in \indexSet}}}}
		\vspace{0.75em}\\
		\left( \textsf{Get} \right) \dfrac{\indexSet \subseteq \indexSet[J] \quad \forall i \in \indexSet\logdot \Gamma, \GVal{\tilde{\Args}_i}{\tilde{\Sort}_i} \vdash \PT_i \triangleright \Delta, \LLoc{\Role_2}{\LT_i}}{\Gamma \vdash \PGet{\Role_2}{\Role_1}{\Set{ \PLab{\Label_j}{\tilde{\Args}_j}{\PT_j} }_{j \in \indexSet[J]}} \triangleright \Delta, \LLoc{\Role_2}{\LGet{\Role_1}{\Set{ \LLab{\Label_i}{\tilde{\Sort}_i}{\LT_i} }_{i \in \indexSet}}}}
		\hspace{2em}
		\left( \textsf{End} \right) \dfrac{}{\Gamma \vdash \PEnd \triangleright \emptyset}
		\vspace{0.75em}\\
		\left( \textsf{Cond} \right) \dfrac{\Gamma \vdash \PT_1 \triangleright \Delta \quad \Gamma \vdash \PT_2 \triangleright \Delta}{\Gamma \vdash \PCond{\cond}{\PT_1}{\PT_2} \triangleright \Delta}
		\hspace{2em}
		\left( \textsf{Par} \right) \dfrac{\Gamma \vdash \PT_1 \triangleright \Delta_1 \quad \Gamma \vdash \PT_2 \triangleright \Delta_2}{\Gamma \vdash \PPar{\PT_1}{\PT_2} \triangleright \Delta_1 \otimes \Delta_2}
		\vspace{0.75em}\\
		\left( \textsf{Res} \right) \dfrac{\Delta, \LLoc{1}{\Proj{\GT}{\Role[1]}}, \ldots, \LLoc{n}{\Proj{\GT}{\Role[n]}} \Step^* \Delta' \quad \Roles{\GT} = \Set{ \Role[1], \ldots, \Role[n] } \quad \Gamma, \GGlobS{a}{G}{s} \vdash \PT \triangleright \Delta'}{\Gamma, \GGlob{a}{G} \vdash \PRes{\Chan}{\PT} \triangleright \Delta, \LInv{a}{1}, \ldots, \LInv{a}{n}}
		\vspace{0.75em}\\
		\left( \textsf{Rec} \right) \dfrac{\Gamma, \GPVar{\PVar}{\Delta} \vdash \PT \triangleright \Delta}{\Gamma \vdash \PRep{\PVar}{\PT} \triangleright \Delta}
		\hspace{2em}
		\left( \textsf{Var} \right) \dfrac{}{\Gamma, \GPVar{\PVar}{\Delta} \vdash \PVar \triangleright \Delta}
	\end{array} \]
	\vspace{-1em}
	\caption{Typing Rules.}
	\label{fig:typingRules}
\end{figure}

Type judgements are derived from the typing rules in Figure~\ref{fig:typingRules}, where we equate within type judgements processes modulo alpha conversion, local types modulo the unfolding of their recursion by the rule $ \LRec{\TVar}{\LT} = \LT\Set{ \Subst{\LRec{\TVar}{\LT}}{\TVar} } $, and session environments modulo terminated types by the rule $ \Delta, \LLoc{\Role}{\LEnd} = \Delta $.

For each session invitation $ \PReq{\Role[2]..\Role[n]}{\PT} $, Rule~\textsf{Req} requires \GGlob{a}{G} in the global environment and rewrites it into \GGlobS{a}{G}{s} to mark that \Chan[a] was used to invite the session \Chan.
Then it checks whether the global type has the invited number of roles, consumes the obligation \LInv{a}{1} from the session environment and adds \LLoc{1}{\Proj{\GT}{\Role[1]}}, \ie requires that the continuation \PT behaves as specified by the projection of \GT to role~\Role[1] in the session environment.
Rule~\textsf{Acc} is similar for a process that accepts to participate as role~\Role[r] in the invited session.

Rule~\textsf{Send} checks whether the process sends if its local type requires this, the roles of the process and the local type match, the transmitted label is one of the labels specified in the local type by $ j \in \indexSet $, the transmitted expressions are of the required sorts by $ \Gamma \Vdash \GVal{\tilde{\expr}}{\tilde{\Sort}_j} $, and the continuation \PT behaves as specified by $ \LT_j $.
Rule~\textsf{Get} checks whether the process receives if its local type requires this, the roles of the process and the local type match, and each branch $ \PT_i $ with the variables $ \GVal{\tilde{\Args}_i}{\tilde{\Sort}_i} $ behaves as specified by $ \LT_i $.
Note that in contrast to \eg \cite{synchMPST,DemangeonHonda12,bocchiChenDemangeonHondaYoshida13,hondaYoshidaCarbone08,hondaYoshidaCarbone16} we allow that receivers implement unnecessary branches, to allow types to follow the reductions of the system and to deal with branches already ruled out by a former step.
However, since the sender is checked as well, the type system ensures that only branches that are specified by the type can happen.

Rule~\textsf{End} states that type judgements on successful termination are valid \wrt arbitrary global environments but only the empty session environment.
Rule~\textsf{Cond} checks whether both branches of a conditional behave similarly.
Instead, to check a parallel composition in processes, Rule~\textsf{Par} requires that it is possible to split the session environment into disjoint parts such that each parallel branch behaves as specified by one part of the session environment.

Rule~\textsf{Res} requires that there is an unused shared channel connected to a global type and that the current session environment extended by the projections of this global type can evolve such that \PT behaves as specified by the remainder of this extended session environment, where the relation $ \Step $ on session environments is given by the Rules~\textsf{Com'} and \textsf{Cut} below.
It is necessary to analyse systems that already entered a session.
To check recursion, the Rules~\textsf{Rec} and \textsf{Var} check whether the body of a recursive process reaches the same session environment after one iteration.

Coherence is used to describe the fact that a system implements all roles of the global types that belong the the considered sessions.

\begin{definition}[Coherence]
	A session environment $ \Delta $ is \emph{coherent \wrt a set of pairs of global types and pairwise distinct names $ \Set{ \left( \GT_i, \Chan[n]_i \right) }_{i \in \indexSet} $}, if for all session channels \Chan in $ \Delta $ there exists $ j \in \indexSet $ such that $ \Chan[n]_j = \Chan $ and $ \Delta \cap \Set{ \LLoc{\Role}{\LT} \mid \Role \text{ is a role and } \LT \text{ is a local type} } = \Set{ \LLoc{\Role}{\Proj{\GT_j}{\Role}} \mid \Role \in \Roles{\GT_j} }$ and for all its shared channels \Chan[a] there exists $ j \in \indexSet $ such that $ \Chan[n]_j = \Chan[a] $ and $ \Set{ \Role \mid \LInv{a}{r} \in \Delta } = \Roles{\GT_j} $.\\
	Moreover, $ \Delta $ is \emph{coherent \wrt a global type \GT} if $ \Delta $ is coherent \wrt $ \Set{ \left( \GT, \Chan[n] \right) } $ for some name $ \Chan[n] $ and $ \Delta $ is \emph{coherent} if $ \Delta $ is coherent \wrt some $ \Set{ \left( \GT_i, \Chan[n]_i \right) }_{i \in \indexSet} $.
\end{definition}

We map the reduction of communications in Rule~\textsf{Com} of Figure~\ref{fig:reductionSemantics} on the rule
\begin{align*}
	\left( \textsf{Com'} \right) \dfrac{j \in \indexSet}{
		\begin{array}{l}
			\Delta, \LLoc{\Role_1}{\LSend{\Role_2}{\Set{ \LLab{\Label_i}{\tilde{\Sort}_i}{\LT_i} }_{i \in \indexSet}}}, \LLoc{\Role_2}{\LGet{\Role_1}{\Set{ \LLab{\Label_i}{\tilde{\Sort}_i}{\LT_i'} }_{i \in \indexSet}}}\\
			\Step \Delta, \LLoc{\Role_1}{\LT_j}, \LLoc{\Role_2}{\LT_j'}
		\end{array}}
\end{align*}
and add a rule to remove superfluous branches of receivers
\begin{align*}
	\left( \textsf{Cut} \right) \dfrac{\emptyset \subset \indexSet[J] \subseteq \indexSet}{\Delta, \LLoc{\Role_2}{\LGet{\Role_1}{\Set{ \LLab{\Label_i}{\tilde{\Sort}_i}{\LT_i'} }_{i \in \indexSet}}} \Step \Delta, \LLoc{\Role_2}{\LGet{\Role_1}{\Set{ \LLab{\Label_i}{\tilde{\Sort}_i}{\LT_i'} }_{i \in \indexSet[J]}}}}
\end{align*}
such that session environments can follow the evolution of processes.

We call a process \emph{role-distributed} if it composes different actors of the same session in parallel (for all sessions and all actors of a session).

\begin{definition}[Well-Typed Processes, Single Session]
	Let \PT be a process, \GT a global type, $ \Gamma $ a global type environment, and $ \Delta $ a session environment.
	For processes with a single session we have:
	\begin{itemize}
		\item \PT is \emph{well-typed \wrt $ \Gamma $ and $ \Delta $} if \PT is role-distributed, $ \Delta $ is coherent, and $ \Gamma \vdash \PT \triangleright \Delta $.
		\item \PT is \emph{well-typed} if there are $ \Gamma, \Delta $ such that \PT is well-typed \wrt $ \Gamma, \Delta $.
		\item \PT is \emph{well-typed \wrt \GT} if there are $ \Gamma, \Delta $ such that \PT is role-distributed, \GT is the only global type in $ \Gamma $, $ \Delta $ is coherent \wrt \GT, and $ \Gamma \vdash \PT \triangleright \Delta $.
	\end{itemize}
\end{definition}

The definition of well-typed processes is more difficult for several interleaved sessions.
As described in \cite{BettiniAtall08}, we have to ensure that actions of different sessions do not cause deadlocks by cyclic dependencies.
Therefore, \cite{BettiniAtall08} introduce an \emph{interaction} type system for global progress in dynamically interleaved multiparty sessions.
The interaction type system introduced in \cite{BettiniAtall08} considers an asynchronous variant of \MPST there senders release their messages onto message queues from which receivers can read in a subsequent step.
To check for global progress in dynamically interleaved multiparty sessions, the interaction type system collects dependencies between interactions of different services, \ie different sessions and their associated shared channels.
Since \cite{BettiniAtall08} considers asynchronous communication, they collect the dependencies of a receiver to the interactions with other services in its continuation.
To obtain an interaction type system for the above synchronous \MPST variant, we have to consider dependencies also for senders, \ie treat senders in the same way as receivers.
Moreover, we have to extend the collection of dependencies also to session invitations, \ie session requests $ \PReq{\Role[2]..\Role[n]}{\PT} $ and their corresponding receivers $ \PAcc{\Role}{\PT} $ have to produce the same kind of dependencies towards interactions of other services as the communication prefixes within the respective session.
A process $ \PT $ is \emph{globally progressing} if it can be typed in the interaction type system.

\begin{definition}[Well-Typed Processes, Interleaved Sessions]
	Let \PT be a process without name clashes on session channels, $ \Set{ \left( \GT_i, \Chan_i \right) }_{i \in \indexSet} $ a set of pairs of global types and pairwise distinct session channels, $ \Gamma $ a global type environment, and $ \Delta $ a session environment.
	For processes with interleaved sessions we have:
	\begin{itemize}
		\item \PT is \emph{well-typed \wrt $ \Gamma $ and $ \Delta $} if \PT is role-distributed, $ \Delta $ is coherent, $ \Gamma \vdash \PT \triangleright \Delta $, and \PT is globally progressing.
		\item \PT is \emph{well-typed} if there are $ \Gamma, \Delta $ such that \PT is well-typed \wrt $ \Gamma, \Delta $.
		\item \PT is \emph{well-typed \wrt $ \Set{ \left( \GT_i, \Chan_i \right) }_{i \in \indexSet} $} if there are $ \Gamma, \Delta $ such that \PT is role-distributed, $ \Set{ \GT_i }_{i \in \indexSet } $ are the global types in $ \Gamma $, $ \Delta $ is coherent \wrt the types $ \Set{ \left( \GT_i, \Chan[n]_i \right) }_{i \in \indexSet} $ for some $ \Chan[n]_i $, $ \Gamma \vdash \PT \triangleright \Delta $, \PT is globally progressing, and for all $ i \in \indexSet $ either $ \Chan[n]_i = \Chan_i $ or $ \Gamma \vdash \PT \triangleright \Delta $ connects the shared channel $ \Chan[n]_i $ with $ \Chan_i $, \ie $ \GGlob{n_i}{G} \in \Gamma $ is transferred into \GGlobS{n_i}{G}{s_i} by one of the Rules~\textnormal{\textsf{Req}}, \textnormal{\textsf{Acc}}, or \textnormal{\textsf{Res}}.
	\end{itemize}
\end{definition}

\subsection{Basic Properties}
\label{sec:basicProperties}

Type judgements are preserved modulo structural congruence.

\begin{lemma}[Structural Congruence]
	\label{lem:structuralCongruence}
	If $ \Gamma \vdash \PT \triangleright \Delta $ and $ \PT \equiv \PT' $ then $ \Gamma \vdash \PT' \triangleright \Delta $.
\end{lemma}

\begin{proof}
	The proof is by induction on the rules of structural congruence that are used to obtain $ \PT \equiv \PT' $.
	In each case we derive from the structure of \PT and Figure~\ref{fig:typingRules} information about the proof of $ \Gamma \vdash \PT \triangleright \Delta $ and use them to show $ \Gamma \vdash \PT' \triangleright \Delta $.
	Thereby, we rely on the commutativity and associativity of $ \otimes $, the fact that a judgement for the process \PEnd can be derived if and only if the considered session environment is empty, and that we equate in judgements local types modulo the unfolding of recursion.
\end{proof}

Type judgements are preserved modulo the substitution of values by values of the same sort.

\begin{lemma}[Substitution]
	\label{lem:substitution}
	If $ \Gamma, \GVal{\Args}{\Sort}, \GVal{\Args[y]}{\Sort} \vdash \PT \triangleright \Delta $ then $ \Gamma, \GVal{\Args[y]}{\Sort} \vdash \PT\Set{ \Subst{\Args[y]}{\Args} } \triangleright \Delta $.
\end{lemma}

\begin{proof}
	The proof is by induction on the typing rules used to obtain $ \Gamma, \GVal{\Args}{\Sort}, \GVal{\Args[y]}{\Sort} \vdash \PT \triangleright \Delta $.
	Since local types do not contain names, they are not affected by the substitution.
	Rule~\textsf{Send} is the only typing rule that checks for sorts of names.
	This case follows from the observation that $ \Gamma, \GVal{\Args}{\Sort}, \GVal{\Args[y]}{\Sort} \Vdash \expr $ implies $ \Gamma, \GVal{\Args[y]}{\Sort} \Vdash \expr\Set{ \Subst{\Args[y]}{\Args} } $.
\end{proof}

The sorts of fresh names are not relevant for type judgements.

\begin{lemma}[Fresh Name]
	\label{lem:freshName}
	If $ \Gamma, \GVal{\Args}{\Sort} \vdash \PT \triangleright \Delta $ and $ \Args \notin \FreeNames{\PT} $ then $ \Gamma \vdash \PT \triangleright \Delta $.
\end{lemma}

\begin{proof}
	The proof is by induction on the typing rules used to obtain $ \Gamma, \GVal{\Args}{\Sort} \vdash \PT \triangleright \Delta $.
\end{proof}

Subject reduction is a fundamental property of type systems that allows for static type checking.
It shows that if a type judgement for a process can be derived with a coherent session environment then we can also derive type judgements for its derivatives.

\begin{lemma}[Subject Reduction]
	\label{lem:subjectReduction}
	If $ \Gamma \vdash \PT \triangleright \Delta $, $ \Delta $ is coherent, and $ \PT \Step \PT' $ then there is $ \Delta' $ such that $ \Gamma \vdash \PT' \triangleright \Delta' $, $ \Delta' $ is coherent, and $ \Delta \Step \Delta' $.
\end{lemma}

\begin{proof}
	Assume $ \Gamma \vdash \PT \triangleright \Delta $ and that $ \Delta $ is coherent.
	The proof is by induction on the reduction rules of Figure~\ref{fig:reductionSemantics} that are used to obtain $ \PT \Step \PT' $.
	\begin{description}
		\item[Case of Rule~\textsf{Link}:]
			In this case $ \PT = \PPar{\PReq{\Role[2]..\Role[n]}{\PT_1}}{\PPar{\PAcc{\Role[2]}{\PT_2}}{\PPar{\ldots}{\PAcc{\Role[n]}{\PT_n}}}} $ and $ \PT' = \PRes{\Chan}{\left( \PPar{\PT_1}{\PPar{\PT_2}{\PPar{\ldots}{\PT_n}}} \right)} $.
			By Figure~\ref{fig:typingRules}, then the proof of $ \Gamma \vdash \PT \triangleright \Delta $ starts with $ n - 1 $ applications of Rule~\textsf{Par} that splits the judgement into $ \Gamma \vdash \PReq{\Role[2]..\Role[n]}{\PT_1} \triangleright \Delta_1 $ and $ \Gamma \vdash \PAcc{\Role[i]}{\PT_i} \triangleright \Delta_i $ for $ 2 \leq i \leq n $ such that $ \Delta = \Delta_1 \otimes \ldots \otimes \Delta_n $.
			By the Rules~\textsf{Req} and \textsf{Acc}, then $ \Roles{\GT} = \Set{ \Role[1], \ldots, \Role[n] } $, $ \Gamma = \Gamma', \GGlob{a}{G} $, $ \Delta_i = \Delta_i', \LInv{a}{i} $, and $ \Gamma', \GGlobS{a}{G}{s} \vdash \PT_i \triangleright \Delta_i', \LLoc{i}{\Proj{\GT}{\Role[i]}} $ for all $ 1 \leq i \leq n $.
			Because of $ \Delta_1 \otimes \ldots \otimes \Delta_n $ and $ \Delta_i = \Delta_i', \LInv{a}{i} $, $ \Delta' = \Delta_1' \otimes \ldots \otimes \Delta_n' $ is defined.
			By $ n - 1 $ applications of Rule~\textsf{Par} and since $ \Delta_i', \LLoc{i}{\Proj{\GT}{\Role[i]}} $ implies that $ \Delta_i' $ does not contain an assignment for \Actor{\Chan}{\Role[i]}, then $ \Gamma', \GGlobS{a}{G}{s} \vdash \PPar{\PT_1}{\PPar{\ldots}{\PT_n}} \triangleright \Delta', \LLoc{1}{\Proj{\GT}{\Role[1]}}, \ldots, \LLoc{n}{\Proj{\GT}{\Role[n]}} $.
			By Rule~\textsf{Res} and the reflexivity of $ \Step^* $, then $ \Gamma \vdash \PT' \triangleright \Delta $.
			By the reflexivity of $ \Step^* $, then $ \Delta \Step^* \Delta $.
		\item[Case of Rule~\textsf{Com}:]
			In this case $ \PT = \PPar{\PSend{\Role_1}{\Role_2}{\Label_j}{\tilde{\expr}}{\PT[Q]}}{\PGet{\Role_2}{\Role_1}{\Set{ \PLab{\Label_i}{\tilde{\Args}_i}{\PT_i} }_{i \in \indexSet}}} $, $ j \in \indexSet $, and $ \PT' = \PPar{\PT[Q]}{\left( \PT_j\Set{ \Subst{\tilde{\expr}}{\tilde{\Args}_j} } \right)} $.
			By Figure~\ref{fig:typingRules}, then the proof of $ \Gamma \vdash \PT \triangleright \Delta $ starts with Rule~\textsf{Par} that splits the judgement into $ \Gamma \vdash \PSend{\Role_1}{\Role_2}{\Label_j}{\tilde{\expr}}{\PT[Q]} \triangleright \Delta_{\PT[Q]} $ and $ \Gamma \vdash \PGet{\Role_2}{\Role_1}{\Set{ \PLab{\Label_i}{\tilde{\Args}_i}{\PT_i} }_{i \in \indexSet}} \triangleright \Delta_{\PT} $ such that $ \Delta = \Delta_{\PT[Q]} \otimes \Delta_{\PT} $.
			By the Rules~\textsf{Send} and \textsf{Get} and coherence, then $ \Gamma \Vdash \GVal{\tilde{\expr}}{\tilde{\Sort}_j} $, $ \Delta_{\PT[Q]} = \Delta_{\PT[Q]}', \LLoc{\Role_1}{\LSend{\Role_2}{\Set{ \LLab{\Label_i}{\tilde{\Sort}_j}{\LT_i} }_{i \in \indexSet}}} $, $ \Gamma \vdash \PT[Q] \triangleright \Delta_{\PT[Q]}', \LLoc{\Role_1}{\LT_j} $, $ \Delta_{\PT} = \Delta_{\PT}', \LLoc{\Role_2}{\LGet{\Role_1}{\Set{ \LLab{\Label_i}{\tilde{\Sort}_j}{\LT_i'} }_{i \in \indexSet}}} $, and $ \Gamma, \GVal{\tilde{\Args}_j}{\tilde{\Sort}_j} \vdash \PT_j \triangleright \Delta_{\PT}', \LLoc{\Role_2}{\LT_j'} $.
			By the Lemmata~\ref{lem:substitution} and \ref{lem:freshName}, then $ \Gamma \vdash \PT_j\Set{ \Subst{\expr}{\tilde{\Args}_j} } \triangleright \Delta_{\PT}', \LLoc{\Role_2}{\LT_j'} $.
			Because $ \Delta_{\PT[Q]} \otimes \Delta_{\PT} $ is defined, $ \Delta_{\PT[Q]} = \Delta_{\PT[Q]}', \LLoc{\Role_1}{\LSend{\Role_2}{\Set{ \LLab{\Label_i}{\tilde{\Sort}_j}{\LT_i} }_{i \in \indexSet}}} $, and $ \Delta_{\PT} = \Delta_{\PT}', \LLoc{\Role_2}{\LGet{\Role_1}{\Set{ \LLab{\Label_i}{\tilde{\Sort}_j}{\LT_i'} }_{i \in \indexSet}}} $, $ \Delta' = \Delta_{\PT[Q]}', \LLoc{\Role_1}{\LT_j} \otimes \Delta_{\PT}', \LLoc{\Role_2}{\LT_j'} $ is defined and is coherent.
			By Rule~\textsf{Par}, then $ \Gamma \vdash \PT' \triangleright \Delta' $.
			By Rule~\textsf{Com'}, then $ \Delta \Step^* \Delta' $.
		\item[Case of Rule~\textsf{If-T}:]
			In this case $ \PT = \PCond{\cond}{\PT_1}{\PT_2} $, $ \PT' = \PT_1 $, and \cond is satisfied.
			By Figure~\ref{fig:typingRules}, then the proof of $ \Gamma \vdash \PT \triangleright \Delta $ starts with Rule~\textsf{Cond} and, thus, $ \Gamma \vdash \PT' \triangleright \Delta $.
			By the reflexivity of $ \Step^* $, then $ \Delta \Step^* \Delta $.
		\item[Case of Rule~\textsf{If-F}:]
			In this case $ \PT = \PCond{\cond}{\PT_1}{\PT_2} $, $ \PT' = \PT_2 $, and \cond is not satisfied.
			By Figure~\ref{fig:typingRules}, then the proof of $ \Gamma \vdash \PT \triangleright \Delta $ starts with Rule~\textsf{Cond} and, thus, $ \Gamma \vdash \PT' \triangleright \Delta $.
			By the reflexivity of $ \Step^* $, then $ \Delta \Step^* \Delta $.
		\item[Case of Rule~\textsf{Par}:]
			In this case $ \PT = \PPar{\PT_1}{\PT_2} $, $ \PT' = \PPar{\PT_1'}{\PT_2} $, and $ \PT_1 \Step \PT_1' $.
			By Figure~\ref{fig:typingRules}, then the proof of $ \Gamma \vdash \PT \triangleright \Delta $ starts with Rule~\textsf{Par} that splits the judgement into $ \Gamma \vdash \PT_1 \triangleright \Delta_1 $ and $ \Gamma \vdash \PT_2 \triangleright \Delta_2 $ such that $ \Delta = \Delta_1 \otimes \Delta_2 $.
			By the induction hypothesis, $ \Gamma \vdash \PT_1 \triangleright \Delta_1 $ and $ \PT_1 \Step \PT_1' $ imply that there is $ \Delta_1' $ such that $ \Gamma \vdash \PT_1' \triangleright \Delta_1' $, $ \Delta_1' $ is coherent, and $ \Delta_1 \Step^* \Delta_1' $.
			Because $ \Delta_1 \otimes \Delta_2 $ is defined and since Rule~\textsf{Com'} can only reduce local types, then $ \Delta' = \Delta_1' \otimes \Delta_2 $ is defined but not necessarily coherent.
			If $ \Delta' $ is not coherent then this is because of superfluous branches in receivers that we remove with Rule~\textsf{Cut}, while Rule~\textsf{Get} ensures that the validity of the type judgement is not affected by removing superfluous branches in the type.
			Let $ \Delta'' $ be the result of removing all superfluous branches from $ \Delta' $ such that $ \Delta'' $ is coherent and, with the Rules~\textsf{Com'} and \textsf{Cut}, $ \Delta \Step^* \Delta'' $.
			By Rule~\textsf{Par}, then $ \Gamma \vdash \PT' \triangleright \Delta'' $.
		\item[Case of Rule~\textsf{Res}:]
			In this case $ \PT = \PRes{\Chan}{\PT[Q]} $, $ \PT' = \PRes{\Chan}{\PT[Q]'} $, and $ \PT[Q] \Step \PT[Q]' $.
			By Figure~\ref{fig:typingRules}, then the proof of $ \Gamma \vdash \PT \triangleright \Delta $ starts with Rule~\textsf{Res} such that $ \Gamma = \Gamma', \GGlob{a}{G} $, $ \Delta = \Delta_{\PT[Q]}, \LInv{a}{1}, \ldots, \LInv{a}{n} $, $ \Delta_{\PT[Q]}, \LLoc{1}{\Proj{\GT}{\Role[1]}}, \ldots, \LLoc{n}{\Proj{\GT}{\Role[n]}} \Step^* \Delta_{\PT[Q]}'' $, $ \Roles{\GT} = \Set{ \Role[1], \ldots, \Role[n] } $, and $ \Gamma', \GGlobS{a}{G}{s} \vdash \PT[Q] \triangleright \Delta_{\PT[Q]}'' $.
			By the induction hypothesis, $ \Gamma', \GGlobS{a}{G}{s} \vdash \PT[Q] \triangleright \Delta_{\PT[Q]}'' $ and $ \PT[Q] \Step \PT[Q]' $ imply that there is $ \Delta_{\PT[Q]}''' $ such that $ \Gamma', \GGlobS{a}{G}{s} \vdash \PT[Q] \triangleright \Delta_{\PT[Q]}''' $, $ \Delta_{\PT[Q]}''' $ is coherent, and $ \Delta_{\PT[Q]}'' \Step^* \Delta_{\PT[Q]}''' $.
			By Rule~\textsf{Res}, then $ \Gamma \vdash \PT' \triangleright \Delta $.
			By the reflexivity of $ \Step^* $, then $ \Delta \Step^* \Delta $.
		\item[Case of Rule~\textsf{Struc}:] This case follows from Lemma~\ref{lem:structuralCongruence}.
	\end{description}
\end{proof}

In the above proof we use the assumption that $ \Delta $ is coherent only to prove that then also $ \Delta' $ is coherent.
The proof of subject reduction without coherence in Theorem~1 of \cite{petersWagnerNestmann19} is obtained as special case of the above proof by removing the parts about coherence and the step relation on session environments.

Session invitations ensure by definition, that the implementation of a session composes its actors in parallel.
If a process is \emph{role-distributed} then so are all its derivatives.

\begin{lemma}
	\label{lem:roleDist}
	If $ \Gamma \vdash \PT \triangleright \Delta $, $ \PT \Step \PT' $, and \PT is role-distributed then $ \PT' $ is role-distributed.
\end{lemma}

\begin{proof}
	The proof is by induction on the reduction rules of Figure~\ref{fig:reductionSemantics} to obtain $ \PT \Step \PT' $, since no reduction rule unifies parallel branches.
\end{proof}

If \PT is well-typed then all actions of the same actor are composed sequentially in \PT.

\begin{proof}[Proof of Lemma~\ref{lem:actorsAreSequential}]
	By the Definitions~\ref{def:localTypes} and \ref{def:projection}, local types are sequential and projection maps global types for each role on a single sequential local type \LT.
	By coherence, for each role there is initially either exactly one $ \LInv{a}{r} $ that is replaced later by the sequential local type \LT or there is exactly one $ \LLoc{r}{\LT} $ in the session environment.
	By Figure~\ref{fig:typingRules}, then \LT cannot be split between different parallel branches and communication prefixes require a corresponding assignment of the respective actor.
	Thus, a single local type \LT cannot be implemented in different parallel branches.
\end{proof}

Well-typedness ensures that there are no conflicts between communication prefixes for both session invitations and communications within sessions.

\begin{lemma}[Linearity]
	\label{lem:linearity}
	If \PT is well-typed then \PT contains, for each shared channel \Chan[a], at most one invitation $ \PReq{\Role[2]..\Role[n]}{\PT'} $ and no two $ \PAcc{\Role}{\PT'} $ for the same role~\Role and, for each pair of a session channel \Chan and a role~\Role, at most one sender $ \PSend{\Role}{\Role'}{\Label}{\tilde{\expr}}{\PT'} $ and at most one receiver $ \PGet{\Role}{\Role'}{\Set{ \PLab{\Label_i}{\tilde{\Args}_i}{\PT_i'} }_{i \in \indexSet}} $ that are guarded only by conditionals.
\end{lemma}

\begin{proof}
	Assume that \PT is well-typed.
	Assume that \PT contains an invitation $ \PReq{\Role[2]..\Role[n]}{\PT'} $ that is guarded only by conditionals.
	By Figure~\ref{fig:typingRules} and $ \Gamma \vdash \PT \triangleright \Delta $, then $ \LInv{a}{1} \in \Delta $.
	Because $ \LInv{a}{1} $ is consumed in Rule~\textsf{Req}, cannot be introduced by typing rules, and cannot be duplicated for different parallel branches, \PT can contain at most one such prefix.\\
	The case of $ \PAcc{\Role}{\PT'} $ is similar.\\
	The cases of $ \PSend{\Role}{\Role'}{\Label}{\tilde{\expr}}{\PT'} $ and $ \PGet{\Role}{\Role'}{\Set{ \PLab{\Label_i}{\tilde{\Args}_i}{\PT_i'} }_{i \in \indexSet}} $ follow from Lemma~\ref{lem:actorsAreSequential}.
\end{proof}

Moreover, well-typedness ensures that, for all unguarded communication prefixes, the respective communication partner exists in a parallel branch but might be guarded.

\begin{lemma}[Error-Freedom]
	\label{lem:errorFreedom}
	If \PT is well-typed and \PT contains an unguarded
	\begin{enumerate}
		\item $ \PReq{\Role[2]..\Role[n]}{\PT'} $ then \PT also contains $ \PAcc{\Role[i]}{\PT_i'} $ for all $ \Role[2] \leq \Role[i] \leq \Role[n] $ that are guarded only by conditionals or prefixes on channels different form \Chan[a].
		\item $ \PAcc{\Role}{\PT'} $ then there is some \Role[n] such that  $ \Role[2] \leq \Role \leq \Role[n] $ and \PT also contains $ \PReq{\Role[2]..\Role[n]}{\PT'} $ and $ \PAcc{\Role[i]}{\PT_i'} $ for all $ \Role[2] \leq \Role[i] \leq \Role[n] $ with $ \Role \neq \Role[i] $ that are guarded only by conditionals or prefixes on channels different form \Chan[a].
		\item $ \PSend{\Role}{\Role'}{\Label_j}{\tilde{\expr}}{\PT'} $ then \PT also contains $ \PGet{\Role}{\Role'}{\Set{ \PLab{\Label_i}{\tilde{\Args}_i}{\PT_i'} }_{i \in \indexSet}} $ with $ j \in \indexSet $ in parallel with the sender.
		\item $ \PGet{\Role}{\Role'}{\Set{ \PLab{\Label_i}{\tilde{\Args}_i}{\PT_i'} }_{i \in \indexSet}} $ then \PT also contains $ \PSend{\Role}{\Role'}{\Label_j}{\tilde{\expr}}{\PT'} $ with $ j \in \indexSet $ in parallel with the receiver.
	\end{enumerate}
\end{lemma}

\begin{proof}
	Assume that \PT is well-typed and that \PT contains an unguarded
	\begin{enumerate}
		\item invitation $ \PReq{\Role[2]..\Role[n]}{\PT'} $.
			By Figure~\ref{fig:typingRules} and $ \Gamma \vdash \PT \triangleright \Delta $, then $ \LInv{a}{1} \in \Delta $ and $ \GGlob{a}{G} \in \Gamma $ for some \GT such that $ \Roles{\GT} = \Set{ \Role[1], \ldots, \Role[n] } $.
			By coherence, then $ \LInv{a}{2}, \ldots, \LInv{a}{n} \in \Delta $.
			By Lemma~\ref{lem:actorsAreSequential} and Figure~\ref{fig:typingRules}, then \PT also contains $ \PAcc{\Role[i]}{\PT_i'} $ for all $ \Role[2] \leq \Role[i] \leq \Role[n] $ that are guarded only by conditionals or prefixes on channels different from \Chan[a].
		\item $ \PAcc{\Role}{\PT'} $.
			By Figure~\ref{fig:typingRules} and $ \Gamma \vdash \PT \triangleright \Delta $, then $ \LInv{a}{\Role} \in \Delta $ and $ \GGlob{a}{G} \in \Gamma $ for some \GT such that $ \Roles{\GT} = \Set{ \Role[1], \ldots, \Role[n] } $ and $ \Role[1] \leq \Role \leq \Role[n] $.
			By coherence, then $ \LInv{a}{1}, \ldots, \LInv{a}{n} \in \Delta $.
			By Lemma~\ref{lem:actorsAreSequential} and Figure~\ref{fig:typingRules}, then \PT also contains $ \PReq{\Role[2]..\Role[n]}{\PT'} $ and $ \PAcc{\Role[i]}{\PT_i'} $ for all $ \Role[2] \leq \Role[i] \leq \Role[n] $ with $ \Role \neq \Role[i] $ that are guarded only by conditionals or prefixes on channels different from \Chan[a].
		\item $ \PSend{\Role}{\Role'}{\Label}{\tilde{\expr}}{\PT'} $.
			By Figure~\ref{fig:typingRules} and $ \Gamma \vdash \PT \triangleright \Delta $, then $ \LLoc{\Role}{\LSend{\Role'}{\Set{ \LLab{\Label_j}{\tilde{\Sort}_j}{\LT_j'} }_{j \in \indexSet[J]}}} \in \Delta $.
			By coherence, then there is some \LT such that $ \LLoc{\Role'}{\LT} \in \Delta $ and $ \LT $ has a sub-type $ \LGet{\Role}{\Set{ \LLab{\Label_j}{\tilde{\Sort}_j}{\LT_j''} }_{j \in \indexSet[J]}} $.
			By Lemma~\ref{lem:actorsAreSequential} and Figure~\ref{fig:typingRules}, then \PT contains also $ \PGet{\Role}{\Role'}{\Set{ \PLab{\Label_i}{\tilde{\Args}_i}{\PT_i'} }_{i \in \indexSet}} $ with $ j \in \indexSet $ in parallel with the sender.
		\item $ \PGet{\Role}{\Role'}{\Set{ \PLab{\Label_i}{\tilde{\Args}_i}{\PT_i'} }_{i \in \indexSet}} $.
			By Figure~\ref{fig:typingRules} and $ \Gamma \vdash \PT \triangleright \Delta $, then there is some $ \indexSet[J] \subseteq \indexSet[I] $ such that $ \LLoc{\Role}{\LGet{\Role'}{\Set{ \LLab{\Label_j}{\tilde{\Sort}_j}{\LT_j'} }_{j \in \indexSet[J]}}} \in \Delta $.
			By coherence, then there is some \LT such that $ \LLoc{\Role'}{\LT} \in \Delta $ and $ \LT $ has a sub-type $ \LSend{\Role}{\Set{ \LLab{\Label_j}{\tilde{\Sort}_j}{\LT_j'} }_{j \in \indexSet[J]}} $.
			By Lemma~\ref{lem:actorsAreSequential} and Figure~\ref{fig:typingRules}, then \PT also contains $ \PSend{\Role}{\Role'}{\Label_j}{\tilde{\expr}}{\PT'} $ with $ j \in \indexSet $ in parallel with the receiver.
	\end{enumerate}
\end{proof}

Progress, \ie the absence of local deadlocks, is a simple consequence of the above results.

\begin{lemma}[Progress]
	\label{lem:progress}
	If $ \PT $ is well-typed and $ \PT \Steps \PT' $ then either $ \PT' \equiv \PEnd $ or there is some $ \PT'' $ such that $ \PT' \Step \PT'' $.
\end{lemma}

\begin{proof}
	Assume $ \Gamma \vdash \PT \triangleright \Delta $, $ \Delta $ is coherent, and $ \PT \Steps \PT' $.
	By Lemma~\ref{lem:subjectReduction}, then there is some $ \Delta' $ such that $ \Gamma \vdash \PT' \triangleright \Delta' $, $ \Delta' $ is coherent, and $ \Delta \Steps \Delta' $.
	If $ \Delta' = \emptyset $ then, by Figure~\ref{fig:typingRules}, $ \PT' \equiv \PEnd $.
	Else $ \Delta' $ contains an assignment of the form $ \LInv{a}{r} $ or of the form $ \LLoc{r}{\LT} $.
	By Figure~\ref{fig:typingRules}, then $ \PT $ contains a corresponding $ \PReq{\Role[2]..\Role[n]}{\PT[Q]} $, $ \PAcc{\Role}{\PT[Q]} $, $ \PSend{\Role}{\Role'}{\Label}{\tilde{\expr}}{\PT[Q]} $, or $ \PGet{\Role}{\Role'}{\Set{ \PLab{\Label_i}{\tilde{\Args}_i}{\PT[Q]_i'} }_{i \in \indexSet}} $ that is guarded only by conditionals.
	Let $ \PT_1 $ be a process such that the sequence $ \PT' \Steps \PT_1 $ resolves these conditionals and unguards the respective action prefix.
	By Lemma~\ref{lem:errorFreedom}, then there is a matching communication partner that is guarded only by conditionals.
	Let $ \PT_2 $ be a process such that the sequence $ \PT_1 \Steps \PT_2 $ resolves these conditionals and unguards the respective action prefix of the communication partner.
	Then $ \PT_2 \Step \PT'' $ is the step that reduces the respective communication.
\end{proof}


\section{Processes versus \SGP-Processes}
\label{sec:proofs}

In \cite{petersWagnerNestmann19} we introduce two algorithms.
The first allows to map systems that are well-formed \wrt to a synchronous global type for the case of a single session.

\begin{definition}
	\label{def:algorithm}
	The partial mapping \MapSGP{\Set{ \PT_i }_{i \in \indexSet}}{\GT} is defined inductively as:
	\begin{enumerate}
		\item \SEnd, \hfill if $ \GT = \GEnd $ \label{algo:GEnd}
		\item $ \SVar_{\TVar} $, \hfill else if $ \GT = \TVar $ \label{algo:TVar}
		\item $ \MapSGP{\Set{ \PT_j' } \cup \Set{ \PT_i }_{i \in \indexSet \setminus \Set{ j }}}{\GT} $, \newline \hphantom{h} \hfill else if there is some $ j \in \indexSet $ such that $ \PT_j = \PRes{\Chan}{\PT_j'} $ \label{algo:PRes}
		\item $ \MapSGP{\Set{ \PT_{j1}, \PT_{j2} } \cup \Set{ \PT_i }_{i \in \indexSet \setminus \Set{ j }}}{\GT} $, \newline \hphantom{h} \hfill else if there is some $ j \in \indexSet $ such that $ \PT_j = \PPar{\PT_{j1}}{\PT_{j2}} $ \label{algo:PPar}
		\item $ \MapSGP{\Set{ \PT_j'\Set{ \Subst{\PRep{\PVar}{\PT_j'}}{\PVar} } } \cup \Set{ \PT_i }_{i \in \indexSet \setminus \Set{ j }}}{\GT} $, \newline \hphantom{h} \hfill else if there is $ j \in \indexSet $ such that $ \PT_j = \PRep{\PVar}{\PT_j'} $ \label{algo:PRep}
		\item $ \SAssign{\tilde{\Args}_m@\Actor{s}{r}}{\tilde{\expr}}{\MapSGP{\Set{ \PT[Q]_m\Set{ \Subst{\tilde{\Args}_m@\Actor{s}{r}}{\tilde{\Args}_m} }, \PT[Q] } \cup \Set{ \PT_i }_{i \in \indexSet \setminus \Set{ k, l }}}{\GT_m}} $, \newline \hphantom{h} \hfill else if there are $ k, l \in \indexSet $, $ m \in \indexSet[J] \subseteq \indexSet[J]' $ such that \newline \hphantom{h} \hfill $ \GT = \GCom{\Role_1}{\Role_2}{\Set{ \GLab{\Label_j}{\tilde{\Sort}_j}{\GT_j} }_{j \in \indexSet[J]}} $, $ \PT_k = \PSend{\Role_1}{\Role_2}{\Label_m}{\tilde{\expr}}{\PT[Q]} $, \newline \hphantom{h} \hfill and $ \PT_l = \PGet{\Role_2}{\Role_1}{\Set{ \PLab{\Label_j}{\tilde{\Args}_j}{\PT[Q]_j} }_{j \in \indexSet[J]'}} $ \label{algo:GCom}
		\item $ \SPar{\MapSGP{\Set{ \PT_i }_{i \in \indexSet_1}}{\GT_1}}{\MapSGP{\Set{ \PT_j }_{j \in \indexSet_2}}{\GT_2}} $, \newline \hphantom{h} \hfill else if there are some $ \indexSet_1 \cup \indexSet_2 = \indexSet $ such that $ \GT = \GPar{\GT_1}{\GT_2} $, \newline \hphantom{h} \hfill $ \bigcup_{i \in \indexSet_1} \Roles{\PT_i} = \Roles{\GT_1} $, and $ \bigcup_{j \in \indexSet_2} \Roles{\PT_j} = \Roles{\GT_2} $ \label{algo:GPar}
		\item $ \SRep{\SVar_{\TVar}}{\MapSGP{\Set{ \PT_i }_{i \in \indexSet}}{\GT'}} $, \hfill else if $ \GT = \GRec{\TVar}{\GT'} $ \label{algo:GRec}
		\item $ \tau.\MapSGP{\Set{ \PT_1', \ldots, \PT_n' }}{\GT} $, \newline \hphantom{h} \hfill else if $ \Set{ \PT_i }_{i \in \indexSet} = \Set{ \PReq{\Role[2]..\Role[n]}{\PT_1'}, \PAcc{\Role[2]}{\PT_2'}, \ldots, \PAcc{\Role[n]}{\PT_n'} } $ \label{algo:PReq}
		\item $ \SCase{\cond}{\MapSGP{\Set{ \PT_{j1} } \cup \Set{ \PT_i }_{i \in \indexSet \setminus \Set{ j }}}{\GT}}{\MapSGP{\Set{ \PT_{j2} } \cup \Set{ \PT_i }_{i \in \indexSet \setminus \Set{ j }}}{\GT}} $, \newline \hphantom{h} \hfill else if there is some $ j \in \indexSet $ such that $ \PT_j = \PCond{\cond}{\PT_{j1}}{\PT_{j2}} $ \label{algo:PCond}
	\end{enumerate}
\end{definition}

The second mapping extends the first to asynchronous session types and multiple sessions.

\begin{definition}
	\label{def:algorithm2}
	\MapSGPI{\Set{ \PT_i }_{i \in \indexSet}}{\Set{ \left( \GT_j, \Chan_j \right) }_{j \in \indexSet[J]}} is defined inductively as:
	\vspace{-0.5em}
\begin{enumerate}
	\item
	\begin{enumerate}
		\item \SEnd, \hfill if $ \indexSet[J] = \emptyset $ \label{algo2:Empty}
		\item \MapSGPI{\Set{ \PT_i }_{i \in \indexSet}}{\Set{ \left( \GT_j, \Chan_j \right) }_{j \in \indexSet[J] \setminus \Set{ k }}}, \newline \hphantom{h} \hfill else if there is some $ k \in \indexSet[J] $ such that $ \GT_k = \GEnd $ \label{algo2:GEnd}
	\end{enumerate}
	\item $ \SVar_{\mathcal{G}} $, \hfill else if $ \mathcal{G} = \Set{ \GT_j }_{j \in \indexSet[J]} = \Set{ \TVar_j }_{j \in \indexSet[J]} $ \label{algo2:TVar}
	\item $ \MapSGPI{\Set{ \PT_k' } \cup \Set{ \PT_i }_{i \in \indexSet \setminus \Set{ k }}}{\Set{ \left( \GT_j, \Chan_j \right) }_{j \in \indexSet[J]}} $, \newline \hphantom{h} \hfill else if there is $ k \in \indexSet $ such that $ \PT_k = \PRes{\Chan}{\PT_k'} $ \label{algo2:PRes}
	\item $ \MapSGPI{\Set{ \PT_{k1}, \PT_{k2} } \cup \Set{ \PT_i }_{i \in \indexSet \setminus \Set{ k }}}{\Set{ \left( \GT_j, \Chan_j \right)}_{j \in \indexSet[J]}} $, \newline \hphantom{s} \hfill else if there is some $ k \in \indexSet $ such that $ \PT_k = \PPar{\PT_{k1}}{\PT_{k2}} $ \label{algo2:PPar}
	\item $ \MapSGPI{\Set{ \PT_k'\Set{ \Subst{\PRep{\PVar}{\PT_k'}}{\PVar} } } \cup \Set{ \PT_i }_{i \in \indexSet \setminus \Set{ k }}}{\Set{ \left( \GT_j, \Chan_j \right) }_{j \in \indexSet[J]}} $, \newline \hphantom{s} \hfill else if there is $ k \in \indexSet $ such that $ \PT_k = \PRep{\PVar}{\PT_k'} $ \label{algo2:PRep}
	\item $ \SAssign{\tilde{\Args}_n@\Actor{s}{r}}{\tilde{\expr}}{\MapSGPI{\mathcal{P}}{\mathcal{G}}} $ with $ \mathcal{P} = \Set{ \PT[Q]_n\Set{ \Subst{\tilde{\Args}_n@\Actor{s}{r}}{\tilde{\Args}_n} }, \PT[Q] } \cup \Set{ \PT_i }_{i \in \indexSet \setminus \Set{ m, o }} $ and $ \mathcal{G} = \Set{ \left( \GT_{l, n} \right) } \cup \Set{ \left( \GT_j, \Chan_j \right) }_{j \in \indexSet[j] \setminus \Set{l }} $, \newline \hphantom{h} \hfill else if there are $ m, o \in \indexSet $, $ l \in \indexSet[J] $, $ n \in \indexSet[K] \subseteq \indexSet[K]' $ such that \newline \hphantom{h} \hfill $ \GT_l = \GCom{\Role_1}{\Role_2}{\Set{ \GLab{\Label_k}{\tilde{\Sort}_k}{\GT_{l, k}} }_{k \in \indexSet[K]}} $, $ \PT_m = \PSend{\Role_1}{\Role_2}{\Label_n}{\tilde{\expr}}{\PT[Q]} $, \newline \hphantom{h} \hfill and $ \PT_o = \PGet{\Role_2}{\Role_1}{\Set{ \PLab{\Label_k}{\tilde{\Args}_k}{\PT[Q]_k} }_{k \in \indexSet[K]'}} $ \label{algo2:GCom}
	\item
	\begin{enumerate}
		\item $ \SPar{\MapSGPI{\Set{ \PT_i }_{i \in \indexSet_1}}{\Set{ \left( \GT_j, \Chan_j \right) }_{j \in \indexSet[J]_1}}}{\MapSGPI{\Set{ \PT_i }_{i \in \indexSet_2}}{\Set{ \left( \GT_j, \Chan_j \right) }_{j \in \indexSet[J]_2}}} $, \newline \hphantom{s} \hfill else if there are some $ \indexSet_1 \cup \indexSet_2 = \indexSet $, $ \indexSet[J]_1 \cup \indexSet[J]_2 = \indexSet[J] $ such that $ \indexSet[J]_1 \cap \indexSet[J]_2 = \emptyset $ and \newline \hphantom{h} \hfill $ \bigcup_{i \in \indexSet_k} \Actors{\PT_i} = \Set{ \Actor{s_j}{r} \mid j \in \indexSet[J]_k \land \Role \in \Roles{\GT_j} } $ for $ k \in \Set{ 1, 2 } $ \label{algo2:split}
		\item $ \MapSGPI{\Set{ \PT_i }_{i \in \indexSet}}{\Set{ \left( \GT_{k1}, \Chan_k \right), \left( \GT_{k2}, \Chan_k \right) } \cup \Set{ \left( \GT_j, \Chan_j \right) }_{j \in \indexSet[J]}} $, \newline \hphantom{s} \hfill else if there is $ k \in \indexSet[J] $ such that $ \GT_k = \GPar{\GT_{k1}}{\GT_{k2}} $ \label{algo2:GPar}
	\end{enumerate}
	\item $ \SRep{\SVar_{\mathcal{G}}}{\MapSGPI{\Set{ \PT_i }_{i \in \indexSet}}{\Set{ \left( \GT_j', \Chan_j \right) }_{j \in \indexSet[J]}}} $, \newline \hphantom{s} \hfill else if $ \GT_j = \GRec{\TVar_j}{\GT_j'} $ for all $ j \in \indexSet[J] $ and $ \mathcal{G} = \Set{ \TVar_j }_{j \in \indexSet[J]} $ \label{algo2:GRec}
	\item $ \tau.\MapSGPI{\Set{ \PT_1', \ldots, \PT_n' } \cup \Set{ \PT_i }_{i \in \indexSet \setminus \Set{ k1, \ldots, kn }}}{\Set{ \left( \GT_j, \Chan_j \right) }_{j \in \indexSet[J]}} $, \newline \hphantom{s} \hfill else if there are $ k1, \ldots, kn \in \indexSet $ such that $ \PT_{k1} = \PReq{\Role[2]..\Role[n]}{\PT_1'} $, \newline \hphantom{h} \hfill $ \PT_{k2} = \PAcc{\Role[2]}{\PT_2'} $, \ldots, $ \PT_{kn} = \PAcc{\Role[n]}{\PT_n'} $ \label{algo2:PReq}
	\item $ \SCase{\cond}{\MapSGPI{\Set{ \PT_{k1} } \cup \mathcal{P}}{\mathcal{G}}}{\MapSGPI{\Set{ \PT_{k2} } \cup \mathcal{P}}{\mathcal{G}}} $ \newline with $ \mathcal{P} = \Set{ \PT_i }_{i \in \indexSet \setminus \Set{ k }} $ and $ \mathcal{G} = \Set{ \left( \GT_j, \Chan_j \right) }_{j \in \indexSet[J]} $, \newline \hphantom{s} \hfill else if there is $ k \in \indexSet $ such that $ \PT_k = \PCond{\cond}{\PT_{k1}}{\PT_{k2}} $ \label{algo2:PCond}
\end{enumerate}
\end{definition}

We observe that each of the Cases~\ref{algo2:Empty}, \ref{algo2:GEnd}, \ref{algo2:TVar}--\ref{algo2:PPar}, \ref{algo2:GCom}, \ref{algo2:GPar}, and \ref{algo2:GRec}--\ref{algo2:PReq} reduces either the set of global types or the considered set of processes.
By unfolding recursion, Case~\ref{algo2:PRep} blows up one of the considered processes.
This is necessary, because the typing system allows that a process and its global type do not loop at the same points but in the same way.
Because of that, the global type $ \GRec{\TVar}{\left( \GCom{\Role_1}{\Role_2}{\GLab{\Label}{\Sort}{\GCom{\Role_2}{\Role_1}{\GLab{\Label}{\Sort}{\TVar}}}} \right)} $ can \eg be implemented for role~$ \Role_1 $ by the process $ \PReq{\Role_2}{\PSend{\Role_1}{\Role_2}{\Label}{\Args[z]}{\PRep{\PVar}{\left( \PGet{\Role_1}{\Role_2}{\PLab{\Label}{\Args}{\PSend{\Role_1}{\Role_2}{\Label}{\Args[z]}{\PVar}}} \right)}}} $.
Case~\ref{algo2:PRep} allows to unfold recursion in processes until the current recursive set of global types is reduced to recursion variables in Case~\ref{algo2:TVar}.
Note that Case~\ref{algo2:TVar} drops the remainder of the process as soon as the loops in the global types are reduced.
With that, the number of unfoldings of recursion in processes in Case~\ref{algo2:PRep} that is necessary to compute \MapSGPI{\Set{ \PT }}{\Set{ \left( \GT_j, \Chan_j \right) }_{j \in \indexSet[J]}} is bounded by the size of the loops in the global types.

Case~\ref{algo2:split} introduces a parallel composition in the \SGP-process if the considered sets of processes can be partitioned into two sets that implement the actors of different sessions.
This case can be applied if we can split the set of sessions into two disjoint sets such that there are no dependencies between the sessions in different sets.
Since the number of sessions is bounded, where session invitations under recursion introduce only a single session \wrt well-typedness, also the number of applications of Case~\ref{algo2:split} that are necessary to compute \MapSGPI{\Set{ \PT }}{\Set{ \left( \GT_j, \Chan_j \right) }_{j \in \indexSet[J]}} is bounded.

Case~\ref{algo2:PCond} globalizes conditionals and therefore copies all other actors to both cases.
Because of that, we apply this case only if it is necessary to unguard a communication partner.
In well-typed systems both cases of a conditional need to follow the same type, \ie need to implement the same communication structure.
Conditionals are a local form of branching, \ie implement alternative behaviours of a single actor.
Only by transmitting information about the outcome of a conditional as sender in a communication can a local conditional influence the remaining actors of the system.
Accordingly, conditionals are usually used to choose between alternatives directly before sending in process implementations (compare to Example~\ref{exa:implementationAuctioneer}) or process implementations can easily be optimized to satisfy this property.
By design, the algorithms in Definition~\ref{def:algorithm} and \ref{def:algorithm2} map such a local conditional only if this is necessary to unguard a communication partner.
Since we give precedence to conditionals that guard senders, we indeed map only necessary conditionals.
Nonetheless, since Case~\ref{algo2:PCond} copies the remaining processes and the type(s), we cannot avoid a blow-up of the size of the generated system in this case that is in the worst case exponentially larger than the original system.
However, if a conditional guides the choice between different branches of an immediately following send-action, then all of the copied actors that are input guarded in the global type will reduce to different cases for the respective two branches.
This explains why the size of our toy example does not grow then we map the conditionals in Example~\ref{exa:implementationAuctioneer} to the \SGP-process in Example~\ref{exa:SGPAuctionieer} in Section~\ref{sec:toy-example}.
The copies of processes in Case~\ref{algo2:PCond} increase the size of the resulting \SGP-process---in comparison to the original system---only with respect to conditionals that do not implement a choice between different labels of a sender or with respect to actors that are in their next step not influenced by the outcome of this conditional as in the type.
\begin{example}
	As example consider the process
	\begin{align*}
		\PT ={}
		& \PReq{2,3,4}{\PCond{\Args > 5}{\PSend{1}{2}{\Label[y]}{42}{\PEnd}}{\PSend{1}{2}{\Label[n]}{0}{\PEnd}}}\\
		\mid \; & \PAcc{2}{\PGet{2}{1}{\Set{ \PLab{\Label[y]}{\Args}{\PEnd}, \quad \PLab{\Label[n]}{\Args}{\PEnd} }}}\\
		\mid \; & \PAcc{3}{\PCond{\Args > 5}{\PSend{3}{4}{\Label[y]}{42}{\PEnd}}{\PSend{3}{4}{\Label[n]}{0}{\PEnd}}}\\
		\mid \; & \PAcc{4}{\PGet{4}{3}{\Set{ \PLab{\Label[y]}{\Args}{\PEnd}, \quad \PLab{\Label[n]}{\Args}{\PEnd} }}}
	\end{align*}
	that is well-typed with respect to
	\begin{align*}
		\GT = \GCom{1}{2}{\begin{array}[t]{l}
				\{ \GLab{\Label[y]}{\Sort[Int]}{\GCom{3}{4}{\Set{\GLab{\Label[y]}{\Sort[Int]}{\GEnd},\quad \GLab{\Label[n]}{\Sort[Int]}{\GEnd}}}},\\
				\hphantom{\{} \GLab{\Label[n]}{\Sort[Int]}{\GCom{3}{4}{\Set{\GLab{\Label[y]}{\Sort[Int]}{\GEnd},\quad \GLab{\Label[n]}{\Sort[Int]}{\GEnd}}}} \}
			\end{array}}
	\end{align*}
	Since the actors $ \Actor{\Chan[s]}{3} $ and $ \Actor{\Chan[s]}{4} $ act independent of the outcome of the conditional of $ \Actor{\Chan[s]}{1} $, the algorithm in Definition~\ref{def:algorithm} copies the translation of the actors $ \Actor{\Chan[s]}{3} $ and $ \Actor{\Chan[s]}{4} $:
	\begin{align*}
		\MapSGP{\PT}{\GT} = \tau.\mathsf{if} \; \Args_{1} > 5
			\begin{array}[t]{l}
				\mathsf{then} \; \SAssign{\Args_2}{42}{\SCase{\Args_3 > 5}{\SAssign{\Args_4}{42}{\SEnd}}{\SAssign{\Args_4}{0}{\SEnd}}}\\
				\mathsf{else} \; \SAssign{\Args_2}{0}{\SCase{\Args_3 > 5}{\SAssign{\Args_4}{42}{\SEnd}}{\SAssign{\Args_4}{0}{\SEnd}}}
			\end{array}
	\end{align*}
\end{example}

However, we observe that in this case the duplication of the behaviour of the actors $ \Actor{\Chan[s]}{3} $ and $ \Actor{\Chan[s]}{4} $ is already visible in the type.
So, if conditionals are used only to guide the choice between labels of an immediately following send-action, then again the corresponding increase of the size of the system in the algorithm is bounded by the size of the global types.

Finally, we observe that the size of a well-typed process is larger or equal to the sum of the sizes of its global types.
We conclude that---except for the conditionals---the algorithm in Definition~\ref{def:algorithm2}, takes a linear amount of steps and, thus, constructs a \SGP-process of a size that is linear \wrt the size of the original system.

\begin{corollary}
	\label{col:size}
	Let \PT be well-typed \wrt $ \Set{ \left( \GT_j, \Chan_j \right) }_{j \in \indexSet[J]} $.
	Assume that \PT uses conditionals only to branch between alternative labels of a sender.
	Then the computation of the \SGP-process \MapSGPI{\Set{ \PT }}{\Set{ \left( \GT_j, \Chan_j \right) }_{j \in \indexSet[J]}} is linear in the size of \PT combined with the sum of the sizes of the types in $ \Set{ \left( \GT_j, \Chan_j \right) }_{j \in \indexSet[J]} $ and produces a \SGP-process that is linear in this size.
\end{corollary}

This is very important.
We map well-typed systems onto \SGP-systems in order to avoid the problem of state space explosion that is caused by the concurrency of communication attempts, \ie to avoid the in the worst case exponential blow-up of states that need to be considered in verification.
\MPST are a very efficient method to analyse the communication structure of the original system.
Corollary~\ref{col:size} ensures that also the computation of the \SGP-system is efficient, \ie fast, and that the construction does not suffer from the problem of state space explosion, \ie the generated \SGP-system is not considerably larger than the original system.
Since the construction sequentialises the original system and thereby removes all forms of interaction and restriction, the verification of the \SGP-abstraction is much easier than the verification of the original system.

For the remainder of this section we assume that \textbf{no alpha conversion is used to rename input binders}.

Before we prove Theorem~2 of \cite{petersWagnerNestmann19}, \ie that the mapping of Definition~\ref{def:algorithm2} returns a \SGP-process whenever it is applied on a well-typed process and its global types, we show some properties on the different cases of Definition~\ref{def:algorithm2} (and implicitly also Definition~\ref{def:algorithm}).
In particular we show that most of the cases preserve well-typedness, \ie if their input is well-typed then so are the inputs of its recursive calls.

The first case replaces the set of considered processes by \SEnd if the global type is terminated and removes empty global types.
This is safe, because processes that are well-typed \wrt \GEnd cannot contain communication prefixes.

\begin{lemma}[Case~\ref{algo2:Empty}]
	\label{lem:algo2Empty}
	If \PT is well-typed \wrt \emph{$ \emptyset $} then \PT contains only parallel compositions, conditionals, successful termination, and restriction.
\end{lemma}

\begin{proof}
	Assume that \PT is well-typed \wrt $ \emptyset $, \ie there are $ \Gamma, \Delta $ such that \PT is role-distributed, $ \Gamma \vdash \PT \triangleright \Delta $, there are no global types in $ \Gamma $, $ \Delta $ is coherent \wrt $ \emptyset $, and $ \Gamma \vdash \PT \triangleright \Delta $ is globally progressing.
	Since there are no global types in $ \Gamma $, the derivation of $ \Gamma \vdash \PT \triangleright \Delta $ cannot use the Rules~\textsf{Req} or \textsf{Acc} and, thus, \PT cannot contain communication on shared channels.
	Because the other typing rules of Figure~\ref{fig:typingRules} can only reduce local types, neither the Rule~\textsf{Send} nor Rule~\textsf{Get} can be used and, thus, \PT does not contain prefixes for sending or receiving within sessions.
	Because of that, \PT cannot contain process variables and, thus, no recursion.
\end{proof}

\begin{lemma}[Case~\ref{algo2:GEnd}]
	\label{lem:algo2GEnd}
	If \PT is well-typed \wrt $ \Set{ \left( \GT_j, \Chan_j \right) }_{j \in \indexSet[J]} $, $ k \in \indexSet[J] $, and \emph{$ \GT_k = \GEnd $} then \PT contains no communication prefixes on $ \Chan_k $ and cannot invite the session $ \Chan_k $.
\end{lemma}

\begin{proof}
	Assume that \PT is well-typed \wrt $ \Set{ \left( \GT_j, \Chan_j \right) }_{j \in \indexSet[J]} $, $ k \in \indexSet[J] $, and $ \GT_k = \GEnd $, \ie there are $ \Gamma, \Delta $ such that \PT is role-distributed, $ \Gamma \vdash \PT \triangleright \Delta $, $ \Set{ \GT_j }_{j \in \indexSet[J]} $ are the global types in $ \Gamma $, $ \Delta $ is coherent \wrt $ \Set{ \left( \GT_j, \Chan[n]_j \right) }_{j \in \indexSet[J]} $, $ \Gamma \vdash \PT \triangleright \Delta $ is globally progressing, and for all $ j \in \indexSet[J] $ either $ \Chan[n]_j = \Chan_j $ or $ \Gamma \vdash \PT \triangleright \Delta $ connects $ \Chan[n]_j $ with $ \Chan_j $.
	Because $ \Roles{\GEnd} = \emptyset $ and $ \Delta $ is coherent \wrt $ \Set{ \left( \GT_j, \Chan[n]_j \right) }_{j \in \indexSet[J]} $, the session environment does not contain the (shared or session) channel $ \Chan[n]_k $, \ie $ \Chan[n]_k \notin \Delta $.
	Then, the derivation of $ \Gamma \vdash \PT \triangleright \Delta $ cannot use the Rules~\textsf{Req} or \textsf{Acc} and, thus, \PT cannot contain communication on $ \Chan[n]_k $.
	Hence, the session $ \Chan_k $ cannot be invited.
	Because the other typing rules of Figure~\ref{fig:typingRules} can only reduce local types, neither the Rule~\textsf{Send} nor Rule~\textsf{Get} can be used and, thus, \PT does not contain prefixes for sending or receiving within the session $ \Chan_k $.
\end{proof}

Since \PT is well-typed \wrt \GEnd is a special case of the conditions \PT is well-typed \wrt $ \Set{ \left( \GT_j, \Chan_j \right) }_{j \in \indexSet[J]} $, $ k \in \indexSet[J] $, and $ \GT_k = \GEnd $, Lemma~\ref{lem:algo2GEnd} holds also for Case~\ref{algo:GEnd} of Definition~\ref{def:algorithm}.

Note that the restriction of session channels without communication is always useless, \ie can be removed modulo structural congruence.
The Cases~\ref{algo2:PRes}, \ref{algo2:PPar}, \ref{algo2:PRep}, and \ref{algo2:PCond} are used to decompose and unfold processes to make them accessible for the other cases.
In all of these cases all recursive calls of the mapping are on sets of processes that---combined by parallel composition---are well-typed to the former global types.
Thus, none of these cases allows the mapping to reduce the global types or to create any \SGP-operators except for conditionals that are not reflected in global types.
Instead they can be seen as preparation cases.

Case~\ref{algo2:PRes} removes restriction, but preserves well-typedness \wrt the same global types in its recursive call.

\begin{lemma}[Case~\ref{algo:PRes}]
	\label{lem:algo2PRes}
	If $ \prod_{i \in \indexSet}{\PT_i} $ is well-typed \wrt $ \Set{ \left( \GT_j, \Chan_j \right) }_{j \in \indexSet[J]} $ and there is some $ k \in \indexSet $ such that $ \PT_k = \PRes{\Chan}{\PT_k'} $ then $ \left( \PPar{\PT_k'}{\left( \prod_{i \in \indexSet \setminus \Set{ k }}{\PT_i} \right)} \right) $ is well-typed \wrt $ \Set{ \left( \GT_j, \Chan_j \right) }_{j \in \indexSet[J]} $.
\end{lemma}

\begin{proof}
	Assume that $ \PT = \prod_{i \in \indexSet}{\PT_i} $ is well-typed \wrt $ \Set{ \left( \GT_j, \Chan_j \right) }_{j \in \indexSet[J]} $ and there is some $ k \in \indexSet $ such that $ \PT_k = \PRes{\Chan}{\PT_k'} $, \ie there are $ \Gamma, \Delta $ such that \PT is role-distributed, $  \Set{ \GT_j }_{j \in \indexSet[J]} $ are the global types in $ \Gamma $, $ \Delta $ is coherent \wrt $ \Set{ \left( \GT_j, \Chan[n]_j \right) }_{j \in \indexSet[J]} $, $ \Gamma \vdash \PT \triangleright \Delta $, \PT is globally progressing, and for all $ j \in \indexSet[J] $ either $ \Chan[n]_j = \Chan_j $ or $ \Gamma \vdash \PT \triangleright \Delta $ connects $ \Chan[n]_j $ with $ \Chan_j $.
	By Figure~\ref{fig:typingRules} and coherence, then there is some $ l \in \indexSet[J] $ such that $ \Chan[n]_l $ is a shared channel, $ \Set{ \LInv{a}{r} \mid \Role \in \Roles{\GT} } \subseteq \Delta $, and $ \GGlob{n}[l]{G}[l] \in \Gamma $.
	Then the derivation of $ \Gamma \vdash \PT \triangleright \Delta $ starts with Rule~\textsf{Par} to separate the judgement into $ \Gamma \vdash \PRes{\Chan}{\PT_j'} \triangleright \Delta_j $ and the judgements for the $ \Gamma \vdash \PT_i \triangleright \Delta_i $ with $ i \neq j $.
	Let $ \Gamma = \Gamma', \GGlob{n}[l]{G}[l] $.
	Since $ \Chan[n]_l $ is not relevant for the derivations $ \Gamma \vdash \PT_i \triangleright \Delta_i $ with $ i \neq j $, we have $ \Gamma', \GGlobS{n}[l]{G}[l]{s}[l] \vdash \PT_i \triangleright \Delta_i $ with $ i \neq j $.
	From $ \Gamma \vdash \PRes{\Chan}{\PT_j'} \triangleright \Delta_j $ and the typing rules, we get $ \Gamma', \GGlobS{n}[l]{G}[l]{s}[l] \vdash \PT_j' \triangleright \Delta_j' $, where $ \Delta_j' = \Set{ \LLoc{r}{\Proj{\GT_l}{\Role}}[l] \mid \Role \in \Roles{\GT_l} } $.
	By Rule~\textsf{Par}, then $ \Gamma', \GGlobS{n}[l]{G}[l]{s}[l] \vdash \PPar{\PT_k'}{\left( \prod_{i \in \indexSet \setminus \Set{ k }}{\PT_i} \right)} \triangleright \Delta_j', \bigcup_{i \neq j} \Delta_i $.
	By coherence, $ \Delta_j', \bigcup_{i \neq j} \Delta_i $ is coherent \wrt $ \Set{ \left( \GT_j, \Chan[n]_j \right) }_{j \in \indexSet[J], j \neq l} \cup \Set{ \left( \GT_l, \Chan_l \right) } $.
	Then, $ \left( \PPar{\PT_k'}{\left( \prod_{i \in \indexSet \setminus \Set{ k }}{\PT_i} \right)} \right) $ is well-typed \wrt $ \Set{ \left( \GT_j, \Chan_j \right) }_{j \in \indexSet[J]} $.
\end{proof}

Case~\ref{algo2:PPar} splits parallel composition and preserves well-typedness \wrt the same global types in its recursive call.

\begin{lemma}[Case~\ref{algo2:PPar}]
	\label{lem:algo2PPar}
	If $ \prod_{i \in \indexSet}{\PT_i} $ is well-typed \wrt $ \Set{ \left( \GT_j, \Chan_j \right) }_{j \in \indexSet[J]} $ and there is some $ k \in \indexSet $ such that $ \PT_k = \PPar{\PT_{k1}}{\PT_{k2}} $ then $ \left( \PPar{\PPar{\PT_{k1}}{\PT_{k2}}}{\left( \prod_{i \in \indexSet \setminus \Set{ k }}{\PT_i} \right)} \right) $ is well-typed \wrt $ \Set{ \left( \GT_j, \Chan_j \right) }_{j \in \indexSet[J]} $.
\end{lemma}

\begin{proof}
	Follows from the typing rules in Figure~\ref{fig:typingRules} and Rule~\textsf{Par} in particular.
\end{proof}

Case~\ref{algo2:PRep} unfolds recursion in a process and preserves well-typedness \wrt the same global type in its recursive call.

\begin{lemma}[Case~\ref{algo2:PRep}]
	\label{lem:algo2PRep}
	If $ \prod_{i \in \indexSet}{\PT_i} $ is well-typed \wrt $ \Set{ \left( \GT_j, \Chan_j \right) }_{j \in \indexSet[J]} $ and there is some $ k \in \indexSet $ such that $ \PT_k = \PRep{\PVar}{\PT_k'} $ then $ \left( \PPar{\PT_k'\Set{ \Subst{\PRep{\PVar}{\PT_k'}}{\PVar} }}{\left( \prod_{i \in \indexSet \setminus \Set{ k }}{\PT_i} \right)} \right) $ is well-typed \wrt $ \Set{ \left( \GT_j, \Chan_j \right) }_{j \in \indexSet[J]} $.
\end{lemma}

\begin{proof}
	Assume that $ \PT = \prod_{i \in \indexSet}{\PT_i} $ is well-typed \wrt $ \Set{ \left( \GT_j, \Chan_j \right) }_{j \in \indexSet[J]} $ and there is some $ k \in \indexSet $ such that $ \PT_k = \PRep{\PVar}{\PT_k'} $, \ie there are $ \Gamma, \Delta $ such that \PT is role-distributed, $  \Set{ \GT_j }_{j \in \indexSet[J]} $ are the global types in $ \Gamma $, $ \Delta $ is coherent \wrt $ \Set{ \left( \GT_j, \Chan[n]_j \right) }_{j \in \indexSet[J]} $, $ \Gamma \vdash \PT \triangleright \Delta $, \PT is globally progressing, and for all $ j \in \indexSet[J] $ either $ \Chan[n]_j = \Chan_j $ or $ \Gamma \vdash \PT \triangleright \Delta $ connects $ \Chan[n]_j $ with $ \Chan_j $.
	By Lemma~\ref{lem:structuralCongruence}, then $ \Gamma \vdash \PPar{\PT_j'\Set{ \Subst{\PRep{\PVar}{\PT_j'}}{\PVar} }}{\left( \prod_{i \in \indexSet \setminus \Set{ j }}{\PT_i} \right)} \triangleright \Delta $.
	Then, $ \left( \PPar{\PT_j'\Set{ \Subst{\PRep{\PVar}{\PT_j'}}{\PVar} }}{\left( \prod_{i \in \indexSet \setminus \Set{ j }}{\PT_i} \right)} \right) $ is well-typed \wrt $ \Set{ \left( \GT_j, \Chan_j \right) }_{j \in \indexSet[J]} $.
\end{proof}

Case~\ref{algo2:PCond} maps a conditional of the original system on a \SGP-conditional and preserves well-typedness \wrt the same global types in both of its recursive calls.

\begin{lemma}[Case~\ref{algo2:PCond}]
	\label{lem:algo2PCond}
	If $ \prod_{i \in \indexSet}{\PT_i} $ is well-typed \wrt $ \Set{ \left( \GT_j, \Chan_j \right) }_{j \in \indexSet[J]} $ and there is some $ k \in \indexSet $ such that $ \PT_k = \PCond{\cond}{\PT_{k1}}{\PT_{k2}} $ then $ \left( \PPar{\PT_{k1}}{\left( \prod_{i \in \indexSet \setminus \Set{ k }}{\PT_i} \right)} \right) $ is well-typed \wrt $ \Set{ \left( \GT_j, \Chan_j \right) }_{j \in \indexSet[J]} $ and $ \left( \PPar{\PT_{k2}}{\left( \prod_{i \in \indexSet \setminus \Set{ k }}{\PT_i} \right)} \right) $ is well-typed \wrt $ \Set{ \left( \GT_j, \Chan_j \right) }_{j \in \indexSet[J]} $.
\end{lemma}

\begin{proof}
	Assume that $ \PT = \prod_{i \in \indexSet}{\PT_i} $ is well-typed \wrt $ \Set{ \left( \GT_j, \Chan_j \right) }_{j \in \indexSet[J]} $ and there is some $ k \in \indexSet $ such that $ \PT_k = \PCond{\cond}{\PT_{k1}}{\PT_{k2}} $, \ie there are $ \Gamma, \Delta $ such that \PT is role-distributed, $  \Set{ \GT_j }_{j \in \indexSet[J]} $ are the global types in $ \Gamma $, $ \Delta $ is coherent \wrt $ \Set{ \left( \GT_j, \Chan[n]_j \right) }_{j \in \indexSet[J]} $, $ \Gamma \vdash \PT \triangleright \Delta $, \PT is globally progressing, and for all $ j \in \indexSet[J] $ either $ \Chan[n]_j = \Chan_j $ or $ \Gamma \vdash \PT \triangleright \Delta $ connects $ \Chan[n]_j $ with $ \Chan_j $.
	By Figure~\ref{fig:typingRules}, the derivation of $ \Gamma \vdash \PT \triangleright \Delta $ starts with some applications of Rule~\textsf{Par} that split the judgement into $ \Gamma \vdash \PT_k \triangleright \Delta_k $ and $ \Gamma \vdash \PT_i \triangleright \Delta_i $ for all $ i \neq k $ such that $ \Delta $ is the disjoint union of $ \Delta_k $ and all $ \Delta_i $.
	By Rule~\textsf{Cond}, then $ \Gamma \vdash \PT_{k1} \triangleright \Delta_k $ and $ \Gamma \vdash \PT_{k2} \triangleright \Delta_k $.
	By Rule~\textsf{Par}, then $ \Gamma \vdash \PPar{\PT_{k1}}{\left( \prod_{i \in \indexSet \setminus \Set{ k }}{\PT_i} \right)} \triangleright \Delta $ and $ \Gamma \vdash \PPar{\PT_{k2}}{\left( \prod_{i \in \indexSet \setminus \Set{ k }}{\PT_i} \right)} \triangleright \Delta $.
	Then $ \left( \PPar{\PT_{k1}}{\left( \prod_{i \in \indexSet \setminus \Set{ k }}{\PT_i} \right)} \right) $ is well-typed \wrt $ \Set{ \left( \GT_j, \Chan_j \right) }_{j \in \indexSet[J]} $ and $ \left( \PPar{\PT_{k2}}{\left( \prod_{i \in \indexSet \setminus \Set{ k }}{\PT_i} \right)} \right) $ is well-typed \wrt $ \Set{ \left( \GT_j, \Chan_j \right) }_{j \in \indexSet[J]} $.
\end{proof}

If a process is well-typed \wrt a set of types containing a communication guarded global type and that communication guard is according to the interaction type system of \cite{BettiniAtall08} not dependent on another session, then it contains a corresponding sender and receiver that are guarded only by conditionals.
Case~\ref{algo:GCom} preserves well-typedness but may introduce superfluous input branches that are not matched by the global type $ \GT_n $ of the continuation of this communication guard.
Because of Rule~\textsf{Get}, the type system abstracts from such superfluous branches of receivers.

\begin{lemma}[Case~\ref{algo2:GCom}]
	\label{lem:algo2GCom}
	If $ \prod_{i \in \indexSet}{\PT_i} $ is well-typed \wrt $ \Set{ \left( \GT_j, \Chan_j \right) }_{j \in \indexSet[J]} $, none of the Cases~\ref{algo2:PRes}, \ref{algo2:PPar}, or \ref{algo2:PRep} can be applied, there is $ l \in \indexSet[J] $ such that the session $ \Chan_l $ is initialised, $ \GT_l = \GCom{\Role_1}{\Role_2}{\Set{ \GLab{\Label_k}{\tilde{\Sort}_k}{\GT_{l, k}} }_{k \in \indexSet[K]'}} $, and this communication does not depend on another session, then there are $ m, o \in \indexSet $, $ n \in \indexSet[K] $ and $ \indexSet[K] \subseteq \indexSet[K]' $ such that every conditional branch of $ \PT_m $ is a version of $ \PSend{\Role_1}{\Role_2}{\Label_n}{\tilde{\expr}}{\PT[Q]} $, every conditional branch of $ \PT_o $ is a version of $ \PGet{\Role_2}{\Role_1}{\Set{ \PLab{\Label_k}{\tilde{\Args}_k}{\PT[Q]_k} }_{k \in \indexSet[K]'}} $,
	and $ \left( \PPar{\PPar{\PT[Q]_n\Set{ \Subst{\tilde{\Args}_n@\Actor{s}{r}}{\tilde{\Args}_n} }}{\PT[Q]}}{\left( \prod_{i \in \indexSet \setminus \Set{ m, o }}{\PT_i} \right)} \right) $ is well-typed \wrt $ \Set{ \left( \GT_j, \Chan_j \right) }_{j \in \indexSet[J] \setminus \Set{ l }} \cup \Set{ \left( \GT_{l, k},  \Chan_l \right) } $.
\end{lemma}

\begin{proof}
	Assume that $ \PT = \prod_{i \in \indexSet}{\PT_i} $ is well-typed \wrt $ \Set{ \left( \GT_j, \Chan_j \right) }_{j \in \indexSet[J]} $, none of the Cases~\ref{algo2:PRes}, \ref{algo2:PPar}, or \ref{algo2:PRep} can be applied, there is $ l \in \indexSet[J] $ such that the session $ \Chan_l $ is initialised, $ \GT_l = \GCom{\Role_1}{\Role_2}{\Set{ \GLab{\Label_k}{\tilde{\Sort}_k}{\GT_{l, k}} }_{k \in \indexSet[K]'}} $, and this communication does not depend on another session, \ie there are $ \Gamma, \Delta $ such that \PT is role-distributed, $  \Set{ \GT_j }_{j \in \indexSet[J]} $ are the global types in $ \Gamma $, $ \Delta $ is coherent \wrt $ \Set{ \left( \GT_j, \Chan[n]_j \right) }_{j \in \indexSet[J]} $, $ \Gamma \vdash \PT \triangleright \Delta $, \PT is globally progressing, and for all $ j \in \indexSet[J] $ either $ \Chan[n]_j = \Chan_j $ or $ \Gamma \vdash \PT \triangleright \Delta $ connects $ \Chan[n]_j $ with $ \Chan_j $.
	By Figure~\ref{fig:typingRules} and coherence, then there is \Chan[a] such that $ \Set{ \LLoc{r}{\Proj{\GT_l}{\Role}} \mid \Role \in \Roles{\GT_l} } \subseteq \Delta $ and $ \GGlobS{a}{G}[l]{s}[l] \in \Gamma $ or $ \GGlob{a}{G}[l] \in \Gamma $, where we have $ \LLoc{r_1}{\LSend{\Role_2}{\Set{ \LLab{\Label_k}{\tilde{\Sort}_k}{\left( \Proj{\GT_{l, k}}{\Role_1} \right)} }_{k \in \indexSet[K]}}}[l] $ for role~$ \Role_1 $ and $ \LLoc{r_2}{\LGet{\Role_1}{\Set{ \LLab{\Label_k}{\tilde{\Sort}_k}{\left( \Proj{\GT_{l, k}}{\Role_2} \right)} }_{k \in \indexSet[K]}}}[l] $ for role~$ \Role_2 $.
	By the Lemmata~\ref{lem:actorsAreSequential} and \ref{lem:errorFreedom} and since this communication does not depend on another session, then there are $ m, o \in \indexSet $, $ n \in \indexSet[K] $, $ \indexSet[K] \subseteq \indexSet[K]' $ such that every conditional branch of $ \PT_m $ is a version of $ \PT_m' = \PSend[\Chan_{\mathnormal{l}}]{\Role_1}{\Role_2}{\Label_n}{\tilde{\expr}}{\PT[Q]} $ and every conditional branch of $ \PT_o $ is a version of $ \PT_o' = \PGet[\Chan_{\mathnormal{l}}]{\Role_2}{\Role_1}{\Set{ \PLab{\Label_k}{\tilde{\Args}_k}{\PT[Q]_k} }_{k \in \indexSet[K]'}} $.
	By Figure~\ref{fig:typingRules}, the derivation of $ \Gamma \vdash \PT \triangleright \Delta $ starts with some applications of the Rules~\textsf{Par} and \textsf{Cond} to split the judgement into $ \Gamma \vdash \PT_m' \triangleright \LLoc{r_1}{\LSend{\Role_2}{\Set{ \LLab{\Label_k}{\tilde{\Sort}_k}{\left( \Proj{\GT_{l, k}}{\Role_1} \right)} }_{k \in \indexSet[K]}}}[l] $, $ \Gamma \vdash \PT_o' \triangleright \LLoc{r_2}{\LGet{\Role_1}{\Set{ \LLab{\Label_k}{\tilde{\Sort}_k}{\left( \Proj{\GT_{l, k}}{\Role_2} \right)} }_{k \in \indexSet[K]}}}[l] $, and $ \Gamma \vdash \PT_i \triangleright \Delta_i $ for all $ m \neq i \neq o $.
	By the Rule~\textsf{Send} and $ n \in \indexSet[K] $, then the judgement for $ m $ implies $ \Gamma \vdash \PT[Q] \triangleright \LLoc{r_1}{\Proj{\GT_{l, n}}{\Role_1}}[l] $.
	By Lemma~\ref{lem:structuralCongruence}, with the Rule~\textsf{Get} and $ n \in \indexSet[K] $, then $ \Gamma \vdash \PT[Q]_n\Set{ \Subst{\tilde{\Args}_n@\Actor{s_{\mathnormal{l}}}{r_2}}{\tilde{\Args}_n} } \triangleright \LLoc{r_2}{\Proj{\GT_n}{\Role_2}} $.
	By Definition~\ref{def:projection}, then $ \Proj{\GT_l}{\Role_i} $ is similar to $ \Proj{\GT_{l, n}}{\Role_i} $ except for unnecessary branches of receivers for all $ m \neq i \neq o $.
	By Rule~\textsf{Get}, then $ \Gamma \vdash \PT_i \triangleright \Delta_i $ implies $ \Gamma \vdash \PT_i \triangleright \Delta_i' $ for all $ m \neq i \neq o $, where $ \Delta_i' = \Delta_i $ or $ \Delta_i = \LLoc{r_i}{\Proj{\GT_l}{\Role_i}}[l] $ and $ \Delta_i' = \LLoc{r_i}{\Proj{\GT_{l, n}}{\Role_i}}[l] $.
	Since \PT is role-distributed, so is $ \left( \PPar{\PPar{\PT[Q]_n\Set{ \Subst{\tilde{\Args}_n@\Actor{s_{\mathnormal{l}}}{r_2}}{\tilde{\Args}_n} }}{\PT[Q]}}{\left( \prod_{i \in \indexSet \setminus \Set{ m, o }}{\PT_i} \right)} \right) $.
	By the Rules~\textsf{Par} and \textsf{Cond}, then $ \left( \PPar{\PPar{\PT[Q]_n\Set{ \Subst{\tilde{\Args}_n@\Actor{s_{\mathnormal{l}}}{r_2}}{\tilde{\Args}_n} }}{\PT[Q]}}{\left( \prod_{i \in \indexSet \setminus \Set{ m, o }}{\PT_i} \right)} \right) $ is well-typed \wrt $ \Set{ \left( \GT_j, \Chan_j \right) }_{j \in \indexSet[J] \setminus \Set{ l }} \cup \Set{ \left( \GT_{l, k},  \Chan_l \right) } $.
\end{proof}

Well-typedness \wrt a parallel global type implies that the respective system can be separated into two parallel partitions.
When considering the interleaving of several sessions, this separation is possible if the two partitions do not share actors.

\begin{lemma}[Case~\ref{algo2:split}]
	\label{lem:algo2split}
	If $ \prod_{i \in \indexSet}{\PT_i} $ is well-typed \wrt $ \Set{ \left( \GT_j, \Chan_j \right) }_{j \in \indexSet[J]} $ and there are some $ \indexSet[J]_1, \indexSet[J]_2 $ such that $ \indexSet[J]_1 \cup \indexSet[J]_2 = \indexSet[J] $, $ \indexSet[J]_1 \cap \indexSet[J]_2 = \emptyset $, and there are no dependencies between the sessions in $ \indexSet[J]_1 $ and the sessions in $ \indexSet[J]_2 $ then there are $ \indexSet_1, \indexSet_2 $ such that $ \indexSet_1 \cup \indexSet_2 = \indexSet $, $ \bigcup_{i \in \indexSet_k} \Actors{\PT_i} = \Set{ \Actor{s_j}{r} \mid j \in \indexSet[J]_k \land \Role \in \Roles{\GT_j} } $, and $ \Set{ \PT_i }_{i \in \indexSet_k} $ is well-typed \wrt $ \Set{ \left( \GT_j, \Chan_j \right) }_{j \in \indexSet[J]_k} $ for all $ k \in \Set{ 1, 2 } $.
\end{lemma}

\begin{proof}
	Assume that $ \prod_{i \in \indexSet}{\PT_i} $ is well-typed \wrt $ \Set{ \left( \GT_j, \Chan_j \right) }_{j \in \indexSet[J]} $ and there are some $ \indexSet[J]_1, \indexSet[J]_2 $ such that $ \indexSet[J]_1 \cup \indexSet[J]_2 = \indexSet[J] $, $ \indexSet[J]_1 \cap \indexSet[J]_2 = \emptyset $, and there are no dependencies between the sessions in $ \indexSet[J]_1 $ and the sessions in $ \indexSet[J]_2 $, \ie there are $ \Gamma, \Delta $ such that \PT is role-distributed, $  \Set{ \GT_j }_{j \in \indexSet[J]} $ are the global types in $ \Gamma $, $ \Delta $ is coherent \wrt $ \Set{ \left( \GT_j, \Chan[n]_j \right) }_{j \in \indexSet[J]} $, $ \Gamma \vdash \PT \triangleright \Delta $, \PT is globally progressing, and for all $ j \in \indexSet[J] $ either $ \Chan[n]_j = \Chan_j $ or $ \Gamma \vdash \PT \triangleright \Delta $ connects $ \Chan[n]_j $ with $ \Chan_j $.
	By Figure~\ref{fig:typingRules}, then the derivation of $ \Gamma \vdash \PT \triangleright \Delta $ starts with Rule~\textsf{Par} that splits the judgement into parallel components.
	Then, there are $ \indexSet_1, \indexSet_2 $ such that $ \indexSet_1 \cup \indexSet_2 = \indexSet $ and $ \bigcup_{i \in \indexSet_k} \Actors{\PT_i} = \Set{ \Actor{s_j}{r} \mid j \in \indexSet[J]_k \land \Role \in \Roles{\GT_j} } $ for all $ k \in \Set{ 1, 2 } $.
	Then, there are $ \Delta_1, \Delta_2 $ such that $ \Delta = \Delta_1 \otimes \Delta_2 $, $ \Gamma \vdash \prod_{i \in \indexSet_1}{\PT_i} \triangleright \Delta_1 $, and $ \Gamma \vdash \prod_{j \in \indexSet_2}{\PT_j} \triangleright \Delta_2 $.
	Since $ \Delta $ is coherent \wrt $ \Set{ \left( \GT_j, \Chan_j \right) }_{j \in \indexSet[J]} $ and the actors of the partitions are distinct, $ \Delta_k $ is coherent \wrt $ \Set{ \left( \GT_j, \Chan_j \right) }_{j \in \indexSet[J]_k} $ for $ k \in \Set{ 1, 2 } $.
	Hence, $ \Set{ \PT_i }_{i \in \indexSet_1} $ is well-typed \wrt $ \Set{ \left( \GT_j, \Chan_j \right) }_{j \in \indexSet[J]_1} $ and $ \Set{ \PT_i }_{i \in \indexSet_2} $ is well-typed \wrt $ \Set{ \left( \GT_j, \Chan_j \right) }_{j \in \indexSet[J]_2} $.
\end{proof}

If a system is well-typed \wrt a set of types containing a parallel global type then the actors of these two parallel types are separated such that we can replace the session channel for one side (Case~\ref{algo2:GPar}).
Accordingly we strengthen Case~\ref{algo2:GPar} to:
\begin{enumerate}
	\setcounter{enumi}{6}
	\item
	\begin{enumerate}
		\setcounter{enumii}{1}
		\item $ \MapSGPI{\Set{ \PT_i' }_{i \in \indexSet_1} \cup \Set{ \PT_i }_{i \in \indexSet_2}}{\Set{ \left( \GT_{k1}, \Chan \right), \left( \GT_{k2}, \Chan_k \right) } \cup \Set{ \left( \GT_j, \Chan_j \right) }_{j \in \indexSet[J]}} $, \newline else if there are $ k \in \indexSet[J] $, $ \indexSet_1, \indexSet_2 $ such that $ \GT_k = \GPar{\GT_{k1}}{\GT_{k2}} $, $ \indexSet_1 \cup \indexSet_2 = \indexSet $, $ \indexSet_1 \cap \indexSet_2 = \emptyset $, $ \Set{ \PT_i }_{i \in \indexSet_1} $ implements all actors of $ \left( \GT_{k1}, \Chan_k \right) $ but no actor of $ \left( \GT_{k2}, \Chan_k \right) $, and $ \PT_i' $ is obtained from $ \PT_i $ by substituting or alpha converting $ \Chan_k $ by some fresh $ \Chan $. \label{algo2:GPar2}
	\end{enumerate}
\end{enumerate}
Note that the result of the algorithm, \ie the \SGP-process, does not contain session channels.
Because of that, the above modification of Case~\ref{algo2:GPar} does not change the result of the algorithm.
We use it only for the proof.

\begin{lemma}[Case~\ref{algo2:GPar2}]
	\label{lem:algo2GPar}
	If $ \prod_{i \in \indexSet}{\PT_i} $ is well-typed \wrt $ \Set{ \left( \GT_j, \Chan_j \right) }_{j \in \indexSet[J]} $, none of the Cases~\ref{algo:PRes}, \ref{algo:PPar}, or \ref{algo:PRep} can be applied, and there is $ k \in \indexSet[J] $ such that $ \GT_k = \GPar{\GT_{k1}}{\GT_{k2}} $ then there are $ \indexSet_1, \indexSet_2 $ such that $ \indexSet_1 \cup \indexSet_2 = \indexSet $, $ \indexSet_1 \cap \indexSet_2 = \emptyset $, $ \Set{ \PT_i }_{i \in \indexSet_1} $ implements all actors of $ \left( \GT_{k1}, \Chan_k \right) $ but no actor of $ \left( \GT_{k2}, \Chan_k \right) $, and $ \PT_i' $ is obtained from $ \PT_i $ by substituting or alpha converting $ \Chan_k $ by some fresh $ \Chan $, and $ \Set{ \PT_i' }_{i \in \indexSet_1} \cup \Set{ \PT_i }_{i \in \indexSet_2} $ is well-typed \wrt $ \Set{ \left( \GT_{k1}, \Chan \right), \left( \GT_{k2}, \Chan_k \right) } \cup \Set{ \left( \GT_j, \Chan_j \right) }_{j \in \indexSet[J]} $.
\end{lemma}

\begin{proof}
	Assume that $ \PT = \prod_{i \in \indexSet}{\PT_i} $ is well-typed \wrt $ \Set{ \left( \GT_j, \Chan_j \right) }_{j \in \indexSet[J]} $, none of the Cases~\ref{algo:PRes}, \ref{algo:PPar}, or \ref{algo:PRep} can be applied, and there is $ k \in \indexSet[J] $ such that $ \GT_k = \GPar{\GT_{k1}}{\GT_{k2}} $, \ie there are $ \Gamma, \Delta $ such that \PT is role-distributed, $  \Set{ \GT_j }_{j \in \indexSet[J]} $ are the global types in $ \Gamma $, $ \Delta $ is coherent \wrt $ \Set{ \left( \GT_j, \Chan[n]_j \right) }_{j \in \indexSet[J]} $, $ \Gamma \vdash \PT \triangleright \Delta $, \PT is globally progressing, and for all $ j \in \indexSet[J] $ either $ \Chan[n]_j = \Chan_j $ or $ \Gamma \vdash \PT \triangleright \Delta $ connects $ \Chan[n]_j $ with $ \Chan_j $.
	Since \PT is role-distributed, there are $ \indexSet_1 \cup \indexSet_2 = \indexSet $ such that $ \prod_{i \in \indexSet_1}{\PT_i} $ is role-distributed, $ \prod_{j \in \indexSet_2}{\PT_j} $ is role-distributed, $ \bigcup_{i \in \indexSet_1} \Actors{\PT_i} \cap \Actors{\GT_k} = \Actors{\GT_{k1}} $, and $ \bigcup_{j \in \indexSet_2} \Actors{\PT_j} \cap \Actors{\GT_k} = \Actors{\GT_{k2}} $.
	By Figure~\ref{fig:typingRules}, then the derivation of $ \Gamma \vdash \PT \triangleright \Delta $ starts with Rule~\textsf{Par} that splits the judgement into parallel components.
	Then there are $ \Delta_1, \Delta_2 $ such that $ \Delta = \Delta_1 \otimes \Delta_2 $, $ \Gamma \vdash \prod_{i \in \indexSet_1}{\PT_i} \triangleright \Delta_1 $, and $ \Gamma \vdash \prod_{j \in \indexSet_2}{\PT_j} \triangleright \Delta_2 $.
	Let $ \Chan $ be fresh.
	Since we removed already all top-level restrictions with Case~\ref{algo2:PRes}, $ \Chan_k $ is free in \PT.
	Then, $ \Gamma \Set{ \Subst{\Chan}{\Chan_k} } \vdash \prod_{i \in \indexSet_1}{\PT_i} \Set{ \Subst{\Chan}{\Chan_k} } \triangleright \Delta_1 \Set{ \Subst{\Chan}{\Chan_k} } $.
	By Rule~\textsf{Par}, then $ \Gamma' \vdash \PT \triangleright \Delta' $, where $ \Gamma' = \Gamma, \GGlobS{a}[1]{\GT}[k1]{s}, \GGlobS{a}[2]{\GT}[k2]{s}[k] $ for some fresh $ \Chan[a]_1, \Chan[a]_2 $ and $ \Delta' = \Delta_1 \Set{ \Subst{\Chan}{\Chan_k} } \otimes \Delta_2 $.
	Since $ \Delta $ is coherent \wrt $ \Set{ \left( \GT_j, \Chan[n]_j \right) }_{j \in \indexSet[J]} $ and the actors of the partitions are distinct, $ \Delta' $ is coherent \wrt $ \Set{ \left( \GT_{k1}, \Chan \right), \left( \GT_{k2}, \Chan_k \right) } \cup \Set{ \left( \GT_j, \Chan[n]_j \right) }_{j \in \indexSet[J]} $.
	Hence, $ \Set{ \PT_i' }_{i \in \indexSet_1} \cup \Set{ \PT_i }_{i \in \indexSet_2} $ is well-typed \wrt $ \Set{ \left( \GT_{k1}, \Chan \right), \left( \GT_{k2}, \Chan_k \right) } \cup \Set{ \left( \GT_j, \Chan_j \right) }_{j \in \indexSet[J]} $.
\end{proof}

Case~\ref{algo2:PReq} maps the communication partners of a session invitation on an empty value update.
It preserves well-typedness in its recursive call \wrt the same global types.
By Lemma~\ref{lem:errorFreedom}, if one of the necessary prefixes for a session invitation is unguarded and this invitation is according to the interaction type system of \cite{BettiniAtall08} not dependent on another session then all other necessary prefixes are composed in parallel and are guarded by conditionals only.

\begin{lemma}[Case~\ref{algo2:PReq}]
	\label{lem:algo2PReq}
	If $ \prod_{i \in \indexSet}{\PT_i} $ is well-typed \wrt $ \Set{ \left( \GT_j, \Chan_j \right) }_{j \in \indexSet[J]} $, $ l \in \indexSet[J] $, none of the Cases~\ref{algo2:PRes}, \ref{algo2:PPar}, or \ref{algo2:PRep} can be applied, the session $ \Chan_l $ is not initialised, and this session initialisation does not depend on another session then there are $ k1, \ldots, kn \in \indexSet $ such that every conditional branch of $ \PT_{k1} $ is a version of $ \PReq{\Role[2]..\Role[n]}{\PT_1'} $, every conditional branch of $ \PT_{k2} $ is a version of $ \PAcc{\Role[2]}{\PT_2'} $, \ldots, every conditional branch of $ \PT_{kn} $ is a version of $ \PAcc{\Role[n]}{\PT_n'} $, and $ \Set{ \PT_1', \ldots, \PT_n' } \cup \Set{ \PT_i }_{i \in \indexSet \setminus \Set{ k1, \ldots, kn }} $ is well-typed \wrt $ \Set{ \left( \GT_j, \Chan_j \right) }_{j \in \indexSet[J]} $.
\end{lemma}

\begin{proof}
	Assume $ \prod_{i \in \indexSet}{\PT_i} $ is well-typed \wrt $ \Set{ \left( \GT_j, \Chan_j \right) }_{j \in \indexSet[J]} $, $ l \in \indexSet[J] $, none of the Cases~\ref{algo2:PRes}, \ref{algo2:PPar}, or \ref{algo2:PRep} can be applied, the session $ \Chan_l $ is not initialised, and this session initialisation does not depend on another session, \ie there are $ \Gamma, \Delta $ such that \PT is role-distributed, $  \Set{ \GT_j }_{j \in \indexSet[J]} $ are the global types in $ \Gamma $, $ \Delta $ is coherent \wrt $ \Set{ \left( \GT_j, \Chan[n]_j \right) }_{j \in \indexSet[J]} $, $ \Gamma \vdash \PT \triangleright \Delta $, \PT is globally progressing, and for all $ j \in \indexSet[J] $ either $ \Chan[n]_j = \Chan_j $ or $ \Gamma \vdash \PT \triangleright \Delta $ connects $ \Chan[n]_j $ with $ \Chan_j $.
	By coherence and since there is no dependency to other sessions, there are $ k1, \ldots, kn \in \indexSet $ such that every conditional branch of $ \PT_{k1} $ is a variant of $ \PReq{\Role[2]..\Role[n]}{\PT_1'} $, every conditional branch of $ \PT_{k2} $ is a variant of $ \PAcc{\Role[2]}{\PT_2'} $, \ldots, every conditional branch of $ \PT_{kn} $ is a variant of $ \PAcc{\Role[n]}{\PT_n'} $.
	By the Rules~\textsf{If-T}, \textsf{If-F}, and \textsf{Link} of Figure~\ref{fig:reductionSemantics}, $ \PT \Steps \PRes{\Chan_l}{\left( \PPar{\prod_{i \in \Set{ k1, \ldots, kn }}{\PT_i'}}{\prod_{i \in \indexSet \setminus \Set{ k1, \ldots, kn }}{\PT_i}} \right)} $.
	By Lemma~\ref{lem:subjectReduction}, Lemma~\ref{lem:roleDist}, and since only conditionals and a session initialisation is performed in these steps, then $ \PRes{\Chan_l}{\left( \PPar{\prod_{i \in \Set{ k1, \ldots, kn }}{\PT_i'}}{\prod_{i \in \indexSet \setminus \Set{ k1, \ldots, kn }}{\PT_i}} \right)} $ is well-typed \wrt $ \Set{ \left( \GT_j, \Chan_j \right) }_{j \in \indexSet[J]} $.
	By Lemma~\ref{lem:algo2PRes}, then $ \Set{ \PT_1', \ldots, \PT_n' } \cup \Set{ \PT_i }_{i \in \indexSet \setminus \Set{ k1, \ldots, kn }} $ is well-typed \wrt the types $ \Set{ \left( \GT_j, \Chan_j \right) }_{j \in \indexSet[J]} $.
\end{proof}

Finally, we prove Theorem~2 of \cite{petersWagnerNestmann19}:
\begin{quote}
	If \PT is well-typed \wrt $ \Set{ \left( \GT_j, \Chan_j \right) }_{j \in \indexSet[J]} $ then the abstraction \MapSGPI{\Set{ \PT }}{\Set{ \left( \GT_j, \Chan_j \right) }_{j \in \indexSet[J]}} is defined and returns a \SGP-process.
\end{quote}

\begin{proof}[Proof of Theorem~2 of \cite{petersWagnerNestmann19}]
	Assume $ \PT $ is well-typed \wrt $ \Set{ \left( \GT_j, \Chan_j \right) }_{j \in \indexSet[J]} $.
	We proceed with an induction over the set $ \Set{ \left( \GT_j, \Chan_j \right) }_{j \in \indexSet[J]} $ and the structure of the types in this set (Definition~\ref{def:globalTypes}).
	\begin{description}
		\item[Case of $ {\indexSet[J]} = \emptyset $:]
			By Case~\ref{algo2:Empty}, then $ \MapSGPI{\Set{ \PT_i }_{i \in \indexSet}}{\Set{ \left( \GT_j, \Chan_j \right) }_{j \in \indexSet[J]}} $ is defined and returns \SEnd.
		\item[Case of $ {\indexSet[J] = \indexSet[J]_1 \cup \indexSet[J]_2} $, $ {\indexSet[J]_1 \cap \indexSet[J]_2 = \emptyset} $, and $ {\indexSet[J]_1}, {\indexSet[J]_2} $ are independent:]
			$ $\\
			By Lemma~\ref{lem:algo2split}, then there are $ \indexSet_1, \indexSet_2, k \in \Set{ 1, 2 } $ such that $ \indexSet_1 \cup \indexSet_2 = \indexSet $, $ \bigcup_{i \in \indexSet_k} \Actors{\PT_i} = \Set{ \Actor{s_j}{r} \mid j \in \indexSet[J]_k \land \Role \in \Roles{\GT_j} } $, and the composition $ \Set{ \PT_i }_{i \in \indexSet_k} $ is well-typed \wrt $ \Set{ \left( \GT_j, \Chan_j \right) }_{j \in \indexSet[J]_k} $.
			Since the two parts do not share actors and we indicate input variables with actors, $ S_1 $ and $ S_2 $ are independent.
			By the induction hypothesis, then both $ \MapSGPI{\Set{ \PT_i }_{i \in \indexSet_k}}{\Set{ \left( \GT_j, \Chan_j \right) }_{j \in \indexSet[J]_k}} $ for $ k \in \Set{ 1, 2} $ are defined and return the \SGP-processes $ \ST_1 $ and $ \ST_2 $.
			By Case~\ref{algo2:split}, then $ \MapSGPI{\Set{ \PT_i }_{i \in \indexSet}}{\Set{ \left( \GT_j, \Chan_j \right) }_{j \in \indexSet[J]}} $ is defined and returns $ \SPar{\ST_1}{\ST_2} $.
		\item[Case of $ \GT_l = \GCom{\Role_1}{\Role_2}{\Set{ \GLab{\Label_i}{\tilde{\Sort}_i}{\GT_i} }_{i \in \indexSet}} $ with $ l \in {\indexSet[J]} $ and $ \GT_l $ is independent:]
			$  $\\
			By the Lemmata~\ref{lem:algo2PRes}, \ref{lem:algo2PPar}, \ref{lem:algo2PRep}, and \ref{lem:algo2PReq}, the mapping can remove restrictions, split parallel compositions, unfold recursions, and initialise the session $ \Chan_l $ of $ \GT_l $ of the process without altering the global types or violating well-typedness.
			Let $ \prod_{i \in \indexSet}{\PT_i} $ be the result of these cases such that $ \prod_{i \in \indexSet}{\PT_i} $ is well-typed \wrt $ \Set{ \left( \GT_j, \Chan_j \right) }_{j \in \indexSet[J]} $.
			By Lemma~\ref{lem:algo2GCom}, there are $ m, o \in \indexSet $, $ n \in \indexSet[K] $ and $ \indexSet[K] \subseteq \indexSet[K]' $ such that in every conditional branch of $ \PT_m $ there is a version of $ \PT_m' = \PSend[\Chan_{\mathnormal{l}}]{\Role_1}{\Role_2}{\Label_n}{\tilde{\expr}}{\PT[Q]} $ and in every conditional branch of $ \PT_o $ there is a version of $ \PT_o' = \PGet[\Chan_{\mathnormal{l}}]{\Role_2}{\Role_1}{\Set{ \PLab{\Label_k}{\tilde{\Args}_k}{\PT[Q]_k} }_{k \in \indexSet[K]'}} $.
			By Lemma~\ref{lem:algo2PCond}, all conditionals that guard either $ \PT_m' $ or $ \PT_o' $ can be mapped on \SGP-conditionals without violating well-typedness, where we possibly have to apply the Lemmata~\ref{lem:algo2PRes}, \ref{lem:algo2PPar}, and \ref{lem:algo2PRep} in between and the order of cases in Definition~\ref{def:algorithm2} allows to resolve exactly these conditionals before resolving the communication in $ \GT_l $.
			This is because the structure of $ \GT_l $ rules out the Cases~\ref{algo2:Empty}, \ref{algo2:TVar}, \ref{algo2:GPar}, and \ref{algo2:GRec}, the fact that the communication that guards $ \GT_l $ is not dependent on another session rules out Case~\ref{algo2:PReq}, and Case~\ref{algo2:GCom} becomes applicable as soon as all of these conditionals are resolved.
			By Lemma~\ref{lem:algo2GCom}, then $ \left( \PPar{\PPar{\PT[Q]_n\Set{ \Subst{\tilde{\Args}_n@\Actor{s_{\mathnormal{l}}}{r_2}}{\tilde{\Args}_n} }}{\PT[Q]}}{\left( \prod_{i \in \indexSet \setminus \Set{ m, o }}{\PT_i} \right)} \right) $ is well-typed \wrt $ \Set{ \left( \GT_j, \Chan_j \right) }_{j \in \indexSet[J] \setminus \Set{ l }} \cup \Set{ \left( \GT_{l, k},  \Chan_l \right) } $.
			By Figure~\ref{fig:typingRules}, all branches of the conditionals are well-typed \wrt $ \Set{ \left( \GT_j, \Chan_j \right) }_{j \in \indexSet[J] \setminus \Set{ l }} \cup \Set{ \left( \GT_{l, k},  \Chan_l \right) } $, \ie we unguard versions of $ \PT_m' $ and $ \PT_o' $---that may differ in their labels but only implement labels that are specified in the type---in all branches.
			By the induction hypothesis, then
			$ \MapSGPI{\mathcal{P}}{\mathcal{G}} $ with $ \mathcal{P} = \Set{ \PT[Q]_m\Set{ \Subst{\tilde{\Args}_m@\Actor{s_{\mathnormal{l}}}{r_2}}{\tilde{\Args}_m} }, \PT[Q] } \cup \Set{ \PT_i }_{i \in \indexSet \setminus \Set{ m, o }} $ and $ \mathcal{G} = \Set{ \left( \GT_j, \Chan_j \right) }_{j \in \indexSet[J] \setminus \Set{ l }} \cup \Set{ \left( \GT_{l, k}, \Chan_l \right) } $ is defined and returns a \SGP-process $ \ST' $ for each of these branches.
			Let \ST be the result of putting the respective version of $ \SAssign{\tilde{\Args}_m@\Actor{s_{\mathnormal{l}}}{r_2}}{\tilde{\expr}}{\ST'} $ in the respective branch of the generated \SGP-conditionals.
			By the Cases~\ref{algo:GCom} and \ref{algo:PCond}, then $ \MapSGPI{\Set{ \PT_i }_{i \in \indexSet}}{\mathcal{G}} $ is defined and returns \ST.
		\item[Case of $ \GT_l = \GPar{\GT_1}{\GT_2} $ with $ l \in {\indexSet[J]} $:]
			By the Lemmata~\ref{lem:algo2PRes}, \ref{lem:algo2PPar}, and \ref{lem:algo2PRep}, the mapping can remove restrictions, split parallel compositions, and unfold recursions of the process without altering the global type or violating well-typedness.
			Let $ \prod_{i \in \indexSet}{\PT_i} $ be the result of these cases such that $ \prod_{i \in \indexSet}{\PT_i} $ is well-typed \wrt $ \Set{ \left( \GT_j, \Chan_j \right) }_{j \in \indexSet[J]} $.
			By Lemma~\ref{lem:algo2GPar}, then there are $ \indexSet_1, \indexSet_2 $ such that $ \indexSet_1 \cup \indexSet_2 = \indexSet $, $ \indexSet_1 \cap \indexSet_2 = \emptyset $, $ \Set{ \PT_i }_{i \in \indexSet_1} $ implements all actors of $ \left( \GT_{k1}, \Chan_k \right) $ but no actor of $ \left( \GT_{k2}, \Chan_k \right) $, and $ \PT_i' $ is obtained from $ \PT_i $ by substituting or alpha converting $ \Chan_k $ by some fresh $ \Chan $, and $ \Set{ \PT_i' }_{i \in \indexSet_1} \cup \Set{ \PT_i }_{i \in \indexSet_2} $ is well-typed \wrt $ \Set{ \left( \GT_{k1}, \Chan \right), \left( \GT_{k2}, \Chan_k \right) } \cup \Set{ \left( \GT_j, \Chan_j \right) }_{j \in \indexSet[J]} $.
			By the induction hypothesis, then $ \MapSGPI{\mathcal{P}}{\mathcal{G}} $ with $ \mathcal{P} = \Set{ \PT_i' }_{i \in \indexSet_1} \cup \Set{ \PT_i }_{i \in \indexSet_2} $ and $ \mathcal{G} = \Set{ \left( \GT_{k1}, \Chan \right), \left( \GT_{k2}, \Chan_k \right) } \cup \Set{ \left( \GT_j, \Chan_j \right) }_{j \in \indexSet[J]} $ is defined and returns a \SGP-process \ST.
			By Case~\ref{algo2:GPar2}, then $ \MapSGPI{\Set{ \PT_i }_{i \in \indexSet}}{\Set{ \left( \GT_j, \Chan_j \right) }_{j \in \indexSet[J]}} $ is defined and returns \ST.
		\item[Case of $ \GT_l = \GRec{\TVar_l}{\GT_l'} $ with $ l \in {\indexSet[J]} $:]
			By the dependency relation, the loops of different interleaved session are unified, \ie by performing the other cases we can reduce the types such that $ \GT_j = \GRec{\TVar_j}{\GT_j'} $ for all $ j \in \indexSet[J] $.
			By the Lemmata~\ref{lem:algo2PRes}, \ref{lem:algo2PPar}, \ref{lem:algo2PRep}, and \ref{lem:algo2PReq} the mapping can remove restrictions, split parallel compositions, unfold recursions, and initialise sessions of the process without altering the global type or violating well-typedness.
			Let $ \prod_{i \in \indexSet}{\PT_i} $ be the result of these cases such that $ \prod_{i \in \indexSet}{\PT_i} $ is well-typed \wrt $ \Set{ \left( \GT_j, \Chan_j \right) }_{j \in \indexSet[J]} $.
			Since type variables are bound and guarded in global types, $ \GT_l $ can reduce to a type variable only after Case~\ref{algo2:GRec} has introduced a \SGP-recursion.
			Note that neither Case~\ref{algo2:GRec} nor Case~\ref{algo2:TVar} introduce requirements on the considered process.
			Unfortunately, the arguments of the recursive call of Case~\ref{algo2:GRec} do not preserve well-typedness, \ie $ \Set{ \PT_i }_{i \in \indexSet} $ is not well-typed \wrt $ \Set{ \left( \GT_j', \Chan_j \right) }_{j \in \indexSet[J]} $, because we removed the recursion binder from the type but not the system.
			By Figure~\ref{fig:typingRules}, $ \Set{ \PT_i }_{i \in \indexSet} $ will behave as required by $ \Set{ \left( \GT_j', \Chan_j \right) }_{j \in \indexSet[J]} $ until $ \Set{ \left( \GT_j', \Chan_j \right) }_{j \in \indexSet[J]} $ is reduced to $ \mathcal{G} = \Set{ \TVar_j }_{j \in \indexSet[J]} $ and, thus, there are $ \PT_i' $ such that $ \Set{ \PT_i' }_{i \in \indexSet} $ is well-typed \wrt $ \Set{ \left( \GT_j', \Chan_j \right) }_{j \in \indexSet[J]} $.
			By the induction hypothesis, then $ \MapSGPI{\Set{ \PT_i' }_{i \in \indexSet}}{\Set{ \left( \GT_j', \Chan_j \right) }_{j \in \indexSet[J]}} $ is defined and returns a \SGP-process \ST.
			By Definition~\ref{def:algorithm2}, the mapping will follow the structure of $ \Set{ \left( \GT_j', \Chan_j \right) }_{j \in \indexSet[J]} $ to reduce the system until $ \Set{ \left( \GT_j', \Chan_j \right) }_{j \in \indexSet[J]} $ is reduced to $ \mathcal{G} $ and then Case~\ref{algo2:TVar} will ignore the remainder of the system.
			Thus, the mapping considers only the parts of $ \Set{ \PT_i }_{i \in \indexSet} $ that are already captured in $ \Set{ \PT_i' }_{i \in \indexSet} $, \ie we have $ \MapSGPI{\Set{ \PT_i }_{i \in \indexSet}}{\Set{ \left( \GT_j', \Chan_j \right) }_{j \in \indexSet[J]}} = \MapSGPI{\Set{ \PT_i' }_{i \in \indexSet}}{\Set{ \left( \GT_j', \Chan_j \right) }_{j \in \indexSet[J]}} = \ST $.
			By Case~\ref{algo2:GRec}, then $ \MapSGPI{\Set{ \PT_i }_{i \in \indexSet}}{\Set{ \left( \GT_j, \Chan_j \right) }_{j \in \indexSet[J]}} $ is defined and returns $ \SRep{\SVar_{\mathcal{G}}}{\ST} $.
		\item[Case of $ \GT_l = \TVar $ with $ j \in {\indexSet[J]} $:]
			By the dependency relation, the loops of different interleaved session are unified, \ie by performing the other cases we can reduce the types such that $ \GT_j = \TVar_j $ for all $ j \in \indexSet[J] $.
			By Case~\ref{algo2:TVar}, then $ \MapSGP{\Set{ \PT_i }_{i \in \indexSet}}{\Set{ \left( \GT_j, \Chan_j \right) }_{j \in \indexSet[J]}} $ is defined and returns $ \SVar_{\mathcal{G}} $ with $ \mathcal{G} = \Set{ \TVar_j }_{j \in \indexSet[J]} $.
		\item[Case of $ \GT_l = \GEnd $ with $ l \in {\indexSet[J]} $:]
			By Lemma~\ref{lem:algo2GEnd}, then \PT contains no communication prefixes on $ \Chan_k $ and cannot invite the session $ \Chan_k $.
			Then, $ \Set{ \PT_i }_{i \in \indexSet} $ is well-typed \wrt $ \Set{ \left( \GT_j, \Chan_j \right) }_{j \in \indexSet[J] \setminus \Set{ l }} $.
			By the induction hypothesis, then we know that $ \MapSGPI{\Set{ \PT_i }_{i \in \indexSet}}{\Set{ \left( \GT_j, \Chan_j \right) }_{j \in \indexSet[J] \setminus \Set{ l }}} $ is defined and returns a \SGP-process \ST.
			By Case~\ref{algo2:GEnd}, then $ \MapSGPI{\Set{ \PT_i }_{i \in \indexSet}}{\Set{ \left( \GT_j, \Chan_j \right) }_{j \in \indexSet[J]}} $ is defined and returns \ST.
	\end{description}
\end{proof}

Alpha conversion may influence the outcome of the mapping, \ie $ \PT \equiv_{\alpha} \PT' $ does not necessarily imply $ \MapSGPS{\Set{ \PT }}{\Set{ \left( \GT_j, \Chan_j \right) }_{j \in \indexSet[J]}} = \MapSGPS{\Set{ \PT' }}{\Set{ \left( \GT_j, \Chan_j \right) }_{j \in \indexSet[J]}} $, because the renaming of input bounded variables changes the names in the vector of the generated \SGP-system.
Because of that, we assume in this section that no sequence of steps will use alpha conversion to rename input binders.
Apart from that, structural congruence does not influence this mapping, because of the Cases~\ref{algo2:PRes}, \ref{algo2:PPar}, and \ref{algo2:PRep}.

\begin{lemma}
	\label{lem:algoStruc}
	Let $ \mathcal{G} = \Set{ \left( \GT_j, \Chan_j \right) }_{j \in \indexSet[J]} $.\\
	If $ \MapSGPS{\Set{ \PT }}{\mathcal{G}} $ is defined and $ \PT \equiv \PT' $, where no alpha conversion is used to rename input binders, then $ \MapSGPS{\Set{ \PT }}{\mathcal{G}} = \MapSGPS{\Set{ \PT' }}{\mathcal{G}} $.
\end{lemma}

\begin{proof}
	Assume that $ \MapSGPS{\Set{ \PT }}{\mathcal{G}} $ is defined.
	We proceed with an induction over the rules of structural congruence that are used to obtain $ \PT \equiv \PT' $.
	\begin{description}
		\item[Case of Alpha Conversion:]
			In this case $ \PT \equiv_{\alpha} \PT' $, \ie $ \PT $ and $ \PT' $ differ only by renamings of names used for restriction binders.
			By Definition~\ref{def:algorithm2} and Lemma~\ref{lem:algo2PRes}, then $ \MapSGPI{\Set{ \PT }}{\mathcal{G}} = \MapSGPI{\Set{ \PT' }}{\mathcal{G}} $.
			Hence, $ \MapSGPS{\Set{ \PT }}{\mathcal{G}} = \MapSGPS{\Set{ \PT' }}{\mathcal{G}} $.
		\item[Case of $ \PPar{\PT[Q]}{\PEnd} \equiv {\PT[Q]} $:]
			In this case $ \PT = \PPar{\PT[Q]}{\PEnd} $ and $ \PT' = \PT[Q] $.
			By Definition~\ref{def:algorithm2}, $ \MapSGPI{\Set{ \PT }}{\mathcal{G}} = \MapSGPI{\Set{ \PT[Q], \PEnd }}{\mathcal{G}} = \MapSGPI{\Set{ \PT[Q] }}{\mathcal{G}} $.
			Then $ \MapSGPS{\Set{ \PT }}{\mathcal{G}} = \MapSGPS{\Set{ \PT' }}{\mathcal{G}} $.
		\item[Case of $ \PPar{\PT[Q]_1}{\PT[Q]_2} \equiv \PPar{\PT[Q]_2}{\PT[Q]_1} $:]
			In this case $ \PT = \PPar{\PT[Q]_1}{\PT[Q]_2} $ and $ \PT' = \PPar{\PT[Q]_2}{\PT[Q]_1} $.
			By Definition~\ref{def:algorithm2}, then $ \MapSGPI{\Set{ \PT }}{\mathcal{G}} = \MapSGPI{\Set{ \PT[Q]_1, \PT[Q]_2 }}{\mathcal{G}} = \MapSGPI{\Set{ \PT' }}{\mathcal{G}} $.
			Then, we have $ \MapSGPS{\Set{ \PT }}{\mathcal{G}} = \MapSGPS{\Set{ \PT' }}{\mathcal{G}} $.
		\item[Case of $ \PPar{\PT[Q]_1}{\left( \PPar{\PT[Q]_2}{\PT[Q]_3} \right)} \equiv \PPar{\left( \PPar{\PT[Q]_1}{\PT[Q]_2} \right)}{\PT[Q]_3} $:]
			In this case $ \PT = \PPar{\PT[Q]_1}{\left( \PPar{\PT[Q]_2}{\PT[Q]_3} \right)} $ and $ \PT' = \PPar{\left( \PPar{\PT[Q]_1}{\PT[Q]_2} \right)}{\PT[Q]_3} $.
			Then $ \MapSGPI{\Set{ \PT }}{\mathcal{G}} = \MapSGPI{\Set{ \PT[Q]_1, \PT[Q]_2, \PT[Q]_3 }}{\mathcal{G}} = \MapSGPI{\Set{ \PT' }}{\mathcal{G}} $, because of Definition~\ref{def:algorithm2}.
			Then, $ \MapSGPS{\Set{ \PT }}{\mathcal{G}} = \MapSGPS{\Set{ \PT' }}{\mathcal{G}} $.
		\item[Case of $ \PRes{\Chan}{\PRes{\Chan'}{\PT[Q]}} \equiv \PRes{\Chan'}{\PRes{\Chan}{\PT[Q]}} $:]
			In this case we have $ \PT = \PRes{\Chan}{\PRes{\Chan'}{\PT[Q]}} $ and $ \PT' = \PRes{\Chan'}{\PRes{\Chan}{\PT[Q]}} $.
			By Definition~\ref{def:algorithm2}, then $ \MapSGPI{\Set{ \PT }}{\mathcal{G}} = \MapSGPI{\Set{ \PT[Q] }}{\mathcal{G}} = \MapSGPI{\Set{ \PT' }}{\mathcal{G}} $.
			Then, we have $ \MapSGPS{\Set{ \PT }}{\mathcal{G}} = \MapSGPS{\Set{ \PT' }}{\mathcal{G}} $.
		\item[Case of $ \PRes{\Chan}{\PEnd} \equiv \PEnd $:]
			In this case $ \PT = \PRes{\Chan}{\PEnd} $ and $ \PT' = \PEnd $.
			By Definition~\ref{def:algorithm2}, $ \MapSGPI{\Set{ \PT }}{\GT} = \MapSGPI{\Set{ \PEnd }}{\GT} $.
			Then $ \MapSGPS{\Set{ \PT }}{\mathcal{G}} = \MapSGPS{\Set{ \PT' }}{\mathcal{G}} $.
		\item[Case of $ \PRep{\PVar}{\PT[Q]} \equiv {\PT[Q]}\Set{ \Subst{\PRep{\PVar}{\PT[Q]}}{\PVar} } $:]
			In this case we have $ \PT = \PRep{\PVar}{\PT[Q]} $ and $ \PT' = {\PT[Q]}\Set{ \Subst{\PRep{\PVar}{\PT[Q]}}{\PVar} } $.
			Then $ \MapSGPI{\Set{ \PT }}{\mathcal{G}} = \MapSGPI{\Set{ \PRep{\PVar}{\PT[Q]} }}{\mathcal{G}} = \MapSGPI{\Set{ {\PT[Q]}\Set{ \Subst{\PRep{\PVar}{\PT[Q]}}{\PVar} } }}{\mathcal{G}} $, because of Definition~\ref{def:algorithm2}.
			Then, we have $ \MapSGPS{\Set{ \PT }}{\mathcal{G}} = \MapSGPS{\Set{ \PT' }}{\mathcal{G}} $.
		\item[Case of $ \PRes{\Chan}{\left( \PPar{\PT[Q]_1}{\PT[Q]_2} \right)} \equiv \PPar{\PT[Q]_1}{\PRes{\Chan}{\PT[Q]_2}} $ if $ \Chan \notin \FreeNames{\PT[Q]_1} $:]
			In this case $ \PT = \PRes{\Chan}{\left( \PPar{\PT[Q]_1}{\PT[Q]_2} \right)} $ and $ \PT' = \PPar{\PT[Q]_1}{\PRes{\Chan}{\PT[Q]_2}} $.
			By Definition~\ref{def:algorithm2}, then $ \MapSGPI{\Set{ \PT }}{\mathcal{G}} = \MapSGPI{\Set{ \PT[Q]_1, \PT[Q]_2 }}{\mathcal{G}} = \MapSGPI{\Set{ \PT' }}{\mathcal{G}} $.
			Finally, we have $ \MapSGPS{\Set{ \PT }}{\mathcal{G}} = \MapSGPS{\Set{ \PT' }}{\mathcal{G}} $.
	\end{description}
\end{proof}

Now we analyse how the original system and its sequentialisation into a \SGP-system are related.
First we prove that \SGP-systems introduce no new behaviour in Theorem~3 of \cite{petersWagnerNestmann19}:
\begin{quote}
	Let $ \mathcal{G} = \Set{ \left( \GT_j, \Chan_j \right) }_{j \in \indexSet[J]} $.
If \PT is well-typed \wrt $ \mathcal{G} $ then for all $ \MapSGPS{\Set{ \PT }}{\mathcal{G}} \Step \SST' $ there exist $ \PT', \mathcal{G}' $ such that $ \PT \Step \PT' $, $ \PT' $ is well-typed \wrt $ \mathcal{G}' $, and $ \MapSGPS{\Set{ \PT' }}{\mathcal{G}'} \SSequiv \SST' $.
\end{quote}

\begin{proof}[Proof of Theorem~3 of \cite{petersWagnerNestmann19}]
	Assume that \PT is well-typed \wrt $ \mathcal{G} $.
	We proceed by an induction on the reduction rules that are used to derive the step $ \MapSGPS{\Set{ \PT }}{\mathcal{G}} \Step \SST' $.
	\begin{description}
		\item[Case of Rule~\textsf{Ass}:]
			In this case we have $ \MapSGPS{\Set{ \PT }}{\mathcal{G}} = \SSys{\Vect{V}}{\SAssign{\tilde{v}}{\tilde{e}}{\ST}} $ and $ \SST' = \Eval{\SSys{\Update{\Vect{V}}{\tilde{v}}{\tilde{\expr}}}{\ST}} $.
			By Lemma~\ref{lem:algo2GCom}, then either $ \SAssign{\tilde{v}}{\tilde{e}}{\ST} $ is the empty assignment $ \tau $ that resulted from mapping a session initialisation (Case~\ref{algo2:PReq}) or this assignment is not empty and resulted from mapping a communication of an initialised session (Case~\ref{algo2:GCom}).
			\begin{description}
				\item[Case of $ \SAssign{\tilde{v}}{\tilde{e}}{\ST} = \tau.\ST $:] Since the value assignment is unguarded,
					$$ \PT \equiv \PRes{\tilde{\Chan}}{\left( \PPar{\PPar{\PReq{\Role[2]..\Role[n]}{\PT_1}}{\PPar{\PAcc{\Role[2]}{\PT_2}}{\PPar{\ldots}{\PAcc{\Role[n]}{\PT_n}}}}}{\PT[Q]} \right)}, $$
					where we do not alpha-convert input binders but use alpha conversion to ensure that $ \Chan $ is not contained in $ \left( \tilde{\Chan} \cup \BoundNames{\PPar{\PT_1}{\PPar{\PT_2}{\PPar{\ldots}{\PT_n}}}} \cup \Names{\PT[Q]} \right) $.
					By Figure~\ref{fig:reductionSemantics}, then $ \PT \Step \PT' = \PRes{\tilde{\Chan}}{\PRes{\Chan}{\left( \PPar{\PPar{\PT_1}{\PPar{\PT_2}{\PPar{\ldots}{\PT_n}}}}{\PT[Q]} \right)}} $, where we do again not alpha-convert input bounded names.
					By the Lemmata~\ref{lem:subjectReduction} and \ref{lem:algo2PReq}, then $ \PT' $ is well-typed \wrt $ \mathcal{G} $.
					By Definition~\ref{def:algorithm2} and Theorem~1 of \cite{petersWagnerNestmann19}, then $ \MapSGPS{\Set{ \PT' }}{\mathcal{G}} = \SST' $.
				\item[Case of $ \SAssign{\tilde{v}}{\tilde{e}}{\ST} \neq \tau.\ST $:] By Definition~\ref{def:algorithm2}, $ \GT_l = \GCom{\Role_1}{\Role_2}{\Set{ \GLab{\Label_i}{\tilde{\Sort}_i}{\GT_{l, i}} }_{i \in \indexSet}} $ for some $ l \in \indexSet[J] $ and $ \Chan_l $ is minimal \wrt to the dependency relation. Since the value assignment is unguarded,
					$$ \PT \equiv \PRes{\tilde{\Chan}}{\left( \PPar{\PPar{\PSend[\Chan_{\mathnormal{l}}]{\Role_1}{\Role_2}{\Label_m}{\tilde{\expr}}{\PT_1}}{\PGet[\Chan_{\mathnormal{l}}]{\Role_2}{\Role_1}{\Set{ \PLab{\Label_j}{\tilde{\Args}_j}{\PT_{2, j}} }_{j \in \indexSet[J]}}}}{\PT[Q]} \right)} $$
					with $ m \in \indexSet $ and $ \indexSet \subseteq \indexSet[J] $, where we do not alpha-convert input bounded names.
					By Figure~\ref{fig:reductionSemantics}, then $ \PT \Step \PT' = \PRes{\tilde{\Chan}}{\left( \PPar{\PPar{\PT_{2,m}\Set{ \Subst{\tilde{\expr}}{\tilde{\Args}_m} }}{\PT_1}}{\PT[Q]} \right)} $, where we do again not alpha-convert input bounded names.
					By the Lemmata~\ref{lem:subjectReduction} and \ref{lem:algo2GCom}, then $ \PRes{\tilde{\Chan}}{\left( \PPar{\PPar{\PT_{2,m}}{\PT_1}}{\PT[Q]} \right)} $ is well-typed \wrt the types $ \mathcal{G}' = \Set{ \left( \GT_j, \Chan_j \right) }_{j \in \indexSet[J] \setminus \Set{ l }} \cup \Set{ \GT_{l, m} } $.
					By Definition~\ref{def:algorithm2} and Theorem~1 of \cite{petersWagnerNestmann19}, then we have $ \MapSGPS{\Set{ \PT' }}{\mathcal{G}'} = \SST' $.
			\end{description}
		\item[Case of Rule~\textsf{Par}:]
			In this case $ \SST = \SSys{\Vect{V}}{\SPar{\ST_1}{\ST_2}} $, $ \SSys{\Vect{V}}{\ST_1} \Step \SSys{\Vect{V}'}{\ST_1'} $, and $ \SST' = \SSys{\Vect{V}'}{\SPar{\ST_1'}{\ST_2}} $.
			By Definition~\ref{def:algorithm2}, Lemma~\ref{lem:algo2split}, and Lemma~\ref{lem:algo2GPar}, then there are some $ \mathcal{G}_{\text{split}} $, $ \indexSet_1 \cup \indexSet_2 = \indexSet $, $ \indexSet[J]_1 \cup \indexSet[J]_2 = \indexSet[J]_{\text{split}} $ such that $ \mathcal{G}' = \Set{ \left( \GT_j', \Chan_j \right) }_{j \in \indexSet[J]_{\text{split}}} $ results from $ \mathcal{G} $ by applications of Case~\ref{algo2:GPar}, $ \indexSet[J]_1 \cap \indexSet[J]_2 = \emptyset $, $ \PT \equiv \PRes{\tilde{\Chan}}{\left( \PPar{\PT_1}{\PT_2} \right)} $, $ \Actors{\PT_k} = \Set{ \Actor{s_j}{r} \mid j \in \indexSet[J]_k \land \Role \in \Roles{\GT_j'} } $, and $ \MapSGPI{\Set{\PT_k}}{\Set{ \left( \GT_j', \Chan_j \right) }_{j \in \indexSet[J]_k}} = \ST_k $ for $ k \in \Set{ 1, 2 } $.
			By the induction hypothesis, $ \MapSGPI{\Set{ \PT_1 }}{\Set{ \left( \GT_j', \Chan_j \right) }_{j \in \indexSet[J]_1}} = \ST_1 $ and $ \SSys{\Vect{V}}{\ST_1} \Step \SSys{\Vect{V}'}{\ST_1'} $ imply that there are $ \PT_1', \mathcal{G}_1' $ such that $ \PT_1 \Step \PT_1' $, $ \PT_1' $ is well-typed \wrt $ \mathcal{G}_1' $, and $ \MapSGPS{\Set{\PT_1'}}{\mathcal{G}_1'} \SSequiv \SSys{\Vect{V}'}{\ST_1'} $.
			By Figure~\ref{fig:reductionSemantics}, then $ P \Step \PT' = \PRes{\tilde{\Chan}}{\left( \PPar{\PT_1'}{\PT_2} \right)} $, where we do not alpha-convert input bounded names.
			By the Lemmata~\ref{lem:subjectReduction}, \ref{lem:algo2PRes}, and \ref{lem:algo2PPar}, then $ \PT' $ is well-typed \wrt $ \mathcal{G}' = \mathcal{G}_1' \cup \Set{ \left( \GT_j', \Chan_j \right) }_{j \in \indexSet[J]_2} $.
			By Definition~\ref{def:algorithm2} and Theorem~1 of \cite{petersWagnerNestmann19}, then $ \MapSGPS{\Set{ \PT' }}{\mathcal{G}'} \SSequiv \SST' $.
		\item[Case of Rule~\textsf{If-T}:]
			In this case we have $ \SST = \SSys{\Vect{V}}{\SCase{\cond}{\ST_1}{\ST_2}} $ and $ \SST' = \SSys{\Vect{V}}{\ST_1} $.
			By Definition~\ref{def:algorithm2}, then $ \PT \equiv \PRes{\tilde{\Chan}}{\left( \PPar{\PCond{\cond}{\PT_1}{\PT_2}}{\PT[Q]} \right)} $, where substitutions of variables to indicate their actor in Case~\ref{algo2:GCom} do not change \cond because the conditional is unguarded, \ie not under an input binder.
			By Figure~\ref{fig:reductionSemantics}, then $ P \Step \PT' = \PRes{\tilde{\Chan}}{\left( \PPar{\PT_1}{\PT[Q]} \right)} $, where we do not need to apply alpha conversion.
			By the Lemmata~\ref{lem:subjectReduction} and \ref{lem:algo2PCond}, then $ \PT' $ is well-typed \wrt $ \mathcal{G} $.
			By Definition~\ref{def:algorithm2} and Theorem~1 of \cite{petersWagnerNestmann19}, then $ \MapSGPS{\Set{ \PT' }}{\mathcal{G}} = \SST' $.
		\item[Case of Rule~\textsf{If-F}:]
			This case is similar to the previous case.
		\item[Case of Rule~\textsf{Struc}:]
			In this case $ \SST \SSequiv \SST_2 $, $ \SST_2 \Step \SST_2' $, and $ \SST' \SSequiv \SST_2' $.
			Let $ \SST_2 = \SSys{\Vect{V}}{\ST_2} $.
			By Lemma~\ref{lem:structuralCongruence} and since $ \Sequiv \; \subseteq \; \equiv $, then $ \ST_2 $ is well-typed \wrt $ \mathcal{G} $.
			By the induction hypothesis, then there are $ \PT', \mathcal{G}' $ such that $ \PT \Step \PT' $, $ \PT' $ is well-typed \wrt $ \mathcal{G}' $, and $ \MapSGPS{\Set{ \PT' }}{\mathcal{G}'} \SSequiv \SST_2' $.
			Because of $ \SST' \SSequiv \SST_2' $, then $ \MapSGPS{\Set{ \PT' }}{\mathcal{G}'} \SSequiv \SST' $.
	\end{description}
\end{proof}

\begin{figure}[tp]
	\[ \begin{array}{c}
		\left( \textsf{Com} \right) \dfrac{j \in \indexSet}{\GCom{\Role_1}{\Role_2}{\Set{ \GLab{\Label_i}{\tilde{\Sort}_i}{\GT_i} }_{i \in \indexSet}} \Step \GT_j}
		\vspace{0.75em}\\
		\left( \textsf{Par-L} \right) \dfrac{\GT_1 \Step \GT_1'}{\GPar{\GT_1}{\GT_2} \Step \GPar{\GT_1'}{\GT_2}}
		\hspace{2em}
		\left( \textsf{Par-R} \right) \dfrac{\GT_2 \Step \GT_2'}{\GPar{\GT_1}{\GT_2} \Step \GPar{\GT_1}{\GT_2'}}
	\end{array} \]
	\vspace{-1em}
	\caption{Reduction Rules of Global Types.}
	\label{fig:reductionRulesGT}
\end{figure}

We define a reduction semantics for global types in Figure~\ref{fig:reductionRulesGT}, where we equate global types by the rules: $ \left( \GPar{\GEnd}{\GT} \right) = \GT = \left( \GPar{\GT}{\GEnd} \right) $ and $ \GRec{\TVar}{\GT} = \GT\Set{ \Subst{\GRec{\TVar}{\GT}}{\TVar} } $.
We show that well-typed processes can follow the reductions of global types.

\begin{lemma}
	\label{lem:SGPsimulateGT}
	Let $ \mathcal{G} = \Set{ \left( \GT_j, \Chan_j \right) }_{j \in \indexSet[J]} $, $ l \in \indexSet[J] $, and let $ \GT_l $ be guarded by a communication that does not depend on another session.
	If \PT is well-typed \wrt $ \mathcal{G} $ and $ \GT_l \Step \GT_l' $ then there is some $ \PT' $ such that $ \mathcal{G'} = \Set{ \left( \GT_j, \Chan_j \right) }_{j \in \indexSet[J] \setminus \Set{ l }} \cup \Set{ \GT_l' } $, $ \MapSGPS{\Set{ \PT }}{\mathcal{G}} \Step^* \MapSGPS{\Set{ \PT' }}{\mathcal{G}'} $, $ \PT \Step^* \PT' $, and $ \PT' $ is well-typed \wrt $ \mathcal{G}' $.
\end{lemma}

\begin{proof}
	Assume that \PT is well-typed \wrt $ \mathcal{G} $ and $ \GT_l \Step \GT_l' $ for some $ l \in \indexSet[J] $ and the communication that is reduced in $ \GT_l \Step \GT_l' $ does not depend on another session.
	We proceed by an induction on the reduction rules of Figure~\ref{fig:reductionRulesGT} that are used to obtain $ \GT_l \Step \GT_l' $.
	\begin{description}
		\item[Case of Rule~\textsf{Com}:]
			In this case $ \GT_l = \GCom{\Role_1}{\Role_2}{\Set{ \GLab{\Label_i}{\tilde{\Sort}_i}{\GT_{l, i}} }_{i \in \indexSet}} $, $ j \in \indexSet $, and $ \GT_l' = \GT_{l, j} $.
			By the Lemmata~\ref{lem:algo2GCom} and \ref{lem:algo2PReq}, then $ \PT \Steps \equiv \PRes{\tilde{\Chan}}{\left( \PPar{\PPar{\PT_1}{\PT_2}}{\PT_3} \right)} $ such that the steps only initialise sessions, $ m \in \indexSet $, $ \indexSet \subseteq \indexSet[J] $, every conditional branch of $ \PT_1 $ is a version of $ \PSend[\Chan_{\mathnormal{l}}]{\Role_1}{\Role_2}{\Label_m}{\tilde{\expr}}{\PT[Q]} $, every conditional branch of $ \PT_2 $ is a version of $ \PGet[\Chan_{\mathnormal{l}}]{\Role_2}{\Role_1}{\Set{ \PLab{\Label_i}{\tilde{\Args}_i}{\PT[Q]_i} }_{i \in \indexSet[J]}} $, and $ \left( \PPar{\PPar{\PT[Q]_j\Set{ \Subst{\tilde{\Args}_m@\Actor{s_{\mathnormal{l}}}{r_2}}{\tilde{\Args}_j} }}{\PT[Q]}}{\PT[Q]} \right) $ is well-typed \wrt $ \mathcal{G}' = \Set{ \left( \GT_j, \Chan_j \right) }_{j \in \indexSet[J] \setminus \Set{ l }} \cup \Set{ \GT_{l, j} } $.
			Since \PT is well-typed \wrt $ \mathcal{G} $, then $ m = j $.
			By Definition~\ref{def:algorithm2}, then $ \MapSGPI{\Set{ \PT }}{\mathcal{G}} $ is such that every of its conditional branches contains a version of $ \SAssign{\tilde{v}}{\tilde{\expr}'}{\MapSGPI{\Set{ \PPar{\PPar{\PT[Q]}{\PT[Q]_j}}{\PT_3} }}{\mathcal{G}'}} $.
			By the reduction semantics, then we can reduce the guarding conditionals and the value assignment in the respective branch such that $ \MapSGPS{\Set{ \PT }}{\mathcal{G}} \Step^* \SST' = \MapSGPS{\Set{ \PT[Q]_j\Set{ \Subst{\tilde{\Args}_m@\Actor{s_{\mathnormal{l}}}{r_2}}{\tilde{\Args}_j} }, \PT[Q], \PT_3 }}{\mathcal{G}'} $.
			Then, there is some $ \PT' $ such that $ \PT \Step^* \PT' $ initialises sessions, reduces the same conditionals, and then performs the communication step, \ie we have $ \PT' \equiv \PRes{\tilde{\Chan}}{\left( \PPar{\PPar{\PT[Q]_j\Set{ \Subst{\tilde{\expr}}{\tilde{\Args}_j} }}{\PT[Q]}}{\PT_3} \right)} $.
			By Figure~\ref{fig:typingRules}, then $ \PT' $ is well-typed \wrt $ \mathcal{G}' $.
			Since $ \SST' \SSequiv \MapSGPI{\Set{ \PT' }}{\mathcal{G}'} $, then $ \MapSGPS{\Set{ \PT }}{\mathcal{G}} \Step^* \MapSGPS{\Set{ \PT' }}{\mathcal{G}'} $.
		\item[Case of Rule~\textsf{Par-L}:]
			In this case $ \GT_l = \GPar{\GT_{l1}}{\GT_{l2}} $, $ \GT_{l1} \Step \GT_{l1}' $, and $ \GT_l' = \GPar{\GT_{l1}'}{\GT_{l2}} $.
			By Lemma~\ref{lem:algo2GPar}, then $ \PT \equiv \PRes{\tilde{\Chan}}{\left( \PPar{\PT_1}{\PT_2} \right)} $ such that $ \Roles{\GT_{l1}} \subseteq \Roles{\PT_1} $, $ \Roles{\GT_{l2}} \subseteq \Roles{\PT_2} $.
			By Figure~\ref{fig:reductionSemantics}, then $ \PT \Step^* \PT' = \PPar{\PT_1'}{\PT_2} $.
			By Theorem~\ref{lem:subjectReduction} and Figure~\ref{fig:typingRules}, then $ \PT' $ is well-typed \wrt $ \mathcal{G}' = \Set{ \left( \GT_j, \Chan_j \right) }_{j \in \indexSet[J] \setminus \Set{ l }} \cup \Set{ \GT_{l1}', \GT_{l2} } $.
			By Definition~\ref{def:algorithm2} and the reduction semantics, then $ \MapSGPS{\Set{ \PT }}{\mathcal{G}} \Step^* \MapSGPS{\Set{ \PT' }}{\mathcal{G}'} $.
		\item[Case of Rule~\textsf{Par-R}:]
			This case is similar to the case above.
	\end{description}
\end{proof}

The reverse direction of Theorem~3 of \cite{petersWagnerNestmann19} does not hold.
Intuitively, a well-typed system and its sequentionalisation into a \SGP-system have the same steps, but \SGP-systems may force an order on steps that are unordered in the original system.
This happens for global types such as $ \GCom{\Role[1]}{\Role[2]}{\GLab{\Label}{\mathbb{N}}{\GCom{\Role[3]}{\Role[4]}{\GLab{\Label}{\mathbb{N}}{\GEnd}}}} $ that combine causally unrelated communications sequentially.

\begin{example}
	\label{exa:prob_run}
	Consider the global type $ \GT = \GCom{1}{2}{\GLab{\Label}{\mathbb{N}}{\GCom{3}{4}{\GLab{\Label}{\mathbb{N}}{\GEnd}}}} $ that consists of two causally independent communications.
	The system
	\begin{align*}
		\PT ={} & \PReq{2..4}{\PSend{1}{2}{l}{5}{\PEnd}} \mid \PAcc{2}{\PGet{2}{1}{\PLab{l}{x}{\PEnd}}}\\
		& \mid \PAcc{3}{\PSend{3}{4}{l}{4}{\PEnd}} \mid \PAcc{4}{\PGet{3}{4}{\PLab{l}{x}{\PEnd}}}
	\end{align*}
	is a well-typed implementation of this global type.
	The algorithm of Definition~\ref{def:algorithm} maps this process to the \SGP-system $ \MapSGPS{\PT}{\GT} = \SSys{\left( \Args_2, \Args_4 \right)}{\ST} $, where $ \ST = \tau.\SAssign*{x_{2}}{5}{\SAssign*{x_{4}}{4}{\SEnd}} $.
	The process \PT has, modulo structural congruence, two maximal runs
	\begin{center}
		\begin{tikzpicture}[]
			\node (a) at (0, 0.5) {$ \PT $};
			\node (b) at (1, 0.5) {$ \PT' $};
			\node (c) at (5, 1) {$ \PRes{\Chan}{\left( \PSend{3}{4}{l}{4}{\PEnd} \mid \PGet{3}{4}{\PLab{l}{x}{\PEnd}} \right)} $};
			\node (d) at (5, 0) {$ \PRes{\Chan}{\left( \PSend{1}{2}{l}{5}{\PEnd} \mid \PGet{2}{1}{\PLab{l}{x}{\PEnd}} \right)} $};
			\node (e) at (9, 0.5) {$ \PEnd $};
			\path[|->] (a) edge (b);
			\path[|->] (b) edge (2.5, 1);
			\path[|->] (b) edge (2.5, 0);
			\path[|->] (7.5, 1) edge (e);
			\path[|->] (7.5, 0) edge (e);
		\end{tikzpicture}
	\end{center}
	where $ \PT' = \PRes{\Chan}{\left( \PSend{1}{2}{l}{5}{\PEnd} \mid \PGet{2}{1}{\PLab{l}{x}{\PEnd}} \mid \PSend{3}{4}{l}{4}{\PEnd} \mid \PGet{3}{4}{\PLab{l}{x}{\PEnd}} \right)} $.
	But the abstraction \MapSGPS{\PT}{\GT} simulates only the sequence of steps at the top
	\begin{align*}
		& \SSys{\left( \Args_2 = 0, \Args_4 = 0 \right)}{\ST}
		\Step \SSys{\left( \Args_2 = 0, \Args_4 = 0 \right)}{\SAssign*{x_{2}}{5}{\SAssign*{x_{4}}{4}{\SEnd}}}\\
		& \Step \SSys{\left( \Args_2 = 5, \Args_4 = 0 \right)}{\SAssign*{x_{4}}{4}{\SEnd}}
		\Step \SSys{\left( \Args_2 = 5, \Args_4 = 4 \right)}{\SEnd}
	\end{align*}
	in that first process \Role[2] receives the value $ 5 $---and the \SGP-process accordingly updates the variable $ \Args_{\Role[2]} $ of \Role[2]---and then \Role[4] receives the value $ 4 $.
\end{example}

Nonetheless, we can show that each step of the original system can be completed into a sequence that can be simulated.
Assume a step $ \PT \Step \PT' $ of the original system.
We need to find a way to simulate this step in the \SGP-system.
Well-typedness of \PT ensures that the step $ \PT \Step \PT' $ respects the specification, \ie the global types, of this process.
Accordingly, either the types are not influenced by the step $ \PT \Step \PT' $ or the step reduces some part of the global types.
In the first case, the step $ \PT \Step \PT' $ is a session initialisation or reduces a conditional and is simulated by an empty value update or the corresponding reduction of a \SGP-conditional.
If the step $ \PT \Step \PT' $ reduces an unguarded part of one of its global types, \ie performs a communication within a session, then it is simulated by the corresponding value updates in the \SGP-system.
Otherwise, the situation is as described in Example~\ref{exa:prob_run}, \ie we have to find an extension $ \PT' \Steps \PT'' $ of the step $ \PT \Step \PT' $ such that the sequence $ \PT \Steps \PT'' $ can be simulated by the \SGP-system.
Therefore, we reduce all guards in the global types that are necessary to unguard the part of the global type that is reduced in the step $ \PT \Step \PT' $ or that on that this guard depends.
By Lemma~\ref{lem:SGPsimulateGT}, $ \PT' $ can reduce accordingly in a sequence $ \PT' \Steps \PT'' $.
The \SGP-system can simulate steps of the global type by construction.
Thus, we can relate $ \PT'' $ to the corresponding reduction of the \SGP-system.

Theorem~4 of \cite{petersWagnerNestmann19}:
\begin{quote}
	Let $ \mathcal{G} = \Set{ \left( \GT_j, \Chan_j \right) }_{j \in \indexSet[J]} $.
If \PT is well-typed \wrt $ \mathcal{G} $ then for all $ \PT \Step \PT' $ there exist $ \PT'', \mathcal{G}'' $ such that $ \PT' \Step^* \PT'' $, $ \PT'' $ is well-typed \wrt $ \mathcal{G}'' $, and $ \MapSGPS{\Set{ \PT }}{\mathcal{G}} \Step^* \MapSGPS{\Set{ \PT'' }}{\mathcal{G}''}  $.
\end{quote}

\begin{proof}[Proof of Theorem~4 of \cite{petersWagnerNestmann19}]
	Assume that \PT is well-typed \wrt $ \mathcal{G} $ and $ \PT \Step \PT' $.
	By Figure~\ref{fig:reductionSemantics}, $ \PT \Step \PT' $ uses exactly one of the axioms:
	\begin{description}
		\item[Case of Rule~\textsf{Link}:]
			In this case:
			\begin{align*}
				\PT &\equiv \PRes{\tilde{\Chan}}{\left( \PPar{\PReq{\Role[2]..\Role[n]}{\PT_1}}{\PPar{\PAcc{\Role[2]}{\PT_2}}{\PPar{\ldots}{\PAcc{\Role[n]}{\PT_n}}}} \mid \PT[Q] \right)}\\
				\PT' &\equiv \PRes{\tilde{\Chan}, \Chan'}{\left( \PPar{\PT_1}{\PT_2}{\PPar{\ldots}{\PT_n}} \mid \PT[Q] \right)}
			\end{align*}
			By Definition~\ref{def:algorithm2}, we have $ \MapSGPI{\Set{ \PT }}{\mathcal{G}} = \tau.\MapSGPI{\Set{ \PT_1, \ldots, \PT_n, \PT[Q] }}{\mathcal{G}} $ and we have $ \MapSGPI{\Set{ \PT' }}{\mathcal{G}} = \MapSGPI{\Set{ \PT_1, \ldots, \PT_n, \PT[Q] }}{\mathcal{G}} $.
			By reflexivity, $ \PT' \Step^* \PT' $.
			By Lemma~\ref{lem:algo2PReq}, $ \PT' $ is well-typed \wrt $ \mathcal{G} $.
			By the reduction semantics, $ \MapSGPS{\Set{ \PT }}{\mathcal{G}} \Step \MapSGPS{\Set{ \PT' }}{\mathcal{G}} $.
		\item[Case of Rule~\textsf{Com}:]
			In this case
			\begin{align*}
				\PT &\equiv \PRes{\tilde{\Chan}}{\left( \PPar{\PPar{\PSend{\Role_1}{\Role_2}{\Label_j}{\tilde{\expr}}{\PT_1}}{\PGet{\Role_2}{\Role_1}{\Set{ \PLab{\Label_i}{\tilde{\Args}_i}{\PT_{2, i}} }_{i \in \indexSet}}}}{\PT_3} \right)}\\
				\PT' &\equiv \PRes{\tilde{\Chan}}{\left( \PPar{\PPar{\PT_{2, j}\Set{ \Subst{\tilde{\expr}}{\tilde{\Args}_j} }}{\PT_1}}{\PT_3} \right)}
			\end{align*}
			and $ j \in \indexSet $.
			By Definition~\ref{def:algorithm2}, then $ \MapSGP{\Set{ \PT }}{\mathcal{G}} $ maps this communication on a \SGP-value-assignment but may guard it by other conditionals from $ \PT_3 $ or value assignments due to communications in $ \PT_3 $.
			Note that, therefore, all these conditionals and communication prefixes have to be consecutively unguarded in the remainder of $ \PT_3 $ and that the communications are captured in $ \mathcal{G} $.
			By Figure~\ref{fig:reductionRulesGT} and Definition~\ref{def:algorithm2}, then there exists $ \GT_l'' $ such that $ \GT_l \Step^* \GT_l'' $ for some $ l \in \indexSet[J] $ reduces the communications in $ \PT_3 $ that correspond to the value assignments that guard \SGP-value-assignment $ \SAssign{\tilde{\Args}_j@\Actor{s}{r_2}}{\tilde{\expr}}{\MapSGP{\Set{ \PT_{2, j}\Set{ \Subst{\tilde{\Args}_j@\Actor{s}{r_2}}{\tilde{\Args}_j} }, \PT_1, \PT_3' }}{\mathcal{G}''}} $.
			By Lemma~\ref{lem:SGPsimulateGT}, then there is some $ \PT'' $ such that we have $ \MapSGPS{\Set{ \PT }}{\mathcal{G}} \Step^* \MapSGPS{\Set{ \PT'' }}{\mathcal{G}''} $ and $ \PT \Step^* \PT'' $.
			Since well-typedness ensures that there are no conflicts and since $ \PPar{\PSend{\Role_1}{\Role_2}{\Label_j}{\tilde{\expr}}{\PT_1}}{\PGet{\Role_2}{\Role_1}{\Set{ \PLab{\Label_i}{\tilde{\Args}_i}{\PT_{2, i}} }_{i \in \indexSet}}} $ is reduced in both of the sequences $ \PT \Step \PT' $ and $ \PT \Step^* \PT'' $, then confluence implies that also $ \PT' \Step^* \PT'' $.
		\item[Case of Rule~\textsf{If-T}:]
			In this case $ \PT \equiv \PRes{\tilde{\Chan}}{\left( \PPar{\PCond{\cond}{\PT_1}{\PT_2}}{\PT_3} \right)} $ and $ \PT' \equiv \PRes{\tilde{\Chan}}{\left( \PPar{\PT_1}{\PT_3} \right)} $.
			By Definition~\ref{def:algorithm2}, $ \MapSGPI{\Set{ \PT }}{\mathcal{G}} $ maps this conditional on a \SGP-conditional but may guard it by other conditionals from $ \PT_3 $ or value assignments due to communications in $ \PT_3 $.
			Note that all these conditionals and communication prefixes have to be consecutively unguarded in the remainder of $ \PT_3 $ and that the communications are captured in $ \mathcal{G} $.
			By Figure~\ref{fig:reductionRulesGT} and Definition~\ref{def:algorithm2}, then there exists $ \mathcal{G}'' $ such that $ \GT_l \Step^* \GT_l'' $ reduces the communications in $ \PT_3 $ that correspond to the value assignments that guard the \SGP-conditional $ \SCase{\cond}{\MapSGPI{\Set{ \PT_1, \PT_3' }}{\mathcal{G}''}}{\MapSGP{\Set{ \PT_2, \PT_3' }}{\mathcal{G}''}} $, where $ \mathcal{G}'' = \Set{ \left( \GT_j, \Chan_j \right) }_{j \in \indexSet[J] \setminus \Set{ l }} \cup \Set{ \GT_l'' } $.
			By Lemma~\ref{lem:SGPsimulateGT}, then there is $ \PT'' $ such that $ \MapSGPS{\Set{ \PT }}{\mathcal{G}} \Step^* \MapSGPS{\Set{ \PT'' }}{\mathcal{G}''} $ and $ \PT \Step^* \PT'' $.
			Since well-typedness ensures that there are no conflicts and since $ \PCond{\cond}{\PT_1}{\PT_2} $ is reduced in both of the sequences $ \PT \Step \PT' $ and $ \PT \Step^* \PT'' $, then confluence implies that also $ \PT' \Step^* \PT'' $.
		\item[Case of Rule~\textsf{If-F}:]
			This case is similar to the case above.
	\end{description}
\end{proof}

Interestingly, the combination of Theorem~3 and Theorem~4 of \cite{petersWagnerNestmann19} is similar to \emph{(weak) operational correspondence} as it is introduced in \cite{gorla10} as criterion for the quality of encodings.
Encodings are mappings from a source language $ \procS $ into a target language $ \procT $.

\begin{definition}[Weak Operational Correspondence, \cite{petersGlabbeek15}]
	An encoding \newline $ \encoding: \procS \to \procT $ is \emph{weakly operationally corresponding} \wrt $ \RelT $ if it is:
	\begin{description}
		\item[Complete:] $ \forall S, S' \logdot S \Step^* S' $ implies $ \left( \exists T \logdot \Encoding{S} \Step^* T \wedge \left( \Encoding{S'}, T \right) \in \RelT \right) $
		\item[Weakly Sound:] $ \forall S, T \logdot \Encoding{S} \Step^* T $ implies\\
			$ \left( \exists S', T' \logdot S \Step^* S' \wedge T \Step^* T' \wedge \left( \Encoding{S'}, T' \right) \in \RelT \right) $
	\end{description}
\end{definition}

We observe that completeness is similar to Theorem~3 of \cite{petersWagnerNestmann19} and weak soundness is similar to Theorem~4 of \cite{petersWagnerNestmann19}, but with the roles of the languages exchanged.
Accordingly, we change the above definition and use a weak variant of completeness.

\begin{definition}
	An encoding $ \encoding: \procS \to \procT $ is \emph{reversed weakly operationally corresponding} \wrt $ \RelT \subseteq \procT^2 $ if it is:
	\begin{description}
		\item[Weakly Complete:] $ \forall S, S' \logdot S \Step^* S' $ implies\\
			$ \left( \exists S'', T'' \logdot S' \Step^* S'' \wedge \Encoding{S} \Step^* T'' \wedge \left( \Encoding{S''}, T'' \right) \in \RelT \right) $
		\item[Sound:] $ \forall S, T \logdot \Encoding{S} \Step^* T $ implies $ \left( \exists S' \logdot S \Step^* S' \wedge \left( \Encoding{S'}, T \right) \in \RelT \right) $
	\end{description}
\end{definition}

Then the mapping $ \MapSGPS{\cdot}{\cdot} $ from well-typed processes into \SGP-systems is reversed weakly operationally corresponding \wrt $ \SSequiv $.
The paper \cite{petersGlabbeek15} relates weak operational correspondence with so-called correspondence simulation.

\begin{definition}[Correspondence Simulation, \cite{petersGlabbeek15}]
	\label{def:correspondenceSimulation}
	A relation \Rel is a \emph{(weak reduction) correspondence simulation} if for each $ \left( \PT, \PT[Q] \right) \in \Rel $:
	\begin{itemize}
		\item $ \PT \Steps \PT' $ implies $ \exists \PT[Q]'\logdot \PT[Q] \Steps \PT[Q]' \land \left( \PT', \PT[Q]' \right) \in \Rel $
		\item $ \PT[Q] \Steps \PT[Q]' $ implies $ \exists \PT'', \PT[Q]''\logdot \PT \Steps \PT'' \land \PT[Q]' \Steps \PT[Q]'' \land \left( \PT'', \PT[Q]'' \right) \in \Rel $
	\end{itemize}
	Two terms $ \PT, \PT[Q] $ are \emph{correspondence similar}, denoted as $ \PT \corrSim \PT[Q] $, if a correspondence simulation relates them.
\end{definition}

With a similar argumentation as in \cite{petersGlabbeek15} to show that if $ \encoding $ is weakly operationally corresponding \wrt a correspondence simulation then $ S \corrSim \Encoding{S} $, we conclude that if $ \encoding $ is reversed weakly operationally corresponding \wrt a correspondence simulation then $ \Encoding{S} \corrSim S $.
By the Theorems~3 and 4 of \cite{petersWagnerNestmann19}, then the sequentialisation of a system is correspondence similar to the system.


\section{Examples}
\label{sec:examples}

\subsection{Toy Example}
\label{sec:toy-example}

Similar to the two Buyer example of \cite{hondaYoshidaCarbone08}, we illustrate our approach by a small example of an auctioneer system consisting of an auctioneer \Role[A] and two alternating bidders \Role[B1] and \Role[B2].
The two bidders alternate in offering bids towards the auctioneer and the auctioneer continues to inform the next bidder about the last bid until the current bid exceeds the maximum of one of the bidders.
As soon as one bidder refuses to offer another bid, the auctioneer informs the respective other bidder that the item was sold to him.
We illustrate the communication structure for the case that \Role[B2] wins the auction:
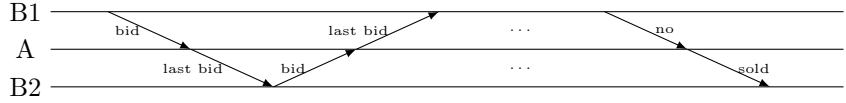
\begin{center}
\begin{tikzpicture}[yscale=0.5, xscale=1.1, node distance=0.3cm, auto]

  \newcounter{x}\newcounter{y}
  \newlength{\mpsthlp}
  \settowidth{\mpsthlp}{$\Role[B1]$}

  \node (N-0-0) at (0,0) {$\Role[B1]$};
  \node (N-0-1) at (0,-1) {\makebox[\mpsthlp]{$\Role[A]$}};
  \node (N-0-2) at (0,-2) {$\Role[B2]$};

  \foreach \x in {0,1,2}
  {
    \foreach \y in {1,...,10}
    {
      \node (N-\y-\x) at (\y,-\x) {};
    }
  }

  \foreach \x in {0,1,2}
  {
    \path (N-0-\x) edge[-] node {} (N-10-\x);
  }

  \node (1) at (6,-0.5) {\tiny \ldots};
  \node (1) at (6,-1.5) {\tiny \ldots};

  \path (N-1-0.center) edge[-latex] node[anchor=east] {\tiny bid} (N-2-1.center);
  \path (N-2-1.center) edge[-latex] node[anchor=east] {\tiny last bid} (N-3-2.center);
  \path (N-3-2.center) edge[-latex] node[anchor=east] {\tiny bid} (N-4-1.center);
  \path (N-4-1.center) edge[-latex] node[anchor=east] {\tiny last bid} (N-5-0.center);

  \path (N-7-0.center) edge[-latex] node[anchor=west] {\tiny no} (N-8-1.center);
  \path (N-8-1.center) edge[-latex] node[anchor=west] {\tiny sold} (N-9-2.center);
\end{tikzpicture}
\end{center}

Assume a program that implements such an auctioneer system.
An analysis of such a program may want to check \eg whether the bidders indeed alternate in offering bids, \ie no bidder is allowed or forced to bid twice without the other bidder in between, or whether no bid exceeds the internal maximum of a bidder, \ie the amount that he or she is willing to pay.
The former property is clearly a property of the communication structure and can easily be checked with \MPST.
The latter property, however, requires to analyse concrete data.
Since the maximum a bidder is willing to pay is some data that is specific to the bidder or may even be specific to a concrete run, this property does not fit into the set of static properties \MPST were designed for.
We show that if one is willing to pay the price of verifying the communication structure of the program with \MPST, one gets as a side-effect a massive reduction in checking properties about the state of the program, \ie properties that require the consideration of concrete runs or concrete values of variables.

An important property that often requires the consideration of concrete runs or concrete values is termination.
\MPST ensure progress for well-typed systems, \ie there are no deadlocks and all runs of the system will follow its specification that is provided by the global type(s).
Progress immediately implies termination, if the considered system does not contain recursion.
But, as in our toy example, many algorithms to compute some value or some decision, rely on a loop that runs until a suitable value was found or a decision was made.
The progress property, that we can obtain by \MPST for such cases, is a crucial argument for a proof of termination, but does not directly imply termination.
The presented method allows us to prove termination, by analysing the evolution of data in concrete runs, automatically and in an efficient way.

To provide a global type for our example of the auctioneer system, we use different labels to convey the intention of actions of the participants:
\Label[bid] indicates a new bid,
\Label[no] indicates that the bidder refuses to make another bid,
\Label precedes the forwarding of the last bid, and
\Label[s] indicates that the item was sold.
Since the only kind of values that are transmitted in this protocol are bids, we use \Sort[Int] as only sort for integer values.
Let $ \Role[A] = \Role[1] $ be the role of the auctioneer, $ \Role[B1] = \Role[2] $ be the role of the first bidder, and $ \Role[B2] = \Role[3] $ be the role of the second bidder.
The global type \GAsys describes the communication structure of this example from a global point of view.

\begin{example}[Global Type of the Auctioneer System]
	\label{exa:globalTypeAuctioneer}
	\vspace{-0.5em}
	\begin{align*}
		& \GAsys ={} \GCom{\Role[B1]}{\Role[A]}{\GLab{\Label[bid]}{\Sort[Int]}{\GCom{\Role[A]}{\Role[B2]}{\GLab{\Label}{\Sort[Int]}{\GRec{\TVar}{}}}}}\\
		& \big( \GCom{\Role[B2]}{\Role[A]}{} \big\{
		\begin{array}[t]{l}
			\GLab{\Label[bid]}{\Sort[Int]}{\GCom{\Role[A]}{\Role[B1]}{\GLab{\Label}{\Sort[Int]}{\GCom{\Role[B1]}{\Role[A]}{}}}} \{
				\begin{array}[t]{l}
					\GLab{\Label[bid]}{\Sort[Int]}{\GCom{\Role[A]}{\Role[B2]}{\GLab{\Label}{\Sort[Int]}{\TVar}}},\\
					\GLab{\Label[no]}{}{\GCom{\Role[A]}{\Role[B2]}{\GLab{\Label[s]}{\Sort[Int]}{\GEnd}}} \},
				\end{array}\\
			\GLab{\Label[no]}{}{\GCom{\Role[A]}{\Role[B1]}{\GLab{\Label[s]}{\Sort[Int]}{\GEnd}}} \big\} \big)
		\end{array}
	\end{align*}
\end{example}

An example of a well-typed implementation of the global type \GAsys of Example~\ref{exa:globalTypeAuctioneer} is given below, \ie \PAsys is well-typed \wrt \GAsys.
The names $ inc $ and $ max $ are place-holders for the actual functions and natural constants that are provided by \textsf{Promela}.
The functions $ inc_{\Role[B1]}, inc_{\Role[B2]} $ are used by the respective bidder to increase the last bid and the constants $ max_{\Role[B1]}, max_{\Role[B2]} $ denote the maximum a bidder is willing to pay.

\begin{example}[Implementation of the Auctioneer System]
	\label{exa:implementationAuctioneer}
	\vspace{-0.5em}
	\begin{align*}
		\PAsys ={}
		& \PReq{B1,B2}{\PGet{\Role[A]}{\Role[B1]}{\PLab{bid}{b}{\PSend{A}{B2}{l}{b}{ \PRep{\PVar_{\Role[A]}}{} }}}}\\
		& \hspace{2em} \PGet{A}{B2}{ \{
			\begin{array}[t]{l}
				\PLab{bid}{b}{\PSend{A}{B1}{l}{b}{\PGet{A}{B1}{ \{
					\begin{array}[t]{l}
						\PLab{bid}{b}{\PSend{A}{B2}{l}{b}{\PVar_{\Role[A]}}} \\
						\PLab{no}{}{\PSend{A}{B2}{s}{b}{\PEnd}} \}
					\end{array} }}} \\
				\PLab{no}{}{\PSend{A}{B1}{s}{b}{\PEnd}} \} \\
			\end{array} }\\
		\mid \; & \PAcc{B1}{\PSend{B1}{A}{bid}{inc_{\Role[B1]}(0)}{\PRep{\PVar_{\Role[B1]}}{}}}\\
		& \hspace{2em} \PGet{B1}{A}{} \{
			\begin{array}[t]{l}
				\PLab{l}{b}{}
					\begin{array}[t]{l}
						\mathsf{if}~inc_{\Role[B1]}(b) \leq \maxB1
							\begin{array}[t]{l}
								\mathsf{then}~\PSend{B1}{A}{bid}{inc_{\Role[B1]}(b)}{\PVar_{\Role[B1]}} \\
								\mathsf{else}~ \PSend{B1}{A}{no}{}{\PEnd}
							\end{array}
					\end{array}\\
				\PLab{s}{b}{\PEnd} \}
			\end{array}\\
		\mid \; & \PAcc{B2}{\PGet{B2}{A}{\PLab{l}{b}{\PRep{\PVar_{\Role[B2]}}{}}}}\\
		& \hspace{2em} \mathsf{if}~ inc_{\Role[B2]}(b) \leq \maxB2
			\begin{array}[t]{l}
				\mathsf{then}~\PSend{B2}{A}{bid}{inc_{\Role[B2]}(b)}{\PGet{B2}{A}{}} \{
					\begin{array}[t]{l}
						\PLab{l}{b}{\PVar_{\Role[B2]}} \\
						\PLab{s}{b}{\PEnd} \}
					\end{array}\\
				\mathsf{else}~\PSend{B2}{A}{no}{}{\PEnd}
			\end{array}
	\end{align*}
\end{example}

To sequentialise the given implementation utilising our algorithm, let \PT be a process that is well-typed \wrt a global type \GT and $ \ST = \MapSGP{\Set{ \PT }}{\GT} $.
Then the corresponding \SGP-system is $ \MapSGPS{\Set{ \PT }}{\GT} = \SSys{\Vect{V}}{\ST} $, where $ \Vect{V} $ is the vector of names in \ST.
Accordingly, the auctioneer system \PAsys of Example~\ref{exa:implementationAuctioneer} that is well-typed \wrt \GAsys in Example~\ref{exa:globalTypeAuctioneer} translates into the \SGP-system $ \SSys{ \left( b_{\Role[A]}, b_{\Role[B1]}, b_{\Role[B2]} \right) }{\SAsys} $, where $ \SAsys = \MapSGP{\Set{ \PAsys }}{\GAsys} $ is given below.

\begin{example}[Sequentialisation of the Auctioneer System]
	\label{exa:SGPAuctionieer}
	\vspace{-0.5em}
	\begin{align*}
		\SAsys ={}
		& \tau.\SAssign*{b_{\Role[A]}}{inc_{\Role[B1]}(0)}{\SAssign*{b_{\Role[B2]}}{b_{\Role[A]}}{\SRep{\SVar}{}}}\\
		& \hspace{1em}
			\begin{array}[t]{l}
				\textsf{if } inc_{\Role[B2]}(b_{\Role[B2]}) \leq max_{\Role[B2]}\\
				\textsf{then } \SAssign*{b_{\Role[A]}}{inc_{\Role[B2]}(b_{\Role[B2]})}{\SAssign*{b_{\Role[B1]}}{b_{\Role[A]}}{}}
					\begin{array}[t]{l}
						\textsf{if } inc_{\Role[B1]}(b_{\Role[B1]}) \leq max_{\Role[B1]}\\
						\textsf{then } \SAssign*{b_{\Role[A]}}{inc_{\Role[B1]}(b_{\Role[B1]})}{\SAssign*{b_{\Role[B2]}}{b_{\Role[A]}}{\SVar}}\\
						\textsf{else } \tau.\SAssign*{b_{\Role[B2]}}{b_{\Role[A]}}{\SEnd}
					\end{array}\\
				\textsf{else } \tau.\SAssign*{b_{\Role[B1]}}{b_{\Role[A]}}{\SEnd}
			\end{array}
	\end{align*}
\end{example}

Since the causal relation of the communications of \GAsys in Example~\ref{exa:globalTypeAuctioneer} is a total order, all properties that hold for \SAsys are also satisfied by \PAsys.

\subsection{Translating \SGP-Systems into \textsf{Promela}}
\label{sec:promela}

To illustrate the verification of system properties, we use the model checker \textsf{Spin} \cite{spin,spin91} and implement the \SGP-system \SAsys in Example~\ref{exa:SGPAuctionieer} using \textsf{Promela}, the input language of \textsf{Spin}.
Therefore, we provide an algorithm to translate a \SGP-process into \textsf{Promela} code.

\begin{figure}[t]
\begin{lstlisting}
short bA = 0;
short bB1 = 0;
short bB2 = 0;

short sold = 0;

byte incB1;
byte incB2;

short maxB1;
short maxB2;
\end{lstlisting}
  \vspace{-1.5em}
  \caption{Promela Implementation preamble}\label{imp:pre}
\end{figure}

First we generate a preamble for the \textsf{Promela} program, \ie declare variables and set their initial values.
The variables are obtained from the vector of variables $ \Vect{V} $ in a \SGP-system $ \SSys{\Vect{V}}{S} $.
Sometimes the initial values are directly specified by the implementation or are given as parameters of the implementation.
Otherwise, the developer has to pick suitable initial values respecting their respective sorts.
The preamble for \SAsys of Example~\ref{exa:SGPAuctionieer} is given in Figure~\ref{imp:pre}.
It introduces the three variables $ b_{\Role[A]}, b_{\Role[B1]}, b_{\Role[B2]} $ of the knowledge vector in our example and initialises them with $ 0 $.
Moreover, the preamble introduces five more variables that are used for the implementation and verification of the LTL-Formula that specify the properties we want to check.
We provide in \cite{petersWagnerNestmann19} an algorithm for the translation of \SGP-processes into \textsf{Promela} but expect that the desired properties are already specified as LTL-Formula.

\begin{figure}[t]
\begin{lstlisting}
active proctype Model() {

 atomic {
   select (incB1 : 1..10);
   select (incB2 : 1..10);
   select (maxB1 : 50..100);
   select (maxB2 : 50..100);
 }

 skip;
 bA = incB1 + 0;
 bB2 = bA;

  LX: if
  :: bB2 + incB2 <= maxB2
          -> bA = incB2 + bB2; bB1 = bA;
	    if
	     :: incB1 + bB1 <= maxB1 -> bA = incB1 + bB1;
	                                bB2 = bA; goto LX;
	     :: else                 -> sold = 2; bB2 = bA;
	                                goto LEnd;
	    fi
  :: else -> sold = 1; bB1 = bA; goto LEnd;
 fi
 LEnd:
}
\end{lstlisting}
  \vspace{-1.5em}
  \caption{Promela Implementation of the Auctioneer System}\label{imp:spin}
\end{figure}

Figure~\ref{imp:spin} presents the \textsf{Promela} implementation of \SAsys from Example~\ref{exa:SGPAuctionieer}.
Following the translation into \textsf{Promela} of \cite{petersWagnerNestmann19}, first a \textsf{proctype} with the name Model is introduced.

The $ \tau $ at the beginning of \SAsys that resulted from the translation of the session initialisation is translated to \textsf{skip}.
Then there are two subsequent value updates on $ b_{\Role[A]} $ and $ b_{\Role[B2]} $ that precede the loop.
The loop is introduced by declaring its recursion variable \textsf{LX}.
Then the \textsf{if-then-else} statements with their respective value updates follow.
Another instance of the loop is generated by \textsf{goto LX}, whereas \textsf{goto LEnd} terminates the program by jumping to its end.

In addition Figure~\ref{imp:spin} declares the domains of the variables \textsf{incB1}, \textsf{incB2}, \textsf{maxB1}, and \textsf{maxB2}.
Note that these variables were not specified by the implementation in \PAsys of Example~\ref{exa:implementationAuctioneer} and thus are not provided by \SAsys.
The variables \textsf{incB1} and \textsf{incB2} denote the value by that bidders increment the last bid.
The domain $ 1 .. 10 $ tells us, that this value to increment the last bid is chosen non-deterministically between $ 1 $ and $ 10 $.
The variables \textsf{maxB1} and \textsf{maxB2} specify the internal maximum a bidder is willing to pay.
We use arbitrary values between $ 50 $ and $ 100 $.
Similar to the initial values of the variables in the knowledge vector of \SAsys, we expect that the developer provides suitable domains.

\subsection{Analysing the Properties of Implementations}
\label{sec:properties}

\begin{figure}[t]
\begin{lstlisting}
ltl p0 { eventually always (bA > 0 && bA == bB1 && bA == bB2) }
ltl p1 { (bB1 < bB2 || bB2 < bB1 || (bB1 == 0 && bB2 == 0))
         until (always (bB1 == bB2)) }
ltl p2 { always (bA == bB1 || bA == bB2) }
ltl p3 { always (((sold == 1) -> always (sold == 1))
              && ((sold == 2) -> always (sold == 2))) }
ltl p4 { always (((sold == 1) -> (bA <= maxB1))
              && ((sold == 2) -> (bA <= maxB2))) }
ltl p5 { eventually always (sold > 0) }
\end{lstlisting}
  \vspace{-1.5em}
  \caption{Promela LTL Formulae}\label{imp:ltl}
\end{figure}

Finally, the developer has to add to the \textsf{Promela} program the LTL-formula for the properties that he or she is interested in.
We add the following six LTL-formulae, where Figure~\ref{imp:ltl} presents their \textsf{Promela} representation:
\begin{align*}
	\eventually\always (\bA > 0 \wedge \bA = \bB1 \wedge \bA = \bB2) & \tag{P1}\\
	(\bB1 < \bB2 \vee \bB2 < \bB1 \vee (\bB1 = 0 \wedge \bB2 = 0)) \until (\always (\bB1 = \bB2)) & \tag{P2} \label{notInMalicious}\\
	\always (\bA = \bB1 \vee \bA = \bB2) & \tag{P3} \label{wrong}\\
	\always ((( sold = 1) \rightarrow \always (sold = 1)) \wedge ((sold = 2) \rightarrow \always (sold = 2))) & \tag{P4}\\
	\always (((sold = 1) \rightarrow (\bA \leq \maxB1)) \wedge ((sold = 2) \rightarrow (\bA \leq \maxB2))) & \tag{P5} \label{maximum}\\
	\eventually \always (sold > 0) & \tag{P6} \label{termination}
\end{align*}
These formulae have following meanings:
\begin{enumerate}
\item there exists one state from which onward all participants always have the same bid value,
\item the bidder always bid higher than the other one until both have always the same bid,
\item the auctioneer has always the same bid as one bidder,
\item only one process can win the auction,
\item the winner did not bid more than its limit, and
\item eventually one bidder will win.
\end{enumerate}
As checked by \textsf{SPIN}, only Property~\ref{wrong} is not satisfied.
This is because only one variable assignment can happen at any time, thus updating two variables to a new value would take at least two steps.
More precisely, Property~\ref{wrong} is violated then the auctioneer receives a new bid.
If \Role[A] receives a bid from \Role[B2] then $ b_{\Role[A]} $ is updated with the value $ \textsf{incB2} + b_{\Role[B2]} $.
In the next step, the variable $ b_{\Role[B1]} $ is updated with $ b_{\Role[A]} $ which restores Property~\ref{wrong}.
But in between these two value updates the value of $ b_{\Role[A]} $ is neither equal to $ b_{\Role[B2]} $ (unless $ \textsf{incB2} = 0 $) nor equal to $ b_{\Role[B2]} $.

Above we validated some interesting properties of our toy example such as that no bidder exceeds its internal maximum (Property~\ref{maximum}) or termination (Property~\ref{termination}).
Does that means, that our implementation is a good implementation of the auctioneer algorithm?
Unfortunately, this is not so easy.

Consider a malicious bidder \Role[B2] that instead of increasing the bid always resubmits the last bid of \Role[B1].
This way the chances to win the auction, if the internal maxima of the two bidders are close, is significantly increased by \Role[B2] in a very unfair way.
Moreover, if \Role[B2] is willing to pay more, \ie if he or she is supposed to win the auction, than this strategy ensures that the bid of \Role[B2] also stays below the limit of \Role[B1], \ie reduces the amount of money \Role[B2] needs to pay.
Obviously, this behaviour is malicious and should be rejected by the auctioneer.
To implement this malicious behaviour it suffices to instantiate \textsf{incB2} with $ 0 $.

This kind of malicious behaviour is detected by the properties in Figure~\ref{imp:ltl}.
More precisely, by setting $ \textsf{incB2} = 0 $, the Property~\ref{notInMalicious} is violated.
Thus, we can detect that there are implementations of this algorithm, that are not acceptable.
We can now use this knowledge to improve the specification.
Indeed we observe, that in our specification \GAsys in Example~\ref{exa:globalTypeAuctioneer}, the auctioneer does not decide anything, \ie this specification does not provide the communication structure that the auctioneer need to check the validity of bids and to reject them.

We provide a revised version of this auctioneer system, where the auctioneer checks the validity of the bids and rejects a bid that was not valid.

\begin{example}[Global Type of the Revised Auctioneer System]
	\label{exa:globalFixedTypeAuctioneer}
	\vspace{-0.5em}
	\begin{align*}
		& \GRAsys = \GCom{\Role[B1]}{\Role[A]}{\GLab{\Label[bid]}{\Sort[Int]}{\GCom{\Role[A]}{\Role[B2]}{\GLab{\Label}{\Sort[Int]}{\GRec{\TVar}{}}}}}\\
		& \hspace{1em} \big( \GCom{\Role[B2]}{\Role[A]}{} \big\{
		\begin{array}[t]{l}
			\GLab{\Label[bid]}{\Sort[Int]}{\GCom{\Role[A]}{\Role[B1]}} \{\\
			\hspace{1em} \begin{array}[t]{l}
				\GLab{\Label}{\Sort[Int]}{\GCom{\Role[B1]}{\Role[A]}{}} \{\!
					\begin{array}[t]{l}
						\GLab{\Label[bid]}{\Sort[Int]}{\GCom{\Role[A]}{\Role[B2]}} \{\!
							\begin{array}[t]{l}
								{\GLab{\Label}{\Sort[Int]}{\TVar}},\\
								\GLab{\Label[s]}{\Sort[Int]}{\GCom{\Role[A]}{\Role[B1]}{\GLab{\Label[r]}{}{\GEnd}}} \},
							\end{array}\\
						\GLab{\Label[no]}{}{\GCom{\Role[A]}{\Role[B2]}{\GLab{\Label[s]}{\Sort[Int]}{\GEnd}}} \},
					\end{array}\\
				\GLab{\Label[s]}{\Sort[Int]}{\GCom{\Role[A]}{\Role[B2]}{\GLab{\Label[r]}{}{\GEnd}}} \},
			\end{array}\\
			\GLab{\Label[no]}{}{\GCom{\Role[A]}{\Role[B1]}{\GLab{\Label[s]}{\Sort[Int]}{\GEnd}}} \big\} \big)
		\end{array}
	\end{align*}
\end{example}

Here, the auctioneer can send two kinds of messages, \ie initiates a branching, after receiving a bid.
Either the auctioneer considers the bid as valid and forwards it to the respective other role using a message with label \Label---as we already encountered in the previous example.
Or the auctioneer detects an invalid bid and transmits a message with the label \Label[r], signalling the bidder that the bid was rejected and, hence, that the other bidder won the auction.

\subsection{Promela}
\label{sec:promela-1}

In Example~\ref{exa:globalFixedTypeAuctioneer} we present $ \GRAsys $; a revised version of the global type of our toy example.
The \textsf{Promela} implementation of this global type is presented below.
Note that the differences between the Figure~\ref{imp:spin} and the code below reflect the additional choice implemented by the auctioneer that allows him to reject invalid bids.

\begin{lstlisting}
active proctype Model() {

 atomic {
   select (incB1 : 0..10);
   select (incB2 : 1..10);
   select (maxB1 : 50..100);
   select (maxB2 : 50..100);
 }

 skip;
 bA = incB1 + 0;
 bB2 = bA;

 LX:
 if
 :: bB2 + incB2 <= maxB2
   -> if
      :: incB2 + bB2 <= bA -> sold = 1; bB1 = bA; goto LEnd;
      :: else              -> bA = incB2 + bB2; bB1 = bA;
         if
	 :: incB1 + bB1 <= maxB1 ->
	   if
	   :: incB1 + bB1 <= bA -> sold = 2; bB2 = bA; goto LEnd;
	   :: else        -> bA = incB1 + bB1; bB2 = bA; goto LX;
	   fi
	 :: else                -> sold = 2; bB2 = bA; goto LEnd;
	 fi
      fi
  :: else                 -> sold = 1; bB1 = bA; goto LEnd;
  fi

 LEnd:
}
\end{lstlisting}

Due to the interleaving of independent actions, the state space of a concurrent system is in the worst case exponentially larger than of its sequentialisation.
As an example, we implemented the \emph{Needham-Schroeder public key protocol} with 10 pairs of processes that interact with the same server.

\begin{lstlisting}
mtype:Id    = {S, A0, B0, A1, B1, A2, B2, A3, B3, A4, B4,
               A5, B5, A6, B6, A7, B7, A8, B8, A9, B9};
mtype:Nonce = {Na0, Nb0, Na1, Nb1, Na2, Nb2, Na3, Nb3, Na4,
               Nb4, Na5, Nb5, Na6, Nb6, Na7, Nb7, Na8, Nb8, Na9, Nb9};

chan as = [0] of {mtype:Id, mtype:Id, mtype:Id};
chan sa = [0] of {mtype:Id, mtype:Id, mtype:Id, mtype:Id};
chan ca = [0] of {mtype:Id, short, mtype:Id, mtype:Id};
chan cb = [0] of {mtype:Id, short, mtype:Id};

ltl p0 { always true }

proctype PA (mtype:Id self; mtype:Id other)
{
  short nonce;
  select (nonce: 0..1);
  short g1;

  as ! self, self, other; // 1
  sa ? eval(self), eval(other), eval(other), S; // 2

  ca ! self, nonce, self, other; // 3 
  ca ? eval(other), eval(nonce), g1, eval(self); // 6

  cb ! self, g1, other; // 7
}

proctype PB (mtype:Id self)
{
  short nonce;
  select (nonce: 0..1);
  mtype:Id other, other2;
  short g2;

  ca ? other, g2, other2, eval(self); // 3

  as ! self, self, other; // 4
  sa ? eval(self), eval(other), eval(other), S; // 5

  ca ! self, g2, nonce, other; // 6
  cb ? eval(other), eval(nonce), eval(self); // 7
}

active proctype Model()
{
  atomic
  {
    run PA(A0, B0);
    run PB(B0);

    run PA(A1, B1);
    run PB(B1);

    run PA(A2, B2);
    run PB(B2);

    run PA(A3, B3);
    run PB(B3);

    run PA(A4, B4);
    run PB(B4);

    run PA(A5, B5);
    run PB(B5);

    run PA(A6, B6);
    run PB(B6);

    run PA(A7, B7);
    run PB(B7);

    run PA(A8, B8);
    run PB(B8);

    run PA(A9, B9);
    run PB(B9);
  }

  mtype:Id p1;

  as ? A0, A0, p1; // 1
  sa ! A0, p1, p1, S; // 2
  as ? A1, A1, p1; // 1
  sa ! A1, p1, p1, S; // 2
  as ? A2, A2, p1; // 1
  sa ! A2, p1, p1, S; // 2
  as ? A3, A3, p1; // 1
  sa ! A3, p1, p1, S; // 2
  as ? A4, A4, p1; // 1
  sa ! A4, p1, p1, S; // 2
  as ? A5, A5, p1; // 1
  sa ! A5, p1, p1, S; // 2
  as ? A6, A6, p1; // 1
  sa ! A6, p1, p1, S; // 2
  as ? A7, A7, p1; // 1
  sa ! A7, p1, p1, S; // 2
  as ? A8, A8, p1; // 1
  sa ! A8, p1, p1, S; // 2
  as ? A9, A9, p1; // 1
  sa ! A9, p1, p1, S; // 2


  as ? B0, B0, p1;  // 4
  sa ! B0, p1, p1, S; // 5
  as ? B1, B1, p1;  // 4
  sa ! B1, p1, p1, S; // 5
  as ? B2, B2, p1;  // 4
  sa ! B2, p1, p1, S; // 5
  as ? B3, B3, p1;  // 4
  sa ! B3, p1, p1, S; // 5
  as ? B4, B4, p1;  // 4
  sa ! B4, p1, p1, S; // 5
  as ? B5, B5, p1;  // 4
  sa ! B5, p1, p1, S; // 5
  as ? B6, B6, p1;  // 4
  sa ! B6, p1, p1, S; // 5
  as ? B7, B7, p1;  // 4
  sa ! B7, p1, p1, S; // 5
  as ? B8, B8, p1;  // 4
  sa ! B8, p1, p1, S; // 5
  as ? B9, B9, p1;  // 4
  sa ! B9, p1, p1, S; // 5
}
\end{lstlisting}

Next we present the \textsf{Promela} implementation of the implementation of its sequentialisation.

\begin{lstlisting}
mtype    = {A0, B0, A1, B1, A2, B2, A3, B3, A4, B4, A5, B5, A6, B6, A7,
            B7, A8, B8, A9, B9};

short nonceA0, nonceA1, nonceA2, nonceA3, nonceA4, nonceA5, nonceA6,
      nonceA7, nonceA8, nonceA9;
short nonceB0, nonceB1, nonceB2, nonceB3, nonceB4, nonceB5, nonceB6,
      nonceB7, nonceB8, nonceB9;

short g1A0, g1A1, g1A2, g1A3, g1A4, g1A5, g1A6, g1A7, g1A8, g1A9;

mtype otherB0, otherB1, otherB2, otherB3, otherB4, otherB5, otherB6,
      otherB7, otherB8, otherB9;
mtype other2B0, other2B1, other2B2, other2B3, other2B4, other2B5,
      other2B6, other2B7, other2B8, other2B9;
short g2B0, g2B1, g2B2, g2B3, g2B4, g2B5, g2B6, g2B7, g2B8, g2B9;

ltl p0 { always true }

active proctype Model() {

  atomic
  {
    select (nonceA0: 0..1);
    select (nonceA1: 0..1);
    select (nonceA2: 0..1);
    select (nonceA3: 0..1);
    select (nonceA4: 0..1);
    select (nonceA5: 0..1);
    select (nonceA6: 0..1);
    select (nonceA7: 0..1);
    select (nonceA8: 0..1);
    select (nonceA9: 0..1);

    select (nonceB0: 0..1);
    select (nonceB1: 0..1);
    select (nonceB2: 0..1);
    select (nonceB3: 0..1);
    select (nonceB4: 0..1);
    select (nonceB5: 0..1);
    select (nonceB6: 0..1);
    select (nonceB7: 0..1);
    select (nonceB8: 0..1);
    select (nonceB9: 0..1);
  }

  //Communication 1&2, no variable update

  skip; skip; skip; skip; skip; skip; skip; skip; skip; skip;
  skip; skip; skip; skip; skip; skip; skip; skip; skip; skip;

  // Communication 3

  atomic {
    otherB0 = A0; other2B0 = A0; g2B0 = nonceA0;
  }
  atomic {
    otherB1 = A1; other2B1 = A1; g2B1 = nonceA1;
  }
  atomic {
    otherB2 = A2; other2B2 = A2; g2B2 = nonceA2;
  }
  atomic {
    otherB3 = A3; other2B3 = A3; g2B3 = nonceA3;
  }
  atomic {
    otherB4 = A4; other2B4 = A4; g2B4 = nonceA4;
  }
  atomic {
    otherB5 = A5; other2B5 = A5; g2B5 = nonceA5;
  }
  atomic {
    otherB6 = A6; other2B6 = A6; g2B6 = nonceA6;
  }
  atomic {
    otherB7 = A7; other2B7 = A7; g2B7 = nonceA7;
  }
  atomic {
    otherB8 = A8; other2B8 = A8; g2B8 = nonceA8;
  }
  atomic {
    otherB9 = A9; other2B9 = A9; g2B9 = nonceA9;
  }

  //Communication 4&5, no variable update

  skip; skip; skip; skip; skip; skip; skip; skip; skip; skip;
  skip; skip; skip; skip; skip; skip; skip; skip; skip; skip;

  //Communication 6

  g1A0 = nonceB0;
  g1A1 = nonceB1;
  g1A2 = nonceB2;
  g1A3 = nonceB3;
  g1A4 = nonceB4;
  g1A5 = nonceB5;
  g1A6 = nonceB6;
  g1A7 = nonceB7;
  g1A8 = nonceB8;
  g1A9 = nonceB9;

  //Communication 7

  skip; skip; skip; skip; skip; skip; skip; skip; skip; skip;
}
\end{lstlisting}



\bibliographystyle{splncs04}
\bibliography{MPSTtoSequential}

\begin{thebibliography}{10}
\providecommand{\url}[1]{\texttt{#1}}
\providecommand{\urlprefix}{URL }
\providecommand{\doi}[1]{https://doi.org/#1}

\bibitem{synchMPST}
Bejleri, A., Yoshida, N.: {Synchronous Multiparty Session Types}. {Electronic
  Notes in Theoretical Computer Science}  \textbf{241},  3--33 (2009).
  \doi{10.1016/j.entcs.2009.06.002}

\bibitem{BettiniAtall08}
Bettini, L., Coppo, M., D~Antoni, L., De~Luca, M., Dezani-Ciancaglini, M.,
  Yoshida, N.: {Global Progress in Dynamically Interleaved Multiparty
  Sessions}. In: Proceedings of CONCUR. LNCS, vol.~5201, pp. 418--433 (2008).
  \doi{10.1007/978-3-540-85361-9\_33}

\bibitem{bocchiChenDemangeonHondaYoshida13}
Bocchi, L., Chen, T.C., Demangeon, R., Honda, K., Yoshida, N.: {Monitoring
  networks through multiparty session types}. In: Proceedings of FORTE. pp.
  50--65. No.~7892 in LNCS (2013). \doi{10.1007/978-3-642-38592-6\_5}

\bibitem{DemangeonHonda12}
Demangeon, R., Honda, K.: {Nested Protocols in Session Types}. In: Proceedings
  of CONCUR. LNCS, vol.~7454, pp. 272--286 (2012).
  \doi{10.1007/978-3-642-32940-1\_20}

\bibitem{gorla10}
Gorla, D.: {Towards a Unified Approach to Encodability and Separation Results
  for Process Calculi}. Information and Computation  \textbf{208}(9),
  1031--1053 (2010). \doi{10.1016/j.ic.2010.05.002}

\bibitem{spin91}
Holzmann, G.J.: {Design and Validation of Computer Protocols}. Prentice Hall
  (1991)

\bibitem{spin}
Holzmann, G.J.: {The model checker SPIN}. IEEE Transactions on software
  engineering  \textbf{23}(5),  279--295 (1997). \doi{10.1109/32.588521}

\bibitem{hondaYoshidaCarbone08}
Honda, K., Yoshida, N., Carbone, M.: {Multiparty Asynchronous Session Types}.
  In: Proceedings of POPL. vol.~43, pp. 273--284. ACM (2008).
  \doi{10.1145/1328438.1328472}

\bibitem{hondaYoshidaCarbone16}
Honda, K., Yoshida, N., Carbone, M.: {Multiparty Asynchronous Session Types}.
  Journal of the ACM (JACM)  \textbf{63}(1) (2016). \doi{10.1145/2827695}

\bibitem{milnerParrowWalker92}
Milner, R., Parrow, J., Walker, D.: {A Calculus of Mobile Processes}.
  Information and Computation  \textbf{100}(1),  1--77 (1992).
  \doi{10.1016/0890-5401(92)90008-4}

\bibitem{parrowCoupled92}
Parrow, J., Sj{\"o}din, P.: {Multiway synchronization verified with coupled
  simulation}. In: Proceedings of CONCUR. pp. 518--533. No.~630 in LNCS (1992).
  \doi{10.1007/BFb0084813}

\bibitem{petersGlabbeek15}
Peters, K., van Glabbeek, R.: {Analysing and Comparing Encodability Criteria}.
  In: Proceedings of EXPRESS/SOS. EPTCS, vol.~190, pp. 46--60 (2015).
  \doi{10.4204/EPTCS.190.4}

\bibitem{petersWagnerNestmann19}
Peters, K., Wagner, C., Nestmann, U.: {Taming Concurrency for Verification
  Using Multiparty Session Types}. In: Proceedings of ICALP (2019), to appear.

\bibitem{yoshida10}
Yoshida, N., Deni{\'e}lou, P.M., Bejleri, A., Hu, R.: {Parameterised Multiparty
  Session Types}. In: Proceedings of FoSSaCS. LNCS, vol.~6014, pp. 128--145
  (2010). \doi{10.1007/978-3-642-12032-9\_10}

\end{thebibliography}

\end{document}